\documentclass[orivec]{llncs}
\usepackage{stmaryrd,bcprulesmhfix,tensor}
\bcprulessavespacetrue
\suppressrulenamesfalse
\usepackage{mathpartir}
\usepackage{upgreek}
\usepackage{proof}
\usepackage[all]{xy}
\usepackage{graphicx}
\usepackage{amsmath}
\usepackage{lineno}
\usepackage{xypic}
\usepackage{graphicx}
\usepackage{wrapfig}
\usepackage{url,amssymb}
\renewcommand{\phi}{\varphi}
\usepackage{xspace}
\newif\iffull\fullfalse
\newif\ifapp\apptrue
\RequirePackage{txfonts}
\usepackage{times}

\iffull
\else

\makeatletter
\renewcommand{\section}{\@startsection{section}{1}{\z@}{-10\p@ \@plus -4\p@ \@minus -4\p@}{5\p@ \@plus 4\p@ \@minus 4\p@}{\normalfont\bfseries\boldmath\rightskip=\z@ \@plus 8em\pretolerance=10000 }}

  \renewcommand\subsection{\@startsection{subsection}{2}{\z@}{-2\p@ \@plus -4\p@ \@minus -4\p@}{-0.5em \@plus -0.22em \@minus -0.1em}{\normalfont\normalsize\bfseries}}
  \makeatother

\renewcommand{\paragraph}[1]{\noindent {\bf #1}}
\fi

 \usepackage[font={small}]{caption, subfig}

\usepackage{paralist}

\newcommand{\keywd}[1]{\mathtt{#1}}

\newcommand{\myread}[1]{!{#1}}
\newcommand{\wcpo}[1]{$\omega$-cpo}
\newcommand{\wcpos}[1]{$\omega$-cpos}
\newcommand{\myref}[1]{\keywd{ref}(#1)}

\newcommand{\sq}[4]{\tensor*[^{#1}_{#2}]{\Diamond}{^{#3}_{#4}}}
\newcommand{\sqsol}{\Diamond}
\newcommand{\intt}[1]{\textit{int}(#1)}
\newcommand{\effto}[1]{\stackrel{#1}{\to}}

\newcommand{\reff}[1]{\textit{loc}(#1)}

\newcommand{\funn}[1]{\textit{fun}(#1)}

\newcommand{\partfun}{\rightharpoondown}

\newcommand{\mtrue}{\keywd{true}}
\newcommand{\mfalse}{\keywd{false}}

\newcommand{\und}[1]{\underline{#1}}

\newcommand{\inttype}{\keywd{int}}

\newcommand{\unittype}{\keywd{unit}}
\newcommand{\unitval}{\keywd{()}}

\newcommand{\labs}{\mathbb{L}}
\newcommand{\dom}[1]{\mathrm{dom}({#1})}
\newcommand{\cod}[1]{\mathrm{cod}({#1})}

\newcommand{\pause}{\vspace{1.5ex}}
\newcommand{\gap}{\quad\quad}

\newcommand{\ba}{\begin{array}}
\newcommand{\ea}{\end{array}}

\newcommand{\squelch}[1]{}

\newcommand{\vfix}[3]{\keywd{rec}\:{#1}\:{#2} = {#3}}
\newcommand{\letin}[2]{\keywd{let}\:{#1}\!\Leftarrow\!{#2}\:\keywd{in}\:}

\newcommand{\opletrec}[3]{\keywd{let\ rec}}

\newcommand{\assign}[2]{{#1}:={#2}}

\newcommand{\rdsin}[1]{\mathrm{rds}({#1})}
\newcommand{\wrsin}[1]{\mathrm{wrs}({#1})}
\newcommand{\alsin}[1]{\mathrm{als}({#1})}

\newcommand{\myif}[3]{\keywd{if}\ #1\ \keywd{then}\ #2\
  \keywd{else}\ #3}

\newcommand{\eff}{\varepsilon}
\newcommand{\reads}{\mathrm{rds}}
\newcommand{\writes}{\mathrm{wrs}}
\newcommand{\allocs}{\mathrm{als}}

\newcommand{\fullonly}[1]{}

\newlength{\lruleZeroname}
\newcommand{\regs}[1]{\mathrm{regs}({#1})}

               {\vspace{1em}{\bf Example.}
               }%
               {
               \qed%
               \vspace{1em}%
               }

\newcommand{\regid}{\ensuremath{\mathsf{r}}}
\newcommand{\Regids}{\ensuremath{\mathit{Regs}}}
\newcommand{\Locs}{\ensuremath{\labs}}
\newcommand{\new}{\ensuremath{\mathit{new}}}
\newcommand{\Stores}{\ensuremath{\mathbb{H}}}
\newcommand{\sem}[1]{\ensuremath{\llbracket {#1} \rrbracket}}
\newcommand{\semV}[1]{\ensuremath{\llceil {#1} \rrceil
}}

\newcommand{\nwrs}{{\ensuremath{\mathit{nwrs}}}}

\newcommand{\elEffs}{\mathcal{E}}
\newcommand{\aEff}[1]{\ensuremath{\mathit{al}_{#1}}}
\newcommand{\rEff}[1]{\ensuremath{\mathit{rd}_{#1}}}
\newcommand{\wEff}[1]{\ensuremath{\mathit{wr}_{#1}}}

\newcounter{Examplecount}
\renewenvironment{example}
{
\stepcounter{theorem} {\bf \noindent Example
\arabic{section}.\arabic{theorem} }}
{\qed}
\renewenvironment{proof}{\vspace{-1mm} \noindent {\bf Proof}\quad}{\qed}

\newcommand{\loc}{\mathfrak{l}}
\newcommand{\cloc}{\ensuremath{\mathsf{l}}\xspace}
\newcommand\w{\ensuremath{\mathsf{w}}\xspace}
\newcommand\q{\ensuremath{\mathsf{q}}\xspace}

\newcommand{\world}{\ensuremath{\mathbf{W}}\xspace}

\newcommand\heap{\ensuremath{\mathsf{h}}\xspace}

\newcommand{\Values}{\mathbb{V}}
\newcommand{\Comps}{\mathbb{C}}
\newcommand{\Astores}{\mathfrak{S}}

\newcommand\val{\ensuremath{\mathsf{v}}\xspace}
\newcommand\vval{\ensuremath{v}\xspace}

\newcommand\cval{\ensuremath{\mathsf{c}}\xspace}
\newcommand\ccval{\ensuremath{c}\xspace}

\newcommand\Std{\ensuremath{\mathit{Std}}\xspace}
\newcommand\Rscr{\ensuremath{\mathcal{R}}\xspace}

\newcommand{\ety}[2]{{#1}\mathrel{\&}{#2}}
\newcommand{\valty}[1]{#1}

\title{Abstract Effects and Proof-Relevant Logical Relations}

\author{Nick Benton\iffull\inst{1}\fi \and Martin Hofmann\iffull\inst{2}\fi  \and Vivek
Nigam\iffull\inst{3}\fi}
\institute{Microsoft Research, LMU Munich and Federal University of Para\'iba}
\squelch{
\institute{\iffull
Microsoft Research, UK, \email{nick@microsoft.com}
\and\fi
LMU, Munich, Germany, 
\email{hofmann@ifi.lmu.de}
\iffull\and
Federal University of Para\'iba, Brazil, \email{vivek.nigam@gmail.com}\fi
}
}

\begin{document}
\maketitle
\begin{abstract}
  We introduce a novel variant of logical relations that maps types 
 not merely to partial equivalence relations on values, as is
  commonly done, but rather to a proof-relevant generalisation
  thereof, namely setoids. 
\squelch{
A setoid is like a category all of whose
  morphisms are isomorphisms (a groupoid) with the exception that no
  equations between these morphisms are required to hold. 
}
  The objects of a setoid establish that values inhabit semantic
  types, whilst its morphisms
  are understood as proofs of semantic equivalence.

  The transition to proof-relevance solves two well-known
  problems caused by the use of existential quantification over future
  worlds in traditional Kripke logical relations: failure of
  admissibility, and spurious functional dependencies.

  We illustrate the novel format with two applications: a direct-style
  validation of Pitts and Stark's equivalences for ``new'' and a
  denotational semantics for a region-based effect system that
  supports type abstraction in the sense that only externally visible
  effects need to be tracked; non-observable internal modifications,
  such as the reorganisation of a search tree or lazy initialisation,
  can count as `pure' or `read only'. This `fictional purity' allows
  clients of a module soundly to validate more effect-based program
  equivalences than would be possible with traditional effect systems.
\end{abstract}
\squelch{
\category{D.3.3}{Programming Languages}{Language Constructs and Features -- Dynamic storage management}
\category{F.3.2}{Logic and Meanings of Programs}{Semantics of Programming Languages -- Denotational semantics, Program analysis}
\category{F.3.2}{Logic and Meanings of Programs}{Studies of Program Constructs -- Type structure}
\terms
Languages, Theory
\keywords
Type and effect systems, region analysis, logical relations, parametricity, program transformation
}

\iffull \else \vspace{-4mm} \fi
\section{Introduction}
\label{sec:intro}
The last decade has witnessed significant progress in modelling and
reasoning about the tricky combination of effects and higher-order
language features (first-class functions, modules, classes).  The
object of study may be ML-, Java-, or assembly-like, but the common
source of trickiness is the way effectful operations may be
\emph{partially} encapsulated behind higher-order
abstractions. Problems in semantics and verification of effectful
languages are often addressed using a range of common techniques that
includes separation and Kripke logical relations (KLRs). The
particular problem motivating the development of the proof-relevant
form of KLR introduced here is that of giving a semantics to effect
systems that accounts for partial encapsulation, though the general
construction is more broadly applicable. As we will see, direct
semantic reasoning in our model (as opposed to generic reasoning based
on refined types) also allows many of the trickiest known equivalences
concerning encapsulated store to be proved.

\squelch{One is to devise models and reasoning principles for establishing
contextual (in)equivalences
\cite{DBLP:conf/mfcs/PittsS93,DBLP:journals/corr/abs-1103-0510}. A
second is to establish equivalence between high-level and low-level
code fragments, e.g. for compiler correctness
\cite{DBLP:conf/icfp/BentonH09,DBLP:conf/popl/HurD11}. A third is to
define Hoare-style logics for showing programs satisfy 
assertions \cite{DBLP:journals/corr/abs-1109-3031}. A fourth, which
we address here, is to study type systems and analyses that
can characterise particular \emph{classes} of behaviour (such as
purity) and be used to justify equivalences more generically.
}

Effect systems \cite{DBLP:conf/lfp/GiffordL86} refine conventional
types by tracking upper bounds on the side-effects of expressions.  A
series of papers, by ourselves and others
\cite{DBLP:conf/popl/KammarP12,DBLP:conf/ppdp/BentonKBH09,DBLP:conf/ppdp/BentonKBH07,DBLP:conf/aplas/BentonKHB06,DBLP:conf/icfp/ThamsborgB11},
have explored the semantics of effect systems for mutable state,
addressing not merely the correctness of analyses, but also the
soundness of effect-dependent optimizations and refactorings. An
example is the commutation of stateful computations $M$ and $N$,
subject to the condition that the sets of storage locations
potentially written by $M$ and $N$ are disjoint, and that neither
potentially reads a location that the other writes.  \squelch{If one
  gives expressions refined monadic types of the form $T_\eff \tau$,
  meaning `computations returning results of type $\tau$ with
  side-effects bounded by $\eff$', for a suitable choice of effect
  annotations $\eff$, then effect-dependent equivalences can be
  formalized in terms of conditions on annotations. For example
  \infrule {\Gamma\vdash M:T_\eff\tau \andalso \Gamma\vdash
    N:T_{\eff'}\tau' \andalso \Gamma,x:\tau,y:\tau' \vdash
    P:T_{\eff''}\tau''\\ \wrsin{\eff}\cap
    \wrsin{\eff'}=\rdsin{\eff}\cap \wrsin{\eff'} = \rdsin{\eff'}\cap
    \wrsin{\eff}=\emptyset} {\Gamma
    \vdash \begin{array}{l}\letin{x}{M}{\letin{y}{N}{P}}\\ =
      \letin{y}{N}{\letin{x}{M}{P}}\end{array} :
    T_{\eff\cup\eff'\cup\eff''}\tau''}

\noindent
where $\rdsin{\eff}$ and $\wrsin{\eff}$ are the sets of possibly-read
and possibly-written variables (or regions) in $\eff$.} Our primary
interest is not syntactic rules for type assignment, but rather
semantic interpretations of effect-refined types that can justify such
equivalences.  Types provide a common interface language that can be
used in modular reasoning about rewrites; types can be assigned to
particular terms by a mixture of more or less sophisticated inference
systems, or by deeper semantic reasoning. \squelch{Clearly, such a
  separation of concerns requires that the interpretation of types be
  independent of any inference system.}

A key notion in compositional reasoning about state is that of
\emph{separation}: invariants depending upon mutually disjoint parts
of the store. Intuitively, if each function with direct access to a
part preserves the corresponding invariant, then all the invariants
will be preserved by any composition of functions. Disjointness is
naively understood in terms of sets of locations.  A memory
allocator, for example, guarantees that its own private
datastructures, memory belonging to clients, and any
freshly-allocated block inhabit mutually disjoint sets of
locations. Since the introduction of fractional permissions,
separation logics often go beyond this simple model, introducing
resources that are combined with a separating conjunction, but which
are not literally interpreted as predicates on disjoint
locations. Research on `domain-specific'
\cite{DBLP:conf/tldi/KrishnaswamiBA10}, `fictional'
\cite{DBLP:conf/vstte/Dinsdale-YoungGW10,DBLP:conf/esop/JensenB12},
`subjective' \cite{leywild:scsl}, or `superficial'
\cite{krishnaswami:superficially} separation \squelch{logics and type theories}
aims to let custom notions of separable resource be used
and combined modularly. This paper presents a semantics for effect
systems supporting fictional, or `abstract', notions of both effects
and separation.

We previously interpreted effect-refined types for stateful
computations as binary relations, defined via preservation of
particular sets of store relations. This already provides some
abstraction. For example, a function that reads a reference, but whose
result is independent of the value read can soundly be counted as pure
(contrasting with models that instrument the concrete semantics). Our
models also validated the masking rule, allowing certain
non-observable effects not to appear in annotations. But here we go
further, generalizing the interpretation of regions to partial
equivalence relations (PERs). This allows, for example, a lookup
function for a set ADT to be assigned a read-but-not-write effect,
even if the concrete implementation involves non-observable writes to
rebalance an internal datastructure. Roughly, there is a PER that
relates two heaps iff they contain well-formed datastructures
representing the same mathematical set, and the ADT operations respect
this PER: looking up equal values in related heaps yields equal
booleans, adding equal values in related heaps yields new related
heaps, and so on. A mutating operation need only be annotated with a
write effect if the updated heap is potentially in a different
equivalence class from the original one. In fact, we further improve
previous treatments of write effects, via a `guarantee' condition that
explicitly captures allowable local updates. Surprisingly, this allows
the update and remove operations for our set ADT to be flagged with
\emph{just} a write effect, despite the fact that the final state of
the set depends on the initial one, exploiting the idempotence of the
updates and validating many more useful program transformations.

Moving to PERs also allows us to revisit the notion of separation,
permitting distinct abstract locations, or regions, to refer to PERs
whose footprints overlap, albeit non-observably, in memory. A module
may, for example, implement two distinct logical references using a
single physical location containing a coding (e.g. $2^i3^j$) of a pair
$(i,j)$ of integers. Or a resource allocator can keep logically
separated tokens tracking each allocated resource, acting as
permissions for deallocation, in a shared datastructure such as a
bitmap or linked list (a well-known problem in modular separation
\cite{krishnaswami:superficially}). The innovation here is a notion of independence of PERs,
capturing the situation where intersection of PERs yields a cartesian
product of quotients of the heap.

\begin{wrapfigure}{r}{.5\textwidth}
\vspace{-3mm}
\[\begin{array}{@{}l}
(T_\varepsilon Q)_\w = QPER(\{(f,f')\mid \heap,\heap'\models\w\Rightarrow\\
\ \forall R\in{\cal R}_\varepsilon(\w). \heap R \heap' \Rightarrow \heap_1 R \heap'_1 \wedge\\
\ \exists \w_1. (\w_1(r)\not= \emptyset\Rightarrow r\in \mathrm{als}(\varepsilon))\wedge  \heap_1,\heap_1'\models \w\otimes\w_1 \wedge\\
\ \heap_1 \sim_{\w_1} \heap'_1 \wedge (v,v')\in Q_{\w\otimes\w_1}\\
\ \mbox{where $(\heap_1,v)=f \heap$ and $(\heap'_1,v')= f' \heap'$\})}
\end{array}
\vspace{-2mm}
\]
\vspace{-2mm}
\caption{Earlier Kripke logical relation, extract}
\label{fig:oldklr}
\vspace{-9mm}
\end{wrapfigure}
The ideas sketched above are intuitively rather compelling, but
formally integrating them into the form of KLR we had previously used
for effect systems turns out to be remarkably
hard. Figure~\ref{fig:oldklr} shows a (tweaked) extract from an earlier
paper \cite{DBLP:conf/ppdp/BentonKBH07}. Here a world $\w$ is just a finite partial bijection between
locations, with region-coloured links; $\heap,\heap'\models \w$
simply means that for each link $(\cloc,\cloc')\in \w $,
$\cloc\in\dom{\heap}$
and $\cloc'\in\dom{\heap'}$.
Two
computations $f,f':\Stores \rightharpoondown
\Stores\times\Values$, where $\Stores,\Values$ are
sets of heaps and values, respectively, are in the relation
$(T_\eff Q)_\w$, where $\eff$ is an effect and the relation $Q$ interprets a
result type, if they preserve all heap relations $R$ in a set
depending on $\varepsilon$ and $\w$, and there exists \emph{some}
disjoint world extension $\w_1$ such that the new heaps are equal on
the domain of $\w_1$, and the result values are $Q$-related at the
extended world $\w\otimes \w_1$. 


The problematic part
is the existential quantification over world extensions -- the
$\exists \w_1$ on the third line -- allowing for the computations
to allocate fresh locations.
This pattern of
quantification occurs in many accounts of generativity, but
the dependence of $\w_1$ on both $\heap$ and $\heap'$ creates serious problems
if one generalizes from bijections to PERs
and tries to prove equivalences.
Roughly, one has to consider varying the initial heap in
which one computation, say $f'$, is started; the
existential then produces a \emph{different} extension $\w_2$
that is
not at all related, even on the side of $f$ where the
heap stays the same, to the $\w_1$ with which one started. The
case of bijections, where $\heap_1$ depends only on
$\heap$ (not on $\heap'$), allows one to deduce sufficient information
about the domain of $\w_1$ from the clause $\heap_1,\heap_1'\models \w\otimes\w_1$,
but this breaks down in the more abstract setting.

To fix this problem, we here
take the rather novel step of replacing the existential
quantifier in the logical relation by appropriate Skolem functions,
explicitly enforcing the correct dependencies. In the language of
type theory, this amounts to replacing an existential  with
a $\Sigma$-type. A statement like $(f,f')\in T_\eff\sem{A}$ is no
longer just a proposition, but we rather have a ``set of
proofs'' $T_\eff\sem{A}(f,f')$ which in particular contains
the aforementioned Skolem functions. We use an explicit version of the 
exact-completion \cite{DBLP:conf/mfps/CarboniFS87,DBLP:conf/lics/BirkedalCRS98} akin to and motivated by ``setoid'' or groupoid interpretations of type theory 
\cite{DBLP:conf/lics/HofmannS94,DBLP:journals/jfp/BartheCP03,DBLP:conf/cpp/Voevodsky11} to make these ideas both rigorous and more
general.

Passing from relations to proof-relevant setoids also solves other
problems. Existential quantification fails to preserve admissibility
of relations, needed to deal with general recursion, and also fails to
preserve `PERness'. The `$QPER(\cdot)$' operation in
Figure~\ref{fig:oldklr} explicitly applies an admissible and (variant)
PER closure operation; this works technically, but is very awkward to
use. We do not need such a closure here. Step indexing
\cite{DBLP:conf/esop/Ahmed06,DBLP:conf/icfp/ThamsborgB11} and the use
of continuations \cite{DBLP:conf/mfcs/PittsS93} can also deal with
admissibility.  However, step-indexing is inherently operational,
whilst continuations lose sufficient abstraction to break some program
equivalences, including commuting computations. Our third way, using
setoids, is pleasantly direct. Finally, allocation
effects are handled differently from reading and writing by the
relation in Figure~\ref{fig:oldklr}, being wired into the
quantification rather than treated more abstractly by relation
preservation. Our setoid-based formulation uses uniform machinery to
treat all effects.

We start by reviewing some preliminary definitions on syntax and
semantics of programs in Section~\ref{sec:semantics}.
Section~\ref{sec:setoids} introduces setoids, which is the setting in
which we specify in Section~\ref{test} the typed semantics and
introduce the notion of abstract effects. In
Section~\ref{sec:logical-relations} we describe proof-relevant logical
relations, prove the fundamental theorem and define observational
equivalence. Section~\ref{sec:application} demonstrates a number of
program equivalences that can be shown by using proof-revelant logical
relations. We conclude and discuss future work in
Section~\ref{sec:conclusions}.

\ifapp
\paragraph{Note:} We have elided many proofs, details of constructions
and examples. This longer version of the paper includes some of this
material in an appendix.
\fi

  \squelch{A long version, possibly formalised
  will contain these details which are not hard to reconstruct
  oneself, but tedious and space-consuming to typeset.}

\section{Syntax and Semantics}
\label{sec:semantics}
We will interpret effect-refined types over a somewhat generic,
untyped denotational model for stateful computations in the category
of predomains ($\omega$-cpos).  We also introduce a meta-language
\cite{moggi1991notions}, providing concrete syntax for functions in
the model. We omit the standard details of interpreting CBV
programming languages via such a metalanguage, or proofs of adequacy,
relating the operationally induced observational (in)equivalence to
(in)equality in the model.


\paragraph{Denotational model} 
\iffull
A \emph{predomain} is an $\omega$-cpo, i.e. a partial order with suprema
of ascending chains. A \emph{domain} is a predomain with a least element,
$\bot$. \squelch{We use predomains and continuous functions, rather than sets
and functions, so as to be able to interpret recursive definitions.}
Recall that $f:A\rightarrow A'$ is \emph{continuous} if it is monotone
$x\leq
y \Rightarrow f(x)\leq f(y)$ and preserves suprema of ascending
chains, i.e., $f(\sup_i x_i)=sup_i f(x_i)$. Any set is a predomain with the
discrete
order. If $X$ is a set and $A$ a predomain then
any $f:X\rightarrow A$ is continuous. A subset $U$ of a
predomain $A$ is \emph{admissible} if whenever $(a_i)_i$ is an ascending
chain in $A$ such that $a_i\in U$ for all $i$, then $\sup_i a_i\in U$,
too. If $f:X\times A\rightarrow A$ is continuous and $A$ is a domain
then one defines $f^\dagger(x)=\sup_if_x^i(\bot)$ with
$f_x(a)=f(x,a)$. One has, $f(x,f^\dagger(x))=f^\dagger(x)$ and if
$U\subseteq A$ is admissible and $f:X\times U\rightarrow U$ then
$f^\dagger:X\rightarrow U$, too. We denote a partial (continuous)
function from set (predomain) $A$ to set (predomain) $B$ by $f:A
\partfun B$.  
\fi
We assume predomains $\Values$ and $\Stores$ modelling values and heaps,
respectively. As much of the metatheory does not rely on the finer
details of how these predomains are defined, we axiomatise the
properties we use.  Firstly, we assume the existence of a set of
(concrete) locations $\labs$ and for each $\heap\in\Stores$ a finite
set $\dom\heap\subseteq\labs$. We also assume a constant
$\emptyset\in\Stores$, the empty heap.  If $
\heap \in \Stores, \cloc \in \dom{\heap}$, then $\heap(\cloc)
\in \Values$. If $v\in
\Values, \heap \in \Stores, \cloc\in\dom{\heap}$ then $\heap[\cloc{\mapsto}
v]\in\Stores$; finally $\new(\heap,v)$ yields a pair $(\cloc,\heap')$
where $\cloc\in\Locs$ and $\heap'\in\Stores$. These three operations
are continuous, in particular, $\heap\leq\heap'\Rightarrow
\dom\heap\subseteq \dom{\heap'}$ and the following axioms hold:
$\dom{\emptyset}=\emptyset$,
$\dom{\heap[\cloc{\mapsto}v]}=\dom{\heap}$,
$(\heap[\cloc{\mapsto}v])(\cloc') =$ $\textit{if }\cloc=\cloc'\textit{
  then }v\textit{ else } \heap(\cloc')$, and if
$\new(\heap,v)=(\cloc,\heap')$ then
$\dom{\heap'}=\dom{\heap}\cup\{\cloc\}$ and $\cloc\not\in\dom{\heap}$
and $\heap'(\cloc)=v$.  Given $\Values$ this abstract datatype can be
implemented in a number of ways, e.g., as finite maps.
We define the domain of computations $\Comps$ to be 
partial continuous functions from $\Stores$ to $\Stores\times
\Values$, the bottom element being the everywhere undefined function.

We assume that $\Values$ embeds tuples of values, i.e., if
$v_1,\dots,v_n\in\Values$ then $(v_1,\dots,v_n)\in\Values$ and it is
possible to tell whether a value is of that form and in this case to
retrieve the components. We also assume that $\Values$ embeds
continuous functions $f:\Values\rightarrow \Comps$, i.e., if $f$ is
such a function then $\funn f\in\Values$ and, finally, locations are
also values, i.e.\ if $\cloc\in\labs$ then $\reff\cloc\in\Values$ and
one can tell whether a value is a location or a function. A canonical
example of
such a $\Values$ is the least solution to the predomain equation 
with $\Comps=\Stores\partfun\Stores\times \Values$ and
\(
\Values \simeq \intt{\mathbb{Z}} + \funn{\Values\rightarrow\Comps} +
\reff{\labs} + \Values^*.
\)

\paragraph{Syntax}
The syntax of untyped values and computations is:
\iffull
\[
 \begin{array}{lcl}
v & ::= & x \mid () \mid c \mid (v_1,v_2) \mid v.1 \mid v.2 \mid
\vfix{f}{x}{t} \\
t &::=& v \mid \letin{x}{t_1}{t_2}\mid v_1\,v_2 \mid
\myif{v}{t_1}{t_2} \mid
\myread{v}\mid \assign{v_1}{v_2}\mid \myref{v}
 \end{array}
\]
Here, $x$ ranges over variables and $c$ over constant
\else

\(
 \begin{array}{lcl}
v & ::= & x \mid () \mid c \mid (v_1,v_2) \mid v.1 \mid v.2 \mid
\vfix{f}{x}{t} \\
t &::=& v \mid \letin{x}{t_1}{t_2}\mid v_1\,v_2 \mid
\myif{v}{t_1}{t_2} \mid
\myread{v}\mid \assign{v_1}{v_2}\mid \myref{v}
 \end{array}
\)

\noindent
Here, $x$ ranges over variables and $c$ over constant
\fi
symbols, each of which has an associated interpretation
$\semV{c}\in\Values$; these include numerals $\underline{n}$ with
$\semV{\underline{n}} = \intt{n}$, arithmetic operations and so on.
$\vfix{f}{x}{t}$ defines a recursive function with body $e$ and
recursive calls made via $f$; we use $\lambda x.t$ as syntactic sugar in
the
case when $f\not\in fv(t)$. Finally, $\myread{v}$ (reading) returns
the contents of location $v$, $\assign{v_1}{v_2}$ (writing) updates
location $v_1$ with value $v_2$, and $\myref{v}$ (allocating) returns a
fresh location intialised with $v$. The metatheory is simplified by
using ``let-normal form'', in which the only elimination for
computations is let, though we sometimes nest computations as
shorthand for let-expanded versions in examples.

\paragraph{Semantics}
The untyped semantics of values $\semV{v}\in\Values\rightarrow\Values$
and terms $\semV{t}\in\Values\rightarrow\Comps$ are defined by an
entirely standard mutual induction, using least fixed points to
interpret recursive functions, projection from tuples for variables
and so on.
\squelch{
are given by
the recursive clauses in Figure~\ref{seme}; note the overloading of
semantic brackets for constants, values and computations. 
The notation
$\eta(x)$ stands for the $i$-th projection from $\eta\in\Values$
if $x$ is $x_i$ and $\eta[x{\mapsto}v]$ (functionally) updates the
$i$-th slot in $\eta$ when $x=x_i$.
}

\paragraph{Types}
Types are given by the grammar:
\iffull
\[
\tau ::= \unittype \mid \inttype\mid A\mid\tau_1\times\tau_2\mid 
\tau_1\effto\eff \tau_2
\]
where $A$ ranges over semantically defined basic types (see
\else
$
\tau ::= \unittype \mid \inttype\mid A\mid\tau_1\times\tau_2\mid 
\tau_1\effto\eff \tau_2,
$
where $A$ ranges over semantically defined basic types (see
\fi
Def.~\ref{defte}). These contain reference types possibly
annotated with regions and abstract types like lists, sets, and even
objects, again possibly refined by regions. The metavariable $\eff$
represents an \emph{effect}, that is a subset of some fixed set of
elementary effects about which we say more later. The core typing
rules for values and computations are shown in
Figure~\ref{tyres}. We do not bake in type rules for
constants and effectful operations but, for a given semantic interpretation
of types, we will be able to justify adding further rules for these
primitives and, more importantly, for more complex expressions
involving them. (The rules given here incorporate
subeffecting; we expect our semantics to extend
to more general subtyping.)


\paragraph{Equations} Figure~\ref{eqth} outlines a core equational theory for
the metalanguage.  The full theory includes congruence rules for all
constructs (like that given for \texttt{rec}), all the usual beta and
eta laws and commuting conversions for conditionals as well as for
\texttt{let}.
We give
a semantic interpretation of typed equality judgements which is sound
for observational equivalence. 
As with typings, further equations involving effectful computations may be
justified semantically in a particular model and added to the theory.
The core theory then allows one to
deduce new semantic equalities from already proven ones. 
The equations are typed: a derivation
$\mathcal{D}$ of $\Gamma\vdash t=t':\ety{\tau}{\eff}$ is canonically
associated with typing derivations $\mathcal{D}.1$ and $\mathcal{D}.2$
of $\Gamma\vdash t:\ety{\tau}{\eff}$ and $\Gamma\vdash t':
\ety{\tau}{\eff}$,
respectively (but note we can semantically justify extending the
type rules). The interpretation of
$\mathcal{D}$ will be a proof object certifying that the 
interpretations of $\mathcal{D}.1$ and $\mathcal{D}.2$ are
semantically equal which then implies (Theorem~\ref{obseq}) typed
observational equivalence of $t$ and $t'$.

\begin{figure*}[tph]
\vspace{-3mm}
\[
\infer{\Gamma\vdash \underline{n} : \valty{\inttype}}{}
\quad 
\infer{\Gamma,x:\tau\vdash x : \valty{\tau}}{}
\quad
\infer{\Gamma\vdash v:\ety{\tau}{\emptyset}}{\Gamma\vdash v:\valty{\tau}}
\quad
\infer{\Gamma\vdash e:\ety{\tau}{\eff_2}}{\Gamma\vdash
e:\ety{\tau}{\eff_1} & \eff_1
\subseteq  \eff_2}
\quad 
\infer{\Gamma\vdash v.i:\valty{\tau_i}}{\Gamma\vdash
v:\valty{\tau_1\times\tau_2}}
\]
\[
\infer{\Gamma\vdash v_1\ v_2 : \ety{\tau_2}{\eff}}
{\Gamma\vdash v_1:\valty{\tau_1\effto{\eff}\tau_2} & \Gamma\vdash
v_2:\valty{\tau_1}}
\quad
\infer{\Gamma \vdash \unitval : \unittype}{}
\qquad
 \infer{\Gamma\vdash\myif{v}{e_1}{e_2} : \ety{\tau}{\eff}}
{\Gamma\vdash v:\valty{\inttype} & 
    \Gamma\vdash e_1:\ety{\tau}{\eff} & 
    \Gamma\vdash e_2:\ety{\tau}{\eff}}
\]
\[
\infer{\Gamma\vdash
(v_1,v_2):\valty{\tau_1\times\tau_2}}
{\Gamma\vdash v_1:\valty{\tau_1} & 
\Gamma\vdash v_2:\valty{\tau_2}}
\quad
\infer{\Gamma\vdash \letin{x}{e_1}{e_2}:\ety{\tau_2}{\eff}}
{\Gamma\vdash e_1:\ety{\tau_1}{\eff} & 
  \Gamma,x{:}\tau_1\vdash e_2:\ety{\tau_2}{\eff}}
\qquad 
\infer{\Gamma\vdash \vfix{f}{x}{e} : \valty{\tau_1\effto{\eff} \tau_2}}
{\Gamma,f{:}\tau_1\effto\eff \tau_2,x{:}\tau_1\vdash
e:\ety{\tau_2}{\eff}}
\]
\caption{Core rules for effect typing\label{tyres}\label{efte}}
\vspace{-5mm}
\end{figure*}
\begin{figure*}[tph]
\vspace{-3mm}
\[
\infer{\Gamma\vdash t=t:\ety{\tau}{\eff}}
{\Gamma\vdash t:\ety{\tau}{\eff}} 
\quad
\infer{\Gamma\vdash t'=t:\ety{\tau}{\eff}}
{\Gamma\vdash t=t':\ety{\tau}{\eff}} 
\quad 
\infer{\Gamma\vdash t=t'':\ety{\tau}{\eff}}
{\Gamma\vdash t=t':\ety{\tau}{\eff} & \Gamma\vdash
t'=t'':\ety{\tau}{\eff}}
\quad 
\infer{\Gamma\vdash v=v':\ety{\tau}{\emptyset}}
{\Gamma\vdash v=v':\valty{\tau}}
\]
\iffull \vspace{2pt} \fi
\[
\infer{\Gamma\vdash (v_1,v_2).i = v_i :\valty{\tau_i}}
{\Gamma\vdash v_1:\valty{\tau_1} & \Gamma\vdash
v_2:\valty{\tau_2}}
\quad
\infer{\Gamma\vdash
(\vfix{f}{x}{t})=(\vfix{f}{x}{t'}):\valty{\tau_1\effto\eff
\tau_2}}
{\Gamma,f:\tau_1\effto\eff \tau_2, x{:}\tau_1\vdash
t=t':\ety{\tau_2}{\eff}}
\quad
\infer{\Gamma\vdash v=(v.1,v.2):\valty{\tau_1\times\tau_2}}
{\Gamma\vdash v:\valty{\tau_1\times\tau_2}}
\]
\iffull \vspace{2pt} \fi
\[ 
\infer{\Gamma\vdash\letin{x}{v}{t} = t[v/x] : \ety{\tau_2}{\eff}}
{\Gamma\vdash v:\ety{\tau_1}{\eff} & \Gamma,x:\tau_1\vdash
t:\ety{\tau_2}{\eff}}
\quad
\infer{\Gamma\vdash (\vfix{f}{x}{t})\,v = t[v/x, (\vfix{f}{x}{t})/f]:
\ety{\tau_2}{\eff}}
{\Gamma,f:\tau_1\effto\eff \tau_2, x{:}\tau_1\vdash
t:\ety{\tau_2}{\eff} & \Gamma\vdash v:\tau_1}
\]
\iffull \vspace{2pt} \fi
\[ 
\infer{\Gamma\vdash\letin{x}{(\letin{y}{t_1}{t_2})}{t_3} =
\letin{y}{t_1}{\letin{x}{t_2}{t_3}} : \ety{\tau_3}{\eff}}
{\Gamma\vdash t_1:\ety{\tau_1}{\eff} & \Gamma\vdash
t_2:\ety{\tau_2}{\eff} & \Gamma,x:\tau_2,y:\tau_1\vdash
t_3:\ety{\tau_3}{\eff}}
\]
\caption{Basic equational theory (extract)\label{eqth}}
\iffull
\else
\vspace{-2mm}
\fi
\end{figure*}
\subsection{Some example programs}\label{examples}

~\newline

\paragraph{Dummy allocation}
Define $\textit{dummy}$ as 
$
\sem{\lambda f.\lambda x.\letin{d}{\myref{0}}{f\ x}},
$ 
so $\textit{dummy}(f)$ behaves like $f$ but makes an
allocation whose result is discarded.
We will be able to show that $\textit{dummy}(f)$ displays no more
abstract effects than $f$, so that whatever program transformation
$f$ can participate in, $\textit{dummy}(f)$ can as well.

\paragraph{Memoisation}
Let \textit{memo} be the memoizing functional
\iffull
\[
\begin{array}{l}\llbracket\lambda f.
\letin{x}{\myref{\underline{0}}}{
\letin{y}{\myref{f\ \underline{0}}}{}}\\
\quad \lambda a.\myif{\textit{eq}\ a\ \myread{x}}{\myread{y}}
{\letin{r}{f\ a}{\assign{x}{a};\assign{y}{r};r}}
\rrbracket
\end{array}\]
where $t_1;t_2=\letin{\_}{t_1}{t_2}$ is sequential composition and
$\textit{eq}$ is an integer equality constant.
\else

\(
\begin{array}{l}\llbracket\lambda f.
\letin{x}{\myref{\underline{0}}}{
\letin{y}{\myref{f\ \underline{0}}}{}}\\
\quad \lambda a.\myif{\textit{eq}\ a\ \myread{x}}{\myread{y}}
{\letin{r}{f\ a}{\assign{x}{a};\assign{y}{r};r}}
\rrbracket
\end{array}
\)

\noindent
where $t_1;t_2=\letin{\_}{t_1}{t_2}$ is sequential composition and
$\textit{eq}$ is an integer equality constant.
\fi
We can justify the typing
$\textit{memo}:(\textit{int}\effto{\emptyset}\textit{int})\effto{\emptyset}(\textit{int}\effto{\emptyset}\textit{int})$, saying that if $f$ is observationally pure, $\textit{memo}\, f$, is too, and so can 
participate in any program equivalence relying
on purity. This was not justified by our previous
model~\cite{DBLP:conf/ppdp/BentonKBH07}.

\paragraph{Set factory}
The next, more complicated, example is a program that can 
create and manipulate sets implemented as linked lists.

If $\cloc\in\labs$ and $\heap\in\Stores$ and $U$ is a finite set of
integers and $P$ is a finite subset of $\labs$ define
$S(\cloc,\heap,U,P)$ to mean that in $\heap$ location $\cloc$ points to
a linked list of integer values occupying at most the locations in $P$
(the ``footprint'') and so that the set of these integer values is
$U$. So, for example, if $\heap(\cloc)=\reff {\cloc_1}$ and
$\heap(\cloc_1)=(\intt 1,\reff{\cloc_2})$ and $\heap(\cloc_2)=(\intt
1,\intt 0)$ then $S(\cloc,\heap,\{1\}, \{\cloc_1,\cloc_2\})$ holds. 

For each location $\cloc$ define functions $\textit{mem}_\cloc$,
$\textit{add}_\cloc$, $\textit{rem}_\cloc$ so that
$\textit{mem}_\cloc(\intt i)$ checks whether $i$ occurs in the list
pointed to by $\cloc$, returning $\intt 1$ iff yes, and---for the fun
of it---removes all duplicates in that list and relocates some of its
nodes.  Thus, in particular, if $\textit{mem}_\cloc(\intt
i)(\heap)=(\heap_1,v)$ then if $S(\cloc,\heap,U,P)$ one has
$S(\cloc,\heap_1,U,P')$ for some $P'$ where $P'\subseteq
P\cup(\dom{\heap_1}\setminus\dom{\heap})$ and $v=\intt 1$ iff $i\in
U$.

The function $\textit{add}_\cloc$ adds its integer argument to the
set, and $\textit{rem}_\cloc$ removes it, each possibly making
``optimizations'' similar to $\textit{mem}_\cloc$.

Now consider a function $\textit{setfactory}$ returning upon each call
a fresh location $\cloc$ and a the tuple of functions
$(\textit{mem}_\cloc,\textit{add}_\cloc,\textit{rem}_\cloc)$.
We will be able to justify the following
 semantic typing for $\textit{setfactory}$: 
\iffull
\[
\textit{setfactory} : \forall\regid. 
\ety{(\textit{int}\effto{\rEff\regid}\textit{int})\times
(\textit{int}\effto{\wEff\regid}\textit{unit})\times
(\textit{int}\effto{\wEff\regid}\textit{unit})} {\aEff\regid}
\]
which  expresses   that  $\textit{setfactory}()$  allocates   in  some
\else

\(
\textit{setfactory} : \forall\regid. 
\ety{(\textit{int}\effto{\rEff\regid}\textit{int})\times
(\textit{int}\effto{\wEff\regid}\textit{unit})\times
(\textit{int}\effto{\wEff\regid}\textit{unit})} {\aEff\regid}
\)

\noindent
which  expresses   that  $\textit{setfactory}()$  allocates   in  some
\fi
(possibly fresh) region $\regid$ and returns operations that only read
$\regid$ (the  first one)  or write  in $\regid$ (the  second and
third one)  even though, physically, all three  functions read, write,
and allocate.

Thus, these functions can participate in corresponding
effect-dependent program equivalences, in particular, two successive
$\textit{mem}$ operations may be swapped and duplicated; identical
updates may even be contracted.
\iffull\footnote{We
could
also consider a more object-oriented variant that works with
a basic type $\textit{set}_\regid$ accepted as argument by the
operations.}
\fi

\paragraph{Interleaved Dummy allocation}
Consider the following example, which looks similar to the Dummy example
above, but where the dummy allocation happens after a proper allocation:
\iffull
\[
\begin{array}{l}
 e_1 = \letin{p}{\myref{0}}\letin{d}{\myref{0}} e; !p ~\textrm{ and }~
 e_2 = \letin{p}{\myref{0}} e; !p.
\end{array} 
\]
Here $d$ is not free in $e$, but $p$ may be free.
\else

\(
\hspace{-4mm}
\begin{array}{l}
 e_1 = \letin{p}{\myref{0}}\letin{d}{\myref{0}} e; !p ~\textrm{ and }~
 e_2 = \letin{p}{\myref{0}} e; !p.
\end{array}
\)

\noindent
Here $d$ is not free in $e$, but $p$ may be free.
\fi
This simple difference leads to many problems when 
attempting to prove
their equivalence. We sketch them below to also motivate
our technical solution introduced formally in the following Sections. 

As normally done the
evolution of the heaps can be formally captured by using Kripke models,
where, intuitively, a world contains the set of locations allocated by
programs.
Whenever there is an allocation,
we advance from the current world $\w$ to a world $\w_1$, which contains
some fresh locations. However, we do not have control over this
evolution. In our example, assume that the programs above
start at the same world $\w$. The allocation of the proper location,
$p$,
in $e_1$ and in $e_2$ will yield two different extensions $\w \to \w_1$ and
$\w \to \w_1'$, where some concrete locations, $\cloc_1$ and $\cloc_2$, are
allocated respectively. In fact, $\w_1$ and $\w_1'$ may even contain other
locations that are not used by the computations. For proving the
equivalence between these programs, we need
a way to capture that $\cloc_1$  and $\cloc_2$ are equivalent, without
requiring to identify the other locations not used by computations.

\begin{wrapfigure}{r}{.25\textwidth}
 \vspace{-12mm}
\begin{displaymath}
\vcenter{\xymatrix@C=1pc@R=0.5pc{ 
& \overline{\w}\\
\w_1\ar[ru]^x & & \w_1'\ar[lu]_{x'}\\
& \underline{\w}\ar[ru]_{u'} \ar[lu]^u
}}
\end{displaymath}
\caption{Pullback square.}
\label{fig:pullback}
\vspace{-9mm}
\end{wrapfigure}
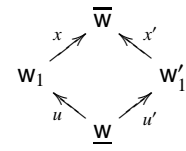
Our solution is to use \emph{pullback squares} as proofs. Their shape is
depicted in Figure~\ref{fig:pullback}.
where $\underline{\w}$ and $\overline{\w}$ are called, respectively, the
\emph{low point} and \emph{apex} of the square.  It helps to interpret
$\overline{\w} $ as a superset of $\w_1 \cup \w_1'$, that is, a world
containing all the locations mentioned in $\w_1$ and $\w_1'$, even the
locations not used by computations, while $\underline{\w} = \w_1
\cap \w_2$ (modulo renaming of location names) is a world containing only
the locations that need to be identified. Intuitively, the low point
is the part of the proof showing that resulting heaps of computations are
equivalent. This is formalized by Definition~\ref{teffde}. In the
example above,
the low point is a world where $\cloc_1$ and $\cloc_2$ are
shown to be equivalent. The remaining locations in $\w_1$ and $\w_1'$ that
are not used by computations may be ignored, that is, not be contained in
$\underline{\w}$. 
The apex, $\overline{\w}$, on the other hand, is the part of the proof 
showing that the corresponding \emph{values} resulting from computations,
$!p$ in the example above, are indeed equivalent (see again
Definition~\ref{teffde}). 

\section{Setoids}
\label{sec:setoids}
We define the \emph{category of setoids} as the exact completion of
the category of predomains, see
\cite{DBLP:conf/mfps/CarboniFS87,DBLP:conf/lics/BirkedalCRS98}. We
give here an elementary description using the language of dependent
types.  A \emph{setoid} $A$ consists of a predomain $|A|$ and for any
two $x,y\in |A|$ a set $A(x,y)$ of ``proofs'' (that $x$ and $y$ are
equal). The set of triples $\{(x,y,p) \mid p\in A(x,y)\}$ must itself
be a predomain and the first and second projections must be
continuous. Furthermore, there are continuous functions $r_A:\Pi x\in
|A|.A(x,x)$ and $s_A:\Pi x,y\in |A|.A(x,y)\rightarrow A(y,x)$ and
$t_A:\Pi x,y,z.A(x,y)\times A(y,z)\rightarrow A(x,z)$.
\iffull
We should explain what continuity of a dependent function like
$t(-,-)$ is: if $(x_i)_i$ and $(y_i)_i$ and $(z_i)_i$ are ascending
chains in $A$ with suprema $x,y,z$ and $p_i\in A(x_i,y_i)$ and $q_i\in
A(y_i,z_i)$ are proofs such that $(x_i,y_i,p_i)_i$ and $(y_i,z_i,q_i)_i$
are ascending chains, too, with suprema $(x,y,p)$ and $(y,z,q)$ then
$(x_i,z_i,t(p_i,q_i))$ is an ascending chain of proofs (by
monotonicity of $t(-,-)$) and its supremum is $(x,z,t(p,q))$.
\fi
\iffull
 Formally,
such dependent functions can be reduced to non-dependent ones using
pullbacks, that is $t$ would be a function defined on the pullback of
the second and first projections from $\{(x,y,p)\mid p\in A(x,y)\}$ to
$|A|$, but we find the dependent notation to be much more readable.
\fi
If $p\in A(x,y)$ we may write $p:x\sim y$ or simply $x\sim y$. We also
omit $|-|$ wherever appropriate. We remark that ``setoids'' also appear in
constructive mathematics and formal proof, see e.g.,
\cite{DBLP:journals/jfp/BartheCP03}, but the proof-relevant nature of
equality proofs is not exploited there and everything is based on sets
(types) rather than predomains. 
A morphism from setoid $A$ to setoid $B$ is an equivalence class of  pairs $f=(f_0,f_1)$ of
continuous functions where $f_0:|A|\rightarrow |B|$ and $f_1:\Pi
x,y\in|A|.A(x,y)\rightarrow B(f_0(x),f_0(y))$. Two such pairs
$f,g:A\rightarrow B$ are \emph{identified} if there exists a
continuous function $\mu:\Pi a\in|A|.B(f(a),g(a))$. 
\begin{proposition}
The category of setoids is cartesian closed; moreover, if $D$ is a setoid such that $|D|$ has a least element $\bot$ and there is also a least proof $\bot\in D(\bot,\bot)$ then there is a morphism of setoids $Y:[D\rightarrow D]\rightarrow D$ satisfying the usual fixpoint equations. 
\end{proposition}
\subsection{Pullback squares}
\label{sec:pullback}
A morphism $u$ in a category $\world$ is a monomorphism if $ux=ux'$
implies $x=x'$ for all morphisms $x,x'$.  A commuting square $xu=x'u'$
of morphisms is a pullback if whenever $xv=x'v'$ there is unique $t$
such that $v=ut$ and $v'=u't$. We write $\sq{x}{u}{x'}{u'}$ or
$\w\sq{x}{u}{x'}{u'}\w'$ (when $\w^( {}' {}^)=\dom{x^( {}' {}^)}$) for such a
pullback square. We call the common codomain of $x$ and $x'$ the
\emph{apex} of the pullback written $\overline{\w}$, while
the common domain of $u,u'$ the \emph{low point} of the
square written $\underline{\w}$. A pullback square $xu=x'u'$ is
\emph{minimal} if
whenever $fx=gx$ and $fx'=gx'$ then $f=g$, in other words, $x$ and
$x'$ are \emph{jointly epic}. A pair of morphisms $u,u'$ with common domain
is a span, a pair of morphisms $x,x'$ with common codomain is a
co-span. A category has pullbacks if every co-span can be completed to
a pullback square.
\begin{definition}[Category of worlds]
  A category $\world$ is a \emph{category of worlds} if it has
  pullbacks and every span can be completed to a minimal pullback
  square and all morphisms are monomorphisms.
\end{definition}
\begin{example}
  The category of sets and injections is a category of worlds.  
Given $f:X\to Z$ and $g:Y\to Z$, we form their pullback as
$X\xleftarrow{f^{-1}} fX\cap gY \xrightarrow{g^{-1}} Y$. This is minimal when $fX\cup gY = Z$. Conversely, given a span $Y\xleftarrow{f} X \xrightarrow{g}Z$, we can complete to a minimal pullback by
\newcommand{\myin}[1]{\mathit{in}_{#1}}
\iffull
\[
(Y\setminus fX) \uplus fX \xrightarrow{[\myin{1}, \myin{3}\circ f^{-1}]}
(Y\setminus fX) + (Z\setminus gX) + X
\xleftarrow{[\myin{2}, \myin{3}\circ g^{-1}]} (Z\setminus gX) \uplus gX
\]
where $[-,-]$ is case analysis on the disjoint union $Y = (Y\setminus
fX)\uplus fX$.
\else

\(
(Y\setminus fX) \uplus fX \xrightarrow{[\myin{1}, \myin{3}\circ f^{-1}]}
(Y\setminus fX) + (Z\setminus gX) + X
\xleftarrow{[\myin{2}, \myin{3}\circ g^{-1}]} (Z\setminus gX) \uplus gX
\)

\noindent
where $[-,-]$ is case analysis on the disjoint union $Y = (Y\setminus
fX)\uplus fX$.
\fi




  Given an arbitrary category $\mathbf{C}$, the
  category of worlds $\world_{\mathbf C}$ has objects pairs
  $(X,f)$ where $X$ is a set and $f:X\rightarrow |\mathbf{C}|$ is an
  $X$-indexed family of $\mathbf{C}$-objects. A morphism from $(X,f)$
  to $(Y,g)$ is an injective function $u:X\rightarrow Y$ and a family
  of isomorphisms $\phi_x:f(x)\simeq g(u(x))$. The first components of
  the pullbacks and minimal pullbacks are constructed as in the
  previous example.
\end{example}
\iffull
\begin{definition}
Let $\world$ be a category of worlds.  Two pullbacks
$\w\sq{x}{u}{x'}{u'}\w'$ and $\w\sq{y}{v}{y'}{v'}\w'$ are isomorphic
if there is an isomorphism $f$ between the two low points of the
squares so that $vf=u$ and $v'f=u'$, thus also $uf^{-1}=v$ and
$u'f^{-1}=v'$.
\end{definition}
\fi\iffull
\begin{lemma}\label{preo}
  Given a category of worlds $\world$, such that $\w, \w', \w'' \in
\world$, 
  if $\w\sq{x}{u}{x'}{u'}\w'$ and $\w'\sq{y}{v}{y'}{v'}\w''$ are
  pullback squares as indicated then there exist $z,z',t,t'$ such that
  $\w\sq{zx}{ut}{z'y'}{v't'}\w''$ is also a pullback. 
\end{lemma}
\fi
\iffull\begin{proof}
  Choose $z,z',t,t'$ in such a way that $\sq{z}{x'}{z'}{y}$ and
  $\sq{u'}{t}{v}{t'}$ are pullbacks.  The verifications are then an
  easy diagram chase.
\end{proof}
\fi

We write $r(\w)$ for $\w\sq{1}{1}{1}{1}\w$ and $s(\sq{x}{u}{x'}{u'}) =
\sq{x'}{u'}{x}{u}$ and $t(\sq{x}{u}{x'}{u'}, \sq{y}{v}{y'}{v'})=
\sq{zx}{z'y'}{ut}{v't'}$ where $z,z',t,t'$ \iffull are given by 
Lemma~\ref{preo} assuming an arbitrary choice. \else are chosen so
that all four participating squares are pullbacks.\fi
\iffull
\begin{lemma}\label{decomp}
A pullback  square $\sq{x}{u}{x'}{u'}$ in a category of worlds is
isomorphic to $t(\sq{x}{1}{1}{x}, \sq{1}{x'}{x'}{1})$.
\end{lemma}
\fi

\subsection{Setoid-valued functors} 
\label{subsec: func-setoids}
A functor $A$ from a category
of worlds $\world$ to the category of setoids comprises as usual for
each $\w\in \world$ a setoid $A\w$ and for each $u:\w\rightarrow \w'$
a morphism of setoids $Au:A\w\rightarrow A\w'$ preserving identities
and composition. If $u:\w\rightarrow \w'$ and $a\in A\w$ we may write
$u.a$ or even $ua$ for $Au(a)$ and likewise for proofs in $A\w$. Note
that $(uv).a=u.(v.a)$. 

\begin{definition}
We call a functor
\emph{pullback-preserving} (p.p.f.) if for every pullback square
$\w\sq{x}{u}{x'}{u'}\w'$ with apex $\overline{\w}$ and low point
$\underline{\w}$ the diagram $A\w\sq{Ax}{Au}{Ax'}{Au'}A\w'$ is a pullback
in $\Std$. This means that there is a continuous function of type 
\iffull
\[
\Pi a\in A\w.\Pi a'\in A\w'.A\overline{w}(x.a,x'.a')\rightarrow \Sigma
\underline{a}\in
A\underline{w}.A\w(u.\underline{a},a)\times
A\w'(u'.\underline{a},a')
\]
\else

\(
\Pi a\in A\w.\Pi a'\in A\w'.A\overline{w}(x.a,x'.a')\rightarrow \Sigma
\underline{a}\in
A\underline{w}.A\w(u.\underline{a},a)\times
A\w'(u'.\underline{a},a')
\)
\fi
\end{definition}
Thus, if two values $a\in A\w$ and $a'\in A\w'$ are equal in a common world
$\overline{\w}$ then this can only be the case because there is a value in
the ``intersection world'' $\underline{\w}$ from which both $a,a'$ arise. 
Intuitively, p.p.f.s will become the denotations
of value types.
\iffull
\begin{lemma}
  If $A$ is a p.p.f., $u:\w\rightarrow \w'$ and $a,a'\in A\w$,
  there is a continuous function $A\w'(u.a,u.a')\rightarrow
  A\w(a,a')$. Moreover, the ``common ancestor'' $\underline{a}$ of $a$ and $a'$ is unique up to
$\sim$. 
\end{lemma}
Note that the ordering on worlds and world morphisms is discrete so
that continuity only refers to the $A\w'(u.a,u.a')$ argument.
\fi
\iffull
\begin{definition}[Morphism of functors]\label{morphfun}
  If $A,B$ are p.p.f., a morphism from $A$ to $B$ is a pair
  $e=(e_0,e_1)$ of continuous functions where
  $e_0:\Pi\w.A\w\rightarrow B\w$ and $e_1:\Pi\w.\Pi\w'.\Pi
  x:\w\rightarrow\w'.\Pi a\in A\w.\Pi a'\in
  A\w'.A\w'(x.a,a')\rightarrow B\w'(x.e_0(a),e_0(a'))$. A proof that
  morphisms $e,e'$ are equal is given by a continuous function
  $\mu:\Pi \w.\Pi a\in A\w.B\w(e(a),e'(a))$.
\end{definition}
These morphisms compose in the obvious way and so
the faithful functors and morphisms between them form a category. 
\fi
\subsection{Fibred setoids}\label{fibse} In order to provide
meanings for computation types we need a weaker variant of p.p.f.,
namely, \emph{fibred setoids}. These lack the facility of transporting
values along world morphisms but instead allow the proof-relevant
comparison of values at different worlds provided the latter are
related by a pullback square.
\begin{definition}
\label{def:fibred-setoids}
  A \emph{fibred setoid} over a category of
  worlds $\world$ is given by a predomain $T\w$ for every
  $\w\in\world$ and for every pullback square $\w\sqsol \w'$ and
  elements $a\in T\w$ and $a'\in T\w'$ a set $T\sqsol(a,a')$ so that
  the set of tuples $(a,a',q)$ with $q\in T\sqsol(a,a')$ is a
  predomain with continuous projections.
 
  Next, we need continuous operations $r,s,t$ so that $r(a)\in
  Tr(\w)(a,a)$ when $a\in T\w$ and $s(q)\in Ts(\sqsol)(a',a)$ when
  $q\in T\sqsol(a,a')$ and $t(q,q')\in Tt(\sqsol,\sqsol')(a,a'')$ when
  $q\in T\sqsol(a,a')$ and $q'\in T\sqsol'(a',a'')$.

  In addition, for any two isomorphic pullback squares $\sqsol$ and
  $\sqsol'$ between $\w$ and $\w'$ there is a continuous operation
  of type $\Pi a\in T\w.\Pi a'\in T\w'.T\sq(a,a')\rightarrow T\sq'(a,a')$.

  Finally, for each pullback square $\sqsol = \w\sq{x}{u}{x'}{u'}\w'$
  with apex $\overline{\w}$ and low point $\underline{\w}$ there is a
  continuous function of type
\iffull
\[
\Pi t\in T\w.\Pi t'\in T\w'.T\sqsol(t,t')\rightarrow \Sigma
\underline{t}\in T\underline{\w}.T\sq{u}{1}{1}{u}(\underline{t},t)\times 
T\sq{u'}{1}{1}{u'}(\underline{t},t')
\]
\else

\(
\Pi t\in T\w.\Pi t'\in T\w'.T\sqsol(t,t')\rightarrow \Sigma
\underline{t}\in T\underline{\w}.T\sq{u}{1}{1}{u}(\underline{t},t)\times 
T\sq{u'}{1}{1}{u'}(\underline{t},t')
\)
\fi
\end{definition}
Note the similarity of the last operation to pullback-preservation. 

\begin{example}
  If $A$ is a p.p.f.,\  we obtain a fibred setoid $S(A)$ as follows: 
  $S(A)\w=A\w$ and if $\w\sq{x}{u}{x'}{u'}\w'$ with apex $\overline \w$, 
 define the proof set $S(A)
  \sq{x}{u}{x'}{u'}(a,a')=A\overline \w(x.a,x'.a')$. 
\iffull If $\sq{x}{u}{x'}{u'}$ and
  $\sq{y}{v}{y'}{v'}$ are two composable pullback squares with
  composite $\sq{zx}{ut}{z'y'}{v't'}$ and $p\in S(A) \sq{x}{x'}{u}{u'}
  (a,a')$ and $p'\in S(A) \sq{y}{y'}{v}{v'} (a',a'')$, then the composite
proof of $t_{S(A)}(p,p')\in S(A)\sq{zx}{ut}{z'y'}{v't'}(a,a'')$ is given
by $t_A(z.p,z'.p')$. Indeed, if $\hat\w=\cod z$ is the apex of the
composite square  then  $z.p\in A{\hat\w}(zx.a,zx'.a')$ and $z'.p'\in
A{\hat\w}(z'y.a',z'y'.a'')$ and $zx'.a'=z'y.a'$ since $zx'=z'y$ so the
two proofs compose in $A{\hat \w}$.  \fi
\end{example}
\iffull
\begin{lemma}\label{baf}
  Let $T$ be a fibred setoid.  The elements $\underline{t}$ given by pullback
  preservation are unique up to $\sim$.  If $u:\w\rightarrow \w'$ is
  an isomorphism then there is a continuous function
  $Tu:T\w\rightarrow T\w'$ and it is bijective up to $\sim$ with
  inverse $T(u^{-1})$. If $\sqsol$ and $\sqsol'$ are isomorphic pullback squares then there are continuous back and forth functions 
$\Pi t.\Pi t'.T\sqsol(t,t')\rightarrow T\sqsol'(t,t')$. 
\end{lemma}
\fi
\begin{definition}
  A \emph{morphism} $f$ from fibred setoid $T$ to fibred setoid $T'$
  is an equivalence class  of pairs of continuous functions $f_0:\Pi \w.T\w\rightarrow T'\w$ and
  $f_1:\Pi \w,\w'.\Pi \w \sqsol\w'.\Pi a\in T\w.\Pi a'\in
  T\w'.T\sqsol(a,a') \rightarrow T'\sqsol(f_0(\w,a),f_0(\w',a'))$. 

  Two such pairs  $f,f'$ are identified if there exists a continuous 
function that assigns to each $\w$ and $a\in T\w$ a
proof $\mu(a)\in T{r(\w)}(f_0(\w,a),f_0'(\w,a))$. 
\end{definition}
\iffull
\begin{lemma}\label{meanterm}
  If $A$ is a p.p.f.\ and $T$ is a fibred setoid then in
  order to specify a morphism from $S(A)$ to $T$ with given first
  component $f_0: \Pi\w.A\w\rightarrow T\w$ it is enough to provide a
  continuous function $f_{0.5}:\Pi\w,\w'.\Pi x:\w\rightarrow \w'.\Pi
  a\in A\w.\Pi a'\in A\w'.A\w'(x.a,a')\rightarrow
  T\sq{x}{1}{1}{x}(f_0(a),f_0(a'))$.
\end{lemma} 
\begin{proof}
  If $(f_0,f_1)$ is a morphism we can define $f_{0.5}$ by
  $f_{0.5}(x,p)=f_1(x,a,a',p)$ noting that $p\in
  S(A)\sq{x}{1}{1}{x}(a,a')$.  Conversely, given $f_{0.5}$ to define
  $f_1$ we pick a pullback square $\w\sq{x}{u}{x'}{u'}\w'$ with apex
  $\overline{\w}$ and $a\in A\w, a'\in A\w'$ and $p\in A{\overline
    \w}(x.a,x'.a')$, i.e., a proof in $S(A){\sqsol}(a,a')$.  Applying
  $f_{0.5}$ to $r(-)$ yields the morphism $p_1\in
  T\sq{x}{1}{1}{x}(f_0(a),f_0(x.a))$; moreover, applying $f_{0.5}$ to
$s(p)$
  yields $p_2\in T\sq{x'}{1}{1}{x'}(f_0(a'),f_0(x.a))$.  Then,
  $t(p_1,s(p_2))\in
  Tt(\sq{x}{1}{1}{x},\sq{1}{x'}{x'}{1})(f_0(a),f_0(a'))$ so that
  Lemmas~\ref{decomp} and \ref{baf} yield the desired proof in the square
  $T\sq{x}{u}{x'}{u'}(f_0(a),f_0(a'))$.

  The second part of the lemma about equality is just a restatement of
  the definition of equality of morphisms of fibred setoids. 
\end{proof}\fi
\iffull
\begin{lemma}
Let $A,B$ be p.p.f.\ For every morphism $e:A\rightarrow B$ there  is a
morphism $S(e):S(A)\rightarrow S(B)$ such that $S(e)_0=e_0$. Thus, in
particular $S(-)$ is a full and faithful functor from the category  of
p.p.f.\ on $\world$ to the category of fibred setoids over $\world$. 
\end{lemma} 
\fi
\subsection{Contravariant functors and relations} \label{reles} The
role of the next concept is to give meaning to abstract stores.
\begin{definition}
  A contravariant functor $\Astores$ from a category of worlds
  $\world$ to the category of setoids comprises for each $\w\in
  \world$ a \emph{nonempty} setoid $\Astores \w$ and for each morphism
  $u:\w_0\rightarrow \w$ a setoid morphism $\Astores u:\Astores
  \w\rightarrow \Astores \w_0$ such that $u\mapsto \Astores u$ preserves
  identities and composition.
\end{definition}
If $\sigma\in\Astores \w$ and $u:\w_0\rightarrow \w$ we
  write $\sigma.u$ or $\sigma u$ for $\Astores u(\sigma)$. Note that
  $\sigma.(uv)=(\sigma.u).v$. Intuitively, $\sigma.u$ can be interpreted
as the abstract heap obtained by forgetting locations in $\sigma$ that have
been ``allocated'' by the world evolution specified by $u$, namely, those 
appearing in \w and not in $\w_0$.
\iffull
The following definition corresponds to the p.p.f. used for values, but
now for abstract heaps: In particular, an abstract heap at the
low-point of a pullback square is the result of forgetting locations from
an abstract heap at its 
apex.
\fi
\begin{definition}
A contravariant functor $\Astores$ preserves \emph{minimal pullbacks}
if whenever $\w\sq{x}{u}{x'}{u'}\w'$ with apex $\overline{\w}$ and low
point
$\underline{\w}$ is a \emph{minimal pullback square} then 
the diagram $\Astores\w\sq{\Astores u}{\Astores u'}{\Astores x}{\Astores
x'} \Astores \w'$ is a pullback in $\Std$.
\end{definition}
This means in particular that if $\sigma\in
\Astores\w, \sigma'\in \Astores\w'$ and $\sigma.u\sim\sigma'.u'$ 
 then there exists a ``pasting''
$\overline{\sigma}\in \Astores\overline{\w}$ such that
$\overline{\sigma}.x\sim \sigma$ and
$\overline{\sigma}.x'\sim\sigma'$ and $\overline{\sigma}$ is unique
up to $\sim$. Moreover the passage from the given data to
$\overline\sigma$ and the witnessing proofs is continuous. 
\iffull
Application to the trivial minimal pullback $\sq{u}{1}{1}{u}$ and
nonemptiness yields the following result.
\begin{lemma}
\label{lem:unique-upto-sim}
   For every $u:\w\rightarrow \w'$ and $\sigma\in\Astores\w$ there is
morphism of setoids $\Astores \w\rightarrow \Astores \w'$ which is right
inverse to $(-).u$. 
\end{lemma}
The ``unique up to $\sim$'' clause allows us in particular to assert the $\sim$-equality of two abstract stores $\sigma,\sigma'\in\Astores\overline{w}$ by proving $\sigma.x\sim\sigma'.x$ and $\sigma.x'\sim\sigma'.x'$ separately when $\sq{x}{u}{x'}{u'}$ is a minimal pullback  with apex $\overline{\w}$.
\fi
\begin{definition}
A \emph{relation} $R$ on such a contravariant functor $\Astores$
consists of an admissible subset $R\w\subseteq \Astores \w \times \Astores
\w$ such that $(\sigma,\sigma')\in R \w$ and $u:\w_0\rightarrow \w$ implies
$(\sigma.u,\sigma'.u)\in R \w_0$ and if $p:\sigma\sim \sigma_1$ and
$p':\sigma'\sim\sigma_1'$ then $(\sigma_1,\sigma_1')\in R\w$, as well. 
\end{definition}

It would be natural to let relations be proof-relevant as well, but we
refrain from doing so at this stage for the sake of simplicity.

\section{Computational model}
\label{test}
We use a setoid interpretation in order to justify
nontrivial type-dependent observational equivalences for the
language above. This interpretation is parametric over an
\emph{instantiation},
defined below. 
\begin{definition}
An \emph{instantiation} comprises the following data. 

\noindent \textbullet~a category of worlds $\world$;

\noindent \textbullet~a full-on-objects subcategory $\mathbf{I}$ of
\emph{inclusions}
  (in other words, a subset of the morphisms closed under composition
  and comprising the identities) with the property that every morphism
  $u$ can be factored as $u=fi$ and $u=jg$ with $f,g$ isomorphisms and
  $i,j$ inclusions;

\noindent \textbullet~a contravariant, minimal-pullback-preserving, functor $\Astores$ from $\world$ to
the category of setoids;

\noindent \textbullet~for each $\w\in\world$ a relation
  $\Vdash_\w\subseteq\Stores\times \Astores \w$ subject to the axiom
  that $\heap\Vdash_\w \sigma$ and $u\in\mathbf{I}(\w_0,\w)$ implies
  $\heap\Vdash_{\w_0}\sigma.u$;

\noindent \textbullet~a set of elementary effects $\elEffs$ and for each
effect $\eff$ a
set  $\Rscr(\eff)$ of relations on $\Astores$. As usual, one defines
effects as sets of elementary effects and extends $\Rscr$ to all
effects by $\Rscr(\emptyset)=$ ``all relations on $\Astores$
 (in the sense described in Section~\ref{reles})'' and
$\Rscr(\eff)=\bigcap_{\eff_0\in\eff}\Rscr(\eff_0)$.
\end{definition}

We give \iffull three \else two \fi examples of instantiations. 
\ifapp
The appendix
\fi
contains a third example, mirroring our previous
model \cite{DBLP:conf/ppdp/BentonKBH09}.

\subsection{Sets of locations}
In the first one, called \emph{sets of locations}, worlds are finite
sets of (allocated) locations (taken from $\labs$) and their morphisms
are injective functions with inclusions being actual
inclusions. Abstract stores are given by $\Astores \w=\{\heap\mid
\dom\heap\supseteq\w\}$ with $\Astores\w(\heap,\heap')=\star$, always,
and $\Astores u$ given by renaming locations.

We put $\heap \Vdash_\w \heap'$ whenever $\heap=\heap'$.  We only have
one elementary effect here, $\aEff{}$, representing the allocation of
one or more fresh names. Note that if $R$ is a relation on $\Astores$
then $R\w$ is either total or empty and if $u:\w\rightarrow\w'$ then
$R\w'\neq\emptyset\Rightarrow R\w\neq\emptyset$. A relation $R$ is in
$\mathcal{R}(\aEff{})$ if for every inclusion $u:\w\rightarrow\w'$ one
also has $R\w\neq\emptyset\Rightarrow R\w'\neq\emptyset$, thus $R$ is
oblivious to world extensions.

\iffull
\subsection{Flat stores} 
The second instantiation, called \emph{flat stores}, assumes that heap
locations contain merely integer values and no pointers. 
Possible worlds are
finite sets of locations together with a function that associates each
location a \emph{region} taken from a fixed set $\Regids$ of
regions. World morphisms must preserve this tagging. We write
$\cloc\in\w$ and $\cloc\in\w(\regid)$ to mean that $\cloc$ occurs in $\w$
and with region $\regid$ in the second case.  Abstract stores
$\Astores \w$ comprise those heaps $\heap\in\Stores$ with $\dom
\heap\supseteq
\w$ and such that $\cloc\in\w$ and $\heap\in\Astores\w$ implies that
$\heap(\cloc)$ is an integer value, $\intt v$ for $v\in\mathbb{Z}$
 (thus all locations hold integer
values). We put $\heap\sim \heap'$ in $\Astores\w$ iff for all $\cloc\in
\w$
one has $\heap(\cloc)=\heap'(\cloc)$. In this case there is a unique proof,
say $\star$.  For morphism $u:\w\rightarrow \w'$ we define $\Astores
u:\Astores \w'\rightarrow \Astores \w$ by renaming concrete locations
according to $u$. The elementary effects are $\rEff\regid,
\wEff\regid,\aEff\regid$ representing reading from within, writing
into, allocating within a region $\regid$. The associated sets of
relations are given by
\iffull
\[
 \begin{array}{lcl}
R\in\mathcal{R}(\rEff{\regid})&\iff& (\sigma,\sigma')\in R\w \Rightarrow
\forall \cloc\in \w(\regid). \sigma(\cloc)=\sigma'(\cloc)
\\[3pt]

R\in\mathcal{R}(\wEff{\regid})&\iff& (\sigma,\sigma')\in R\w \Rightarrow
\forall \cloc\in \w(\regid).\forall v\in\mathbb{Z}.\Rightarrow 
(\sigma[\cloc{\mapsto}\intt{v}],\sigma'[\cloc{\mapsto}\intt{
v}])\in R\w
\\[3pt]

R\in\mathcal{R}(\aEff{\regid})&\iff& (\sigma,\sigma')\in R\w 
\Rightarrow \forall \w_1.\forall u\in\mathbf{I}(\w,\w_1). (\dom
{\w_1}\setminus \dom{\w} \subseteq \dom{\w_1(\regid)}) \\
&& \quad
\Rightarrow \forall \sigma_1\in \Astores \w_1,\sigma_1'\in \Astores \w_1'.
\sigma_1.u\sim \sigma\wedge \sigma_1'.u\sim\sigma'\wedge\\
&& \qquad \forall\cloc\in
\dom{\w_1}\setminus\dom \w.  \sigma_1(\cloc)=\sigma_1'(\cloc) \Rightarrow
(\sigma_1,\sigma_1')\in R \w_1
 \end{array}
\]
This essentially models the setting of our earlier relation-based
\else

\(
\hspace{-5mm}
 \begin{array}{lcl}
R\in\mathcal{R}(\rEff{\regid})&\iff& (\sigma,\sigma')\in R\w \Rightarrow
\forall \cloc\in \w(\regid). \sigma(\cloc)=\sigma'(\cloc)
\\[3pt]

R\in\mathcal{R}(\wEff{\regid})&\iff& (\sigma,\sigma')\in R\w \Rightarrow
\forall \cloc\in \w(\regid).\forall v\in\mathbb{Z}.\Rightarrow 
(\sigma[\cloc{\mapsto}\intt{v}],\sigma'[\cloc{\mapsto}\intt{
v}])\in R\w
\\[3pt]

R\in\mathcal{R}(\aEff{\regid})&\iff& (\sigma,\sigma')\in R\w 
\Rightarrow \forall \w_1.\forall u\in\mathbf{I}(\w,\w_1). (\dom
{\w_1}\setminus \dom{\w} \subseteq \dom{\w_1(\regid)}) \\
&& \quad
\Rightarrow \forall \sigma_1\in \Astores \w_1,\sigma_1'\in \Astores \w_1'.
\sigma_1.u\sim \sigma\wedge \sigma_1'.u\sim\sigma'\wedge\\
&& \qquad \forall\cloc\in
\dom{\w_1}\setminus\dom \w.  \sigma_1(\cloc)=\sigma_1'(\cloc) \Rightarrow
(\sigma_1,\sigma_1')\in R \w_1
 \end{array}
\)

\noindent
This essentially models the setting of our earlier relation-based
\fi
account of reading, writing, and allocation with integer values stores
\cite{DBLP:conf/ppdp/BentonKBH09} the difference being that allocation
is modelled with relations on the same level as reading and writing
and that the stores being related share the same layout. 
\fi

\subsection{Abstract locations}
To formulate the \iffull third \else second \fi instantiation, called
\emph{Heap PERs}, we need the concept of an \emph{abstract location}
which generalises physical locations in that it models a portion of
the store that can be read from and updated. Such portion may comprise
a fixed set of physical locations or a varying such set (as in the
case of a linked list with some given root). It may also reside in
just a part of a physical location, e.g., comprise the two low order
bits of an integer value stored in a physical location. Furthermore,
the equality on such abstract location may be coarser than physical
equality, e.g., two linked lists might be considered equal when they
hold the same set of elements, and there may be an invariant, e.g.\
the linked list should contain integer entries and be neither circular
nor aliased with other parts of the heap. This then prompts us to
model an abstract location as a partial equivalence relation (PER) on
heaps together with two more components that describe how
modifications of the abstract location interact with the heap as a
whole. Thus, next to a PER, an abstract location also contains a bunch
of (continuous) functions that model \emph{writing to the} abstract
location. These functions are closed under composition (thus form a
category) and are idempotent in the sense of the PER modelling
equality.

Thirdly, a ``footprint'' which is a heap-dependent set of
physical locations which overapproximates the effect of ``the
guarantee'' so as to enable the creation of fresh abstract locations
not knowing the precise nature of the other abstract locations that
are already there. (These footprints are very similar to
accessibility maps, first introduced for reasoning in a model of state based on FM-domains \cite{DBLP:conf/tlca/BentonL05}.)

\begin{definition}
\label{def:abstract-location}
  An \emph{abstract location} $\loc$ (on the chosen predomain
  $\Stores$) consists of the following data:
\begin{compactitem}
\item a nonempty, admissible partial equivalence relation (PER) $\loc^R$ on
  $\Stores$ modelling the ``semantic equality'' on the bits of the
  store that $\loc$ uses (a ``rely-condition'');
\item a set $\loc^G$ of continuous functions on $\Stores$ closed by
composition, modelling the functions that ``write only on $\loc$'' leaving 
other locations alone (a ``guarantee condition'');
\item a continuous function $\loc^F:\Pi
  \heap\in\Stores.\mathcal{P}(\dom\heap)$ describing the ``footprint''
   of the abstract location (where the ordering on the powerset
$\mathcal{\dom\heap}$ is of course discrete).
\end{compactitem}
subject to the conditions
\begin{compactitem}
\item if $\iota\in \loc^G$ and $(\heap,\heap')\in\loc^R$ then $(\iota(\heap),\iota(\heap')), 
(\iota(\heap),\iota(\iota(\heap))), (\iota(\heap'),\iota(\iota(\heap')))\in\loc^R$,  
\item if $\forall\cloc\in\loc^F(\heap).\heap_1(\cloc)=\heap(\cloc)$
  and $\forall\cloc\in\loc^F(\heap').\heap_1'(\cloc)=\heap'(\cloc)$
  then $(\heap,\heap')\in\loc^R$ implies
  $(\heap_1,\heap_1')\in\loc^R$; thus $\loc^R$ ``looks'' no further
  than the footprint;
\item if $\iota \in \loc^G$ and $\iota(\heap) = \heap_1$ then
  $\dom\heap\subseteq\dom{\heap_1}$ and
  $\cloc\in\dom\heap\setminus \loc^F(\heap)$ implies $\cloc\not\in
  \loc^F(\heap_1)$ and $\heap(\cloc)=\heap_1(\cloc)$.
\end{compactitem}
Two abstract locations $\loc_1,\loc_2$ are independent if 
\begin{compactitem}
\item for $i=1,2$ and $\iota(\heap)=\heap_1$ for $\iota\in\loc^G_i$ one has 
$(\heap,\heap)\in\loc_i^R, (\heap,\heap')\in\loc_{3-i}^R\Rightarrow (\heap_1,\heap')\in\loc_{3-i}^R$
and
$\cloc\in\dom\heap\setminus\loc_{3-i}^F(\heap)$
then $\cloc\notin\loc_{3-i}^F(\heap_1)$;
\item If $(\heap_1,\heap_1)\in\loc_1^R$ and $(\heap_2,\heap_2)\in\loc_2^R$
there exists $\heap$ such that $(\heap,\heap_1)\in\loc_1^R$ and
$(\heap,\heap_2)\in\loc_2^R$. (Amounting to
$\heap/(\loc_1^R\cap\loc_2^R)$ being a cartesian product of $\heap/\loc_1^R$ and $\heap/\loc_2^R$.)
\end{compactitem}
If $\loc_1, \loc_2$ are independent, we form a joint location
$\loc_1\otimes\loc_2$ by $(\loc_1\otimes\loc_2)^R=\loc_1^R\cap\loc_2^R$
and
$(\loc_1\otimes\loc_2)^G=(\loc_1^G\cup\loc_2^G)^*$ and
$(\loc_1\otimes\loc_2)^F(\heap)=\loc_1^F(\heap)\cup\loc_2^F(\heap)$. 

\end{definition}


If $\cloc\in\Locs$ is a concrete location, we can define an abstract
counterpart by putting $\cloc^R=\{(\heap,\heap')\mid
\heap(\cloc)=\heap'(\cloc)\}$ and $\cloc^G$ is the set with a
write function for each value that may be stored in $\cloc$. 
For instance, if $\cloc$ stores booleans, then $\cloc^G$ contains the
functions $write_\mtrue$ and $write_\mfalse$, where
$write_\mtrue(\heap) = \heap'$ such that $\heap'(\cloc) =
\mtrue$ and 
for all other locations $\cloc' \neq \cloc$, $\heap'(\cloc') =
\heap(\cloc')$.
When $\cloc_1 \neq \cloc_2$ then the induced abstract locations are
independent.

The next example illustrates that abstract locations may be
independent although their footprints share some concrete locations.
Fix a concrete location $\cloc$ and define two abstract
locations $\loc_1$ and $\loc_2$ both with footprint consisting of the
location $\cloc$. Moreover, $(\heap, \heap')$ belong, respectively, to
the rely of location $\loc_i$ ($i=1,2$) if $\heap(\cloc)$ and
$\heap'(\cloc)$ are both integers whose $i$-th significant bit
agrees. The ``guarantee'' $\loc_i^G$ might then contain functions that
set the $i$-th bit to some fixed value and leave the other bits
alone. It is easy to see that $\loc_1,\loc_2$ are independent.

Thirdly, let $\cloc_1, \cloc_2$ be two distinct concrete locations and
for heap $\heap$ and finite integer sets $U_1,U_2$ define
$P(\heap,U_1,U_2)$ to mean that in $\heap$ the locations $\cloc_1,
\cloc_2$ point to non-overlapping integer lists with \emph{sets} of
elements $U_1$ and $U_2$. Now define abstract location $\loc_i$ by
$\loc_i^R= \{(\heap,\heap')\mid \exists
U_1,U_2.P(\heap,U_1,U_2)\wedge P(\heap',U_1,U_2)\}$ 
 and $\loc_i^F(\heap)=$ ``locations reachable from
$\cloc_i$'' if $\cloc$ points to a well-formed list of integers in
$\heap$ and $\emptyset$ otherwise. The guarantee component $\loc_i^G$
contains all the (idempotent) functions $\iota$ that leave the locations
not in the footprint of $\loc_i$ alone. That $\iota(\heap) = \heap'$,
such that $\heap'(\cloc') = \heap(\cloc')$ for all $\cloc' \in
\dom{\heap}\setminus \loc_i^F$.
Again, $\loc_1$ and $\loc_2$ are independent.

The role of the footprints $\loc^F$ is to provide a minimum amount of
interaction with physical allocation. If $\loc$ is an abstract
location and $\heap_0$ the current heap so that
$(\heap_0,\heap_0)\in\loc^R$ then we may, e.g., allocate
$(\heap_1,\cloc)=\new(\heap_0,\intt 0)$, and define an abstract location
$\loc_1$ by

\(
 \begin{array}{lcl}
\loc_1^R&=&\{(\heap,\heap')\mid
\heap(\cloc)=\heap'(\cloc)\in\intt{\mathbb{Z}}\wedge
\cloc\not\in\loc^F(\heap)\wedge \cloc\not\in\loc^F(\heap')\}\\
\loc_1^G&=&\{\iota\mid \iota(\heap) = \heap_1 \Rightarrow \forall
\cloc'\neq\cloc.\heap(\cloc')=\heap_1(\cloc')\}\\
\loc_1^F(\heap)&=&\{\cloc\}  
 \end{array}
\)

\noindent
We now know that $\loc$ and $\loc_1$ are independent and, furthermore, 
$(\heap_1,\heap_1)\in(\loc\otimes\loc_1)^R$. 
\begin{definition}
Abstract  locations $\loc_1,\dots,\loc_n$ are mutually independent if they are pairwise independent and whenever $(\heap_i,\heap_i)\in\loc_i$ for $i=1\dots n$ then there is $\heap$ such that $(\heap_i,\heap)\in\loc_i$ for $i=1\dots n$. 
\end{definition}
\begin{lemma}
Abstract locations $\loc_1,\dots,\loc_{n+1}$ are mutually independent iff $\loc_1,\dots,\loc_n$ are mutually independent and $\loc_{n+1}$ is independent of $\loc_1\otimes\dots\otimes\loc_n$. 
\end{lemma}
\subsection{Heap PERs}
We are now ready to formulate the \iffull third \else second \fi
instantiation \emph{Heap
  PERs}. We assume an infinite set of \emph{regions} $\Regids$. A
world $\w$ comprises a finite set of mutually independent abstract
locations (written $\w$) and as in the case of flat stores a tagging of 
locations with regions from $\Regids$ 
location. We write $\loc\in \w(\regid)$ to mean that $\loc\in\w$
is tagged with $\regid$. We define $\Astores \w =
\{\heap\in\Stores\mid \forall \loc\in\w.(\heap,\heap)\in\loc^R\}$ and
$\Astores \w(\sigma,\sigma')=\{\star\}\iff \forall \loc\in\w.
(\sigma,\sigma')\in\loc^R$ and $\Astores
\w(\sigma,\sigma')=\emptyset$ otherwise. Again, $\heap\Vdash_\w \sigma$
iff $\heap=\sigma$. 

A morphism from $\w$ to $\w'$ is given by an injective function $u_0:
\w\rightarrow \w'$ and a pair of partial continuous functions
$u_1,u_2:\Stores\partfun\Stores$. Intuitively, the function $u_1$
is used
to map the heaps in the PERs of locations in $\w$ to $\w'$ according to the
renaming of locations specified in $u_0$, while $u_2$
does the same but from $\w'$ to $\w$. Formally,
$\forall \sigma,\sigma'\in\Astores \w.  \forall
\loc\in\w.(\sigma,\sigma')\in\loc^R\Rightarrow
(u_1(\sigma),u_1(\sigma'))\in u_0(\loc)^R \wedge
(u_2(u_1(\sigma)),\sigma)\in\loc^R$ and $\forall
\sigma,\sigma'\in\Astores \w'.  \forall
\loc\in\w.(\sigma,\sigma')\in u_0(\loc)^R \Rightarrow
(u_2(\sigma),u_2(\sigma'))\in \loc^R \wedge
(u_1(u_2(\sigma)),\sigma)\in u_0(\loc)^R$. The same is valid for 
guarantees of locations, by replacing $\cdot^R$ by $\cdot^G$.
Now, $\Astores
u(\sigma) = u_2(\sigma)$.  Such a morphism $u$ is an inclusion if
$u_0$ is an inclusion and $u_1,u_2$ are the identity function.

The elementary effects track reading, writing, and allocating at the
level of regions: $\wEff \regid$ (writing within
region $\regid$), $\rEff\regid$ (reading from within region $\regid$),
$\aEff\regid$ (allocating within region $\regid$).
The sets of relations on $\Astores$ modelling elementary effects are then
given by 
\iffull
\[
 \begin{array}{lll}
R\in\mathcal{R}(\rEff{\regid})& \iff & 
(\sigma,\sigma')\in R\w \Rightarrow \forall \loc\in\w(\regid).
(\sigma,\sigma')\in\loc^R\\[2pt]

R\in\mathcal{R}(\wEff{\regid})&\iff& 
(\sigma,\sigma')\in R\w \Rightarrow \forall \loc\in\w(\regid).\forall
\iota \in \loc^G. (\iota(\heap), \iota(\heap')) \in R\w\\[2pt]


R\in\mathcal{R}(\aEff{\regid})& \iff & 
(\sigma,\sigma')\in R\w \Rightarrow \forall \w_1.\forall
u\in\mathbf{I}(\w,\w_1).(\w_1\setminus \w \subseteq
\w_1(\regid))\Rightarrow \forall
\sigma_1,\sigma_1' \in\Astores\w_1.\\
&& (\sigma_1.u\sim\sigma \wedge
\sigma_1'.u\sim\sigma'\wedge(\sigma_1,\sigma_1')\in \bigcap_{\loc\in
\w_1\setminus \w}
\loc^R)\Rightarrow (\sigma_1,\sigma_1')\in R\w_1
 \end{array}
\]
Thus, a relation $R\in\mathcal{R}(\rEff{\regid})$ ensures that
\else

\(
\hspace{-5mm}
 \begin{array}{lll}
R\in\mathcal{R}(\rEff{\regid})& \iff & 
(\sigma,\sigma')\in R\w \Rightarrow \forall \loc\in\w(\regid).
(\sigma,\sigma')\in\loc^R\\[2pt]

R\in\mathcal{R}(\wEff{\regid})&\iff& 
(\sigma,\sigma')\in R\w \Rightarrow \forall \loc\in\w(\regid).\forall
\iota \in \loc^G. (\iota(\heap), \iota(\heap')) \in R\w\\[2pt]


R\in\mathcal{R}(\aEff{\regid})& \iff & 
(\sigma,\sigma')\in R\w \Rightarrow \forall \w_1.\forall
u\in\mathbf{I}(\w,\w_1).(\w_1\setminus \w \subseteq
\w_1(\regid))\Rightarrow \forall
\sigma_1,\sigma_1' \in\Astores\w_1.\\
&& (\sigma_1.u\sim\sigma \wedge
\sigma_1'.u\sim\sigma'\wedge(\sigma_1,\sigma_1')\in \bigcap_{\loc\in
\w_1\setminus \w}
\loc^R)\Rightarrow (\sigma_1,\sigma_1')\in R\w_1
 \end{array}
\)

\noindent
Thus, a relation $R\in\mathcal{R}(\rEff{\regid})$ ensures that
\fi
locations being read contain ``equal'' (in the sense of $\loc^R$)
values; a relation $R\in\mathcal{R}(\wEff{\regid})$ is oblivious to
writes to any abstract location in $\regid$, and a relation
$R\in\mathcal{R}(\aEff{\regid})$ is oblivious to extensions of the
current world provided that it only adds abstract locations in region
$\regid$, that the initial contents of these newly allocated locations
are ``equal'' in the sense of $(-)^R$ and that nothing else is
changed.

\section{Proof-relevant Logical Relations}
\label{sec:logical-relations}
Given an instantiation, e.g.\ one of the above examples, we 
interpret types (and typing contexts) as p.p.f.\ over
$\world$ and types with effect as a fibred setoid over $S(\world)$. A
term in context $\Gamma\vdash e:\ety{\tau}{\eff}$ will be interpreted as a
morphism $\sem{e}$ from $S(\sem{\Gamma})$ to $T_\eff\sem{\tau}$ where
$T_\eff$ takes p.p.f.\ and effects to fibred setoids and is
given below in Definition~\ref{teffde}. Derivations of equations
will be interpreted as equality proofs between the corresponding morphisms
and can be used to deduce observational equivalences
(Theorem~\ref{obseq}). 

This, however, requires a loose relationship of the setoid interpretation
with the actual meanings of raw terms which is given by realization
relations $\Vdash^A$. Their precise format and role are described in the
following two  definitions. 
\begin{definition}\label{defte}
  A \emph{semantic type} is a pair $(A,\Vdash^A)$ where $A$ is a
  p.p.f.\ (on $\world$) and $\Vdash^A_\w$ is an admissible
  subset of $\Values\times A\w$ for each $\w\in \world$ such that for
  every inclusion $u:\w\rightarrow \w'$ one has that $\vval\Vdash^A_\w
  \val$ implies $\vval\Vdash^A_{\w'} u.\val$.  A \emph{semantic
computation} is a pair $(T,\Vdash^T)$ where $T$ is
  a fibred setoid over $\world$ and $\Vdash_\w^T$ is an
  admissible subset of $\Comps\times T\w$ for each $\w$.
\end{definition}
\begin{definition}
  Let $(\Gamma,\Vdash^\Gamma)$ and $(A,\Vdash^A)$ be semantic types
  and let $(T,\Vdash^T)$ be a semantic computation.  If
  $e:S(\Gamma)\rightarrow T$ is a morphism of fibred setoids and
  $f:\Values\rightarrow \Comps$ then we write $f\Vdash^{\Gamma\vdash
    T}e$ to mean that for some representative $(f_0,f_1)$ of $f$ one has that 
 whenever $\eta\Vdash^\Gamma_\w \gamma$ then
  $f_0(\eta)\Vdash^T_\w e(\gamma)$ holds for all worlds $\w$.
\end{definition}

The following definition, corresponding to that in Fig.~\ref{fig:oldklr}, is where the machinery introduced above pays off.
In particular, it defines the semantics of computations, where proofs,
i.e., pullback squares, are constructed. 

\begin{definition}\label{teffde}\label{def:semantic-computation}
  Let $A$ be a semantic type and $\eff$ an effect. A semantic
  computation $T_\eff A$ is defined as follows: 

\smallskip
\noindent \textbullet~ (Objects) Elements of $(T_\eff A)\w$ are pairs
$(\cval_0,\cval_1)$ of partial continuous functions where
\iffull
\[
 \begin{array}{ll}
  \cval_0 : & \Astores \w\partfun \Sigma
  \w_1.\mathbf{I}(\w,\w_1)\times \Astores \w_1\times A\w_1
\end{array}
\]
and $\cval_1$ is as follows.  If $R\in \Rscr(\eff)$ and
\else

\(
 \begin{array}{ll}
  \cval_0 : & \Astores \w\partfun \Sigma
  \w_1.\mathbf{I}(\w,\w_1)\times \Astores \w_1\times A\w_1
\end{array}
\)

\noindent
and $\cval_1$ is as follows.  If $R\in \Rscr(\eff)$ and
\fi
$(\sigma,\sigma')\in R\w$ then $\cval_1(R,\sigma,\sigma')$ either is
undefined and $\cval_0(\sigma)$ and $\cval_0(\sigma')$ are both
undefined or else $c_1(R,\sigma,\sigma')$ is defined and then
$\cval_0(\sigma)$ and $\cval_0'(\sigma')$ are both defined, say
$\cval_0(\sigma)=(\w_1,u,\sigma_1,a)$ and
$\cval_0(\sigma')=(\w_1',u',\sigma_1',a')$. In this case,
$\cval_1(R,\sigma,\sigma')$ returns a pair $(\sq{x}{v}{x'}{v'},p)$
where $\w_1\sq{x}{v}{x'}{v'}\w_1'$ such that $xu=x'u'$.  Furthermore, $p\in
A \overline w (x.a,x'.a')$ and,
finally, $(\sigma_1.u,\sigma_1'.u')\in R \underline{\w}$ where
$\underline\w$ and $\overline\w$ are low point and apex of 
$\sq{x}{v}{x'}{v'}$.  

\smallskip
\noindent \textbullet~ (Proofs) As usual, proofs only look at the $(-)_0$
components.
  Thus, if $(\cval_0,\_)\in T_\eff A \w$ and $(\cval_0',\_)\in T_\eff
  A \w'$ and $\sq{x}{v}{x'}{v'}$ is in $S(\world)(\w,\w')$ with apex
  and low point $\overline\w, \underline\w$ then a proof in $(T_\eff A)
  \sq{x}{v}{x'}{v'}(\cval,\cval')$ is a partial continuous function
  $\mu$ which given $\sigma\in\Astores \w$ and $\sigma'\in\Astores
  \w'$ and $p:\sigma.v\sim \sigma'.v'$ either is undefined and then
  $\cval_0(\sigma)$ and $\cval_0'(\sigma')$ are both undefined or else
  is defined and then $\cval_0(\sigma)$ and $\cval_0'(\sigma')$ are
  both defined with results, say,
  $\cval_0(\sigma)=(\w_1,u,\sigma_1,\val)$ and
  $\cval_0'(\sigma')=(\w_1',u',\sigma_1',\val')$. In this case,
  $\mu(p)$ returns a tuple $(\sq{x_1}{v_1}{x_1'}{v_1'},q)$ satisfying
  $x_1uv=x_1'u'v'$ and $q\in A{\overline{\w_1}}(x_1.\val,x_1'.\val')$
  with $\overline {\w_1} = \cod{x_1}$ and $\sigma_1.v_1\sim \sigma_1.v_1'$
  in $\Astores{\underline {\w_1}}$. 

\smallskip
\noindent \textbullet~ (Realization)
  If $c\in \Comps$, we define $\ccval\Vdash^{T_\eff A}_\w
(\cval_0,\cval_1)$ to mean that
  whenever  $\heap\Vdash_\w \sigma$ then $\ccval(\heap)$ is defined iff
$\cval_0(\sigma)$ is defined and if $\ccval(\heap)=(\heap_1,\vval)$ and
$\cval_0(\sigma)=(\w_1,u,\sigma_1,\val)$ then
$\heap_1\Vdash_{\w_1}\sigma_1$ and $\vval\Vdash^A_{\w_1} \val$. 
\end{definition}
\iffull
\begin{proposition}
The semantic computation $T_\eff A$ as defined in
Definition~\ref{def:semantic-computation} is a fibred setoid.
\end{proposition}
\begin{proof}
  The tricky case is to show the existence of a transitive operation.
  It is here that we require the independence of abstract locations as
  stated in Definition~\ref{def:abstract-location}, which implies that
  $\Astores$ is also minimal-pullback-preserving.

Assume that there are proofs in $p_1 : T_\eff A \sq{x_1}{v_1}{x_1'}{v_1'}
(\cval, \cval')$ and $p_2 : T_\eff A \sq{x_2}{v_2}{x_2'}{v_2'} (\cval',
\cval'')$ where 
$\w \sq{x_1}{v_1}{x_1'}{v_1'} \w'$ and $\w' \sq{x_2}{v_2}{x_2'}{v_2'}
\w''$. We also have $\sigma \in \Astores \w$ and $\sigma'' \in \Astores
\w''$, such that they are equivalent in the pullback of
the low points of these two pullback squares. Let $\underline{\q}$ be such
pullback. 

In order to use the proofs $p_1$ and $p_2$, we need to construct
from $\sigma$ and $\sigma''$ an abstract heap $\sigma' \in \Astores\w'$.
Let $\overline{\q}$ be the minimal pullback over the apexes of the two pullback
squares $\w \sq{x_1}{v_1}{x_1'}{v_1'} \w'$ and $\w'
\sq{x_2}{v_2}{x_2'}{v_2'} \w''$. Then $\w$ and $\w''$ form a pullback
square with apex $\overline{\q}$ and low point $\underline{\q}$. Since
$\Astores$ is minimal-pullback-preserving, there is a $\sigma_{\q} \in \Astores
\overline{\q}$, such that it is equivalent to $\sigma$ and $\sigma''$ when
taken to the world $\underline{\q}$. We now define $\sigma' \in \Astores
\w'$ to be $\sigma_{\q}$ taken to the world $\w'$. We thus have $\sigma'
\in \Astores
\w'$, and $\sigma'' \in \Astores \w''$, such that $\sigma.v_1 \sim
\sigma'.v_1'$ and $\sigma'.v_2' \sim \sigma''.v_2'$.

We can now use the $p_1$ and $p_2$. In particular, 
let $\cval(\sigma) = (\w_1, u_1, \sigma_1, \val_1)$, $\cval'(\sigma') =
(\w_1', u_1', \sigma_1', \val_1')$, and $\cval''(\sigma'') = (\w_1'',
u_1'', \sigma_1'', \val_1'')$. From the proofs, we get two
pullback squares $\w_1\sq{}{}{}{}\w_1'$ and $\w_1'\sq{}{}{}{}\w_1''$. It
is easy to show that the values obtained are equal in the minimal pullback over 
the
apexes of these two pullback squares and that the abstract heaps are
equivalent in the pullback of their low points.
\end{proof} 
\else
Proving that a semantic computation $T_\eff A$ as  in
Definition~\ref{def:semantic-computation} is a fibred setoid is 
nontrivial.  The tricky case is the existence of a transitivity
 operation. 
It is here that we need the independence of abstract locations as
stated in Definition~\ref{def:abstract-location}, which implies that 
$\Astores$ is also minimal-pullback-preserving. 
\fi
\iffull
\begin{definition}[cartesian product]
If $(A,\Vdash^A)$ and $(B,\Vdash^B)$ are semantic types their cartesian product $(A\times B,\Vdash^{A\times B})$ is defined by $(A\times B)\w=A\w\times B\w$ (cartesian product of setoids) and $(v_1,v_2)\Vdash_\w^{A\times B}(a,b)\iff v_1\Vdash_\w^{A}a\wedge v_2\Vdash_\w^{B}b$. 
\end{definition}
\begin{definition}[function space]
  Let $(A,\Vdash^A)$ be a semantic type and $(T,\Vdash^T)$ be a
  semantic computation. We define a semantic type $(A{\Rightarrow}T,
  \Vdash^{A{\Rightarrow}T})$ as follows.  
  An object $f$ of $(A{\Rightarrow} T)\w$ is a pair $(f_0,f_1)$ of
  continuous functions where $f_0$ assigns to each $\w_1$ and
  $v:\w\rightarrow \w_1$ a continuous function $f_0(v):A\w_1\rightarrow
  T\w_1$. The second component $f_1$ assigns to each
  $v:\w\rightarrow \w_1$ and $v_1:\w_1\rightarrow \w_2$ a continuous
  function $\Pi a\in A\w_1.\Pi a'\in A\w_2.A\w_2(v_1.a,a')\rightarrow
  T\sq{v_1}{1}{1}{v_1}(f_0(v,a),f_0(v_1v,a'))$.

  If $f,f'\in |A{\Rightarrow}T|$ then a proof $\mu\in
  (A{\Rightarrow}T)(f,f')$ is a continuous function assigning to each
  $v:\w\rightarrow \w_1$ and $a\in A\w_1$ a proof $\mu(v,a)\in
  T\sq{1}{1}{1}{1}(f_0(v,a), f_0'(v,a))$.

  If $u:\w\rightarrow \w'$ and $f=(f_0,f_1)\in (A{\Rightarrow}T)\w$
  then $u.f\in (A{\Rightarrow}T)\w'$ is given by precomposition with
  $u$, i.e., $(u.f)_0(v,a)=f_0(vu,a)$, etc.


  As for the realisation relation $\Vdash^{A{\Rightarrow}T}$ we put
  $v\Vdash^{A{\Rightarrow}T}_\w f$ to mean that $v=\funn g$ for some
  $g$ and whenever $i:\w\rightarrow \w_1$ is an inclusion and
  $u\Vdash^A_{\w_1}a$ then $g(u)\Vdash^T_{\w_1} f(i,a)$.
\end{definition}
Notice that unlike morphisms the elements of the function space are \emph{not} identified if they are ``provably equal.''. 
Notice also that if $v\Vdash^{A{\Rightarrow} T}_\w f$ implies
$v\Vdash^{A{\Rightarrow}T}_{\w_1}i.f$ whenever $i:\w\rightarrow \w_1$
is an inclusion.
\else
Details, along with the construction of the cartesian product
$(A\times B,\Vdash^{A\times B})$ and function space
 $(A{\Rightarrow}T, \Vdash^{A{\Rightarrow}T})$, given semantic types
 $(A,\Vdash^A)$ and $(B,\Vdash^B)$ and computation $(T,\Vdash^T)$,
may be found in the 
\ifapp
appendix.
\fi
\fi
\subsection{Fundamental theorem}
Given a semantic type $\sem{A}$ for each basic type
$A$ we can interpret any type $\tau$ as a semantic type $\sem{\tau}$
by putting
$\sem{\tau_1\effto\eff\tau_2}=\sem{\tau_1}{\Rightarrow}T_\eff\sem{\tau_2}
$. A
typing context $\Gamma=x_1{:}\tau_1,\dots,x_n{:}\tau_n$ is interpreted
as the semantic type
$\sem{\Gamma}=(1\times\sem{\tau_1})\times\dots)\times\sem{\tau_n}$
where $1$ is the constant functor returning the discrete setoid
$\{()\}$.

To every typing derivation $\Gamma\vdash t:\ety{\tau}{\eff}$ we then
associate a morphism 
$
\sem{\Gamma\vdash t:\ety{\tau}{\eff}}:S(\sem{\Gamma})\rightarrow
T_\eff\sem{\tau}
$
such that $\sem{t}\Vdash^{\sem{\Gamma}\rightarrow
T_\eff\tau}\sem{\Gamma\vdash t:\ety{\tau}{\eff}}$. 
(Note: \emph{this} is point where the untyped semantics is related with the
abstract one.)
For every equality derivation $\Gamma\vdash t=t':\ety{\tau}{\eff}$ we
have $\sem{\Gamma\vdash
t:\ety{\tau}{\eff}} = 
\sem{\Gamma\vdash t':\ety{\tau}{\eff}} 
$,
where the two typing derivations $\Gamma\vdash t:\ety{\tau}{\eff}$ and
$\Gamma\vdash t':\ety{\tau}{\eff}$ are the canonical ones associated with
the equality derivation $\Gamma\vdash t=t':\ety{\tau}{\eff}$.
\iffull 
In what follows we define semantic counterparts to the generic
syntactic constructions common to all instantiations, namely
application and abstraction, sequential composition, and recursion
that allow us to define this interpretation of derivations in a
compositional fashion. Having given these semantic counterparts we
then omit the formal definition of the interpretation $\sem{-}$.
\begin{lemma}[Abstraction]
Let $\Gamma,A$ be  semantic types, $T$ a semantic computation. There is a function $\lambda$ so that if $e:S(\Gamma\times A)\rightarrow T$ is a morphism of fibred setoids then $\lambda(e):S(\Gamma)\rightarrow A{\Rightarrow}T$. Moreover, if $e\sim e'$ then $\lambda(e)\sim\lambda(e')$ and if $f\Vdash^{\Gamma\times A\rightarrow T} e$ then $\lambda\eta.\lambda a.f(\eta,a)\Vdash^{\Gamma\rightarrow A{\Rightarrow}T} \lambda(e)$. 
\end{lemma}
\begin{lemma}[Application]
Let $A$ be a semantic type and $T$ be a semantic computation. There  is a
morphism $\textit{app} : S((A{\Rightarrow}T)\times A)\rightarrow T$ and
$\lambda (f,a).f(a)\Vdash^{((A{\Rightarrow}T)\times A)\rightarrow T}
\textit{app}$.  
\end{lemma}
We elide assertions about $\sim$-versions of beta-eta-equality, and semantic 
rendering of subeffecting and the existence of ``value morphisms'' of type $S(A)\rightarrow T_\eff A$ for any semantic type $A$. 

\begin{lemma}[let]
Let $\Gamma,A,B$ be semantic types and $\eff$ an effect. 
There is a function $\textit{let}$ such that if $e_1:S(\Gamma)\rightarrow
T_\eff A$ and $e_2:S(\Gamma\times A)\rightarrow T_\eff B$ are morphisms
then $\textit{let}(e_1,e_2):S(\Gamma)\rightarrow T_\eff B$. Moreover, if
$e_1\sim
e_1'$ and $e_2\sim e_2'$ then $\textit{let}(e_1,e_2)\sim
\textit{let}(e_1',e_2')$. 
Finally, if $f_1\Vdash^{\Gamma\rightarrow T_\eff A} e_1$ and
$f_2\Vdash^{\Gamma\times A\rightarrow T_\eff B} e_2$ then $\lambda \eta.
\lambda \heap.\textit{let } (\heap_1,\vval){=}f_1(\eta)(\heap)\textit{ in }
f_2(\eta,\vval)(\heap_1)\Vdash^{\Gamma\rightarrow T_\eff
A}\textit{let}(e_1,e_2)$. 
\end{lemma}

\begin{proof}
Consider the following definition for the first component of the morphism
$\textit{let}(e_1,e_2)$ which is only defined when ${e_1}$ and
${e_2}$ are defined. The type of this component is $\sem{\Gamma}\w \to
T_\eff \sem{B}\w$. Hence, assume a world $\w$, and a context $\gamma \in
\sem{\Gamma}\w$, then one returns an object $(\cval_0, \cval_1) \in
T_\eff \sem{B}\w$. The first component $\cval_0$ is:
\(
 \Pi \w. \Pi \gamma \in \sem{\Gamma}\w . \Pi \sigma \in \Astores \w.
{e_2}(\w_1)(\gamma,\val_1)\sigma_1
\)
where ${e_1}(\w)(\gamma)\sigma = (\w_1, u_1, \sigma_1, \val_1)$.

For the second component, $\cval_1$, assume a relation $R \in
\Rscr(\eff)$, and two abstract heaps $\sigma, \sigma' \in \Astores \w$
such that $(\sigma, \sigma') \in R\w$. From ${e_1}$ we get a 
proof $\w_1 \sq{x_1}{v_1}{x_1'}{v_1'} \w_1'$, where
${e_1}(\w)(\gamma)\sigma
=
(\w_1, u_1, \sigma_1, \val_1)$ and ${e_1}(\w)(\gamma)\sigma' =
(\w_1', u_1', \sigma_1', \val_1')$, such that $(\sigma_1.v_1,
\sigma_1'.v_1') \in R$ and $p : \sem{A}\overline{\w_1}(x_1.\val_1,
x_1'.\val_1')$. Applying ${e_2}$ on $\sigma_1.v_1$ and
$\sigma_1'.v_1'$ we get a proof $\q_2 \sq{y_2}{v_2}{y_2'}{v_2'} \q_2'$,
such that $(\tilde{\sigma_2}.v_2, \tilde{\sigma_2}'.v_2') \in R$. However,
we need to show that the heaps obtained from applying ${e_2}$ on $\sigma_1$
and $\sigma_1'$ (using the correct world and context), namely $\sigma_2$
and $\sigma_2'$, are related. For this we rely on
the morphism $(e_2)_1$. In particular,
we use $(e_2)_1$ on the pullback $\w_1 \sq{1}{x_1}{x_1}{1} \und{\w_1}$ and
obtain a pullback $\w_2 \sq{}{}{}{} \q_2$ such that $\sigma_2$ and
$\tilde{\sigma_2}$
are equal in its low point. Similarly, applying $(e_2)_1$ on the
pullback $\und{\w_1} \sq{x_1'}{1}{1}{x_1'} \w_1'$, we get a pullback 
$\q_2' \sq{}{}{}{} \w_2'$, where $\tilde{\sigma_2}'$ is equal to
$\sigma_2'$ in its pullback. Using Lemma~\ref{preo}, we compose the
pullbacks $\w_2 \sq{}{}{}{} \q_2$, $\q_2 \sq{}{}{}{} \q_2'$ and $\q_2'
\sq{}{}{}{} \w_2'$, obtaining a common pullback $\und{\q}$, where 
$\sigma_2$ and $\sigma_2'$ when taken to $\und{\q}$ are in $R$. 

The morphism ${\textit{let}(e_1,e_2)\sim
\textit{let}(e_1',e_2')}$ can be then defined when ${e_1 \sim e_1'}$
and ${e_2 \sim e_2'}$ are defined.
Assume a pullback $\w \sq{1}{1}{1}{1}\w$ and an abstract heap $\sigma \in
\Astores \w$ and a context $\gamma \in \sem{\Gamma}\w$. Using the
morphism between $e_1$ and $e_1'$ on these objects, we obtain a
pullback $\w_1 \sq{x_1}{v_1}{x_1'}{v_1'} \w_1'$, $p_1 \in
\sem{A}\overline{\w_1}(x_1.\val_1, x_1'.\val_1')$ and $q_1 : \sigma_1.v_1
\sim \sigma_1'.v_1'$, where ${e_1}(\w)(\gamma)\sigma = (\w_1, u_1,
\sigma_1,
\val_1)$ and ${e_1'}(\w)(\gamma)\sigma = (\w_1', u_1', \sigma_1',
\val_1')$. 
From the pullback preserving property of computations and $p_1$,
there is a common value $\und{\val} \in \sem{A}\und{\w_1}$ and context
$\und{\gamma} \in \sem{\Gamma}\und{\w_1}$ which are equal, respectively, to
$\val_1$ and $\val_1'$, and $\gamma$ and $\gamma'$ (when taken to the
correct world). We then construct a
proof $\sem{\Gamma \times A} \und{\w_1}$. We now apply twice the
morphism between $e_2$ and $e_2'$ once in the pullback $\w_1 \sq{}{}{}{}
\und{\w_1}$ and
another on the
pullback $\und{\w_1} \sq{}{}{}{} \w_1'$, obtaining two pullbacks
$\w_2 \sq{}{}{}{} \q_2$ and $\q_2 \sq{}{}{}{} \w_2'$. From
Lemma~\ref{preo}, we can compose them where the resulting values and
heaps are equal.
\end{proof}
\begin{lemma}[fix]\label{fixlse}
  Let $\Gamma,D$ be semantic types so that for each $\w$ the predomain
  $D\w$ is a domain with least element $\bot\w$ such that
  $(\bot\w,\bot\w,r(\bot\w))\leq (d,d',p)$ holds for every proof $p\in
  D(d,d')$ and such that $x.\bot_\w=\bot_{\w'}$ holds for every
  $x:\w\rightarrow\w'$.\footnote{For example $D=A {\Rightarrow}T_\eff B$
for semantic types $A,B$.}
\begin{compactenum}
\item[i] There then exists a function
  $\textit{fix}$ so that whenever $e:\Gamma\times D\rightarrow D$ then
  $\textit{fix}(e):\Gamma\rightarrow 
  D$
\item[ii] If $e\sim e'$ then $\textit{fix}(e)\sim
  \textit{fix}(e')$. Furthermore, the fixpoint and unrolling equations
  from Lemma~\ref{fixlse} hold. 
\item[iii] Finally, if $f\Vdash^{\Gamma\times D\rightarrow D}e$ then
$f^\dagger\Vdash \textit{fix}(e)$.    
\end{compactenum}
\end{lemma}
\begin{proof}
{\bf MH: dagger notation and accompanying lemma have disappeared} 
For every $\w$ we have $e_0\w:\Gamma\w\times D\w
\rightarrow D\w$. We can thus form $\textit{fix}(e)_0\w:=(e_0\w)^\dagger :
\Gamma \w\rightarrow D\w$.  It remains to define $\textit{fix}(e)_1$. To do that, we
recall that we 
have an ascending chain of elements 
$\textit{fix}^n(e)_0\w(\gamma)\in D\w$ given by 
$\textit{fix}^0(e)_0\w(\gamma)=\bot_\w$ and 
$\textit{fix}^{n+1}(e)_0\w(\gamma)=e_0\w(\gamma,\textit{fix}
^n(e)_0\w(\gamma))$ and have
$\textit{fix}(e)_0\w(\gamma)=\sup_n\textit{fix}^n(e)_0\w\gamma$. 
Now suppose that $\gamma\in\Gamma\w$ and $x:\w\rightarrow \w'$ and
$\gamma'\in\Gamma\w'$ and $p\in\Gamma\w'(x.\gamma,\gamma')$. Write
$d_n=\textit{fix}^n_0\w(\gamma)$ and $d_n'=\textit{fix}^n_0\w'(\gamma')$.
Inductively, we get proofs $p_n\in
D\w'(x.d_n,d_n')$ where $p_0=r(\bot_{\w'})$ (note that
$x.\bot_\w=\bot_{\w'}$) and $p_{n+1}=e_1(p,p_n)$. Since
$(x.\bot_\w,\bot_{\w'},r(\bot_{\w'}))\leq (x.d_1,d_1',p_1)$ we obtain by
monotonicity of $e_1$ and induction that $(x.d_n,d_n',p_n)$ is an ascending
chain with supremum $(x.\sup_n d_n,\sup_n d_n',q)$ for some proof $q$
which we take as $\textit{fix}(e)_1(p)$. Note that the passage from $p$ to
$q$ is continuous. 
\end{proof}
\else 
In essence, one has to provide a semantic counterpart for every
syntactic concept, e.g. let, fix, etc. Details are in the 
\ifapp
appendix.
\fi
\fi

\subsection{Observational equivalence}
Let $\textrm{Int}$ stand for the constant functor that returns the
discrete setoid on the set $\mathbb{Z}$ of integers. We define
$v\Vdash^{\mathrm{Int}}_\w i\iff v=\intt i$. We also assume that there
is some initial store and abstract store $\heap_0, \sigma_0$ and a
world $\w_0$ such that $\heap_0\Vdash_{\w_0}\sigma_0$. For instance,
$\w_0$ can be the empty world with no locations and accordingly
$\heap_0$ the initial store at startup. 

\begin{definition}
  Let $(A,\Vdash^A)$ be a semantic type.  We define an
  \emph{observation of type $A$} as a morphism $o:A\rightarrow T_\eff
  \mathrm{Int}$ for some $\eff$ and a function $f$ so that
  $f\Vdash^{A\rightarrow T_\eff\mathrm{Int}}o$.

  Two values $v,v'$ are \emph{observationally equivalent at type $A$}
  if for all observations $f,o$ of type $A$ one has that $f(v)(\heap_0)$ is
  defined iff $f(v')(\heap_0)$ is defined and when
  $f(v)(\heap_0)=(\heap_1,v_1)$ and $f(v')(\heap_0)=(\heap_1',v_1')$
  then $v_1=v_1'$.
\end{definition}
Taking $o=\sem{\vdash f:\tau\effto\eff\inttype}$ immediately yields the following:
\begin{proposition}
If $v,v'$ are observationally equivalent at type $\sem{\tau}$ and $f$ is a
term such that $\vdash f:\tau\effto\eff\inttype$ then $\sem{f}(v)(\heap_0)$
is
  defined iff $\sem{f}(v')(\heap_0)$ is defined and when
  $\sem{f}(v)(\heap_0)=(\heap_1,v_1)$ and
$\sem{f}(v')(\heap_0)=(\heap_1',v_1')$
  then $v_1=v_1'$.
\end{proposition}
\begin{theorem}[Observational equivalence]
\label{thm:obs-equivalence}\label{obseq}
If $(A,\Vdash^A)$ is a semantic type and $v\Vdash^A_{\w_0}e$ and
$v'\Vdash^A_{\w_0}e'$ with $e\sim e'$ in $A_{\w_0}$ then $v$ and  $v'$ are
observationally equivalent at type $A$. 
\end{theorem}
\begin{proof}
  We have $f(v)\Vdash^{T_\eff \mathrm{Int}}_{\w_0} o(e)$ and $f(v')\Vdash^{T_\eff \mathrm{Int}}_{\w_0} o(e')$ and also $\mu: 
o(e)\sim_{\sq{1}{1}{1}{1}} o(e')$ in $T_\eff\mathrm{Int}$ for some $\mu$ as in Definition~\ref{teffde}.

The application $\mu$ to $\sigma_0,\sigma_0,r(\sigma_0)$ either is
undefined in which case $o(e)(\sigma_0)$ and $o(e')(\sigma_0)$ and
$f(v)(\heap_0)$ and $f(v')(\heap_0)$ are all undefined, the latter by
the definition of $\Vdash^{T_\eff \mathrm{Int}}$. Otherwise, we get
$f(v)(\heap_0)=(\heap_1,v_1)$ and $f(v')(\heap_0)=(\heap_1',v_1')$ and
$o(e)(\sigma_0)=(\sigma_1,i_1)$ and $o(e')(\sigma_0)=(\sigma_1',i_1')$
where, by definition of realization in $T_\eff\mathrm{Int}$ and
$\mathrm{Int}$, we have $v_1=\intt{i_1}$ and $v_2=\intt{i_2}$. Now,
$\mu(\sigma_0,\sigma_0,r(\sigma_0))$ returns a pullback
$(\sq{x_1}{v_1}{x_1'}{v_1'},q)$ such that, in particular, $x_1.i_1\sim
x_2.i_2$, whence $i_1=i_2$ since $\mathrm{Int}$ is constant and then
$v_1=v_2$ as required.
\end{proof}

\section{Applications}
\label{sec:application}
In what follows we use our semantics to establish a number of
effect-dependent semantic equalities, hence program equivalences in
the sense of observational equivalences. We also give some
semantically justified typings of concretely given functions, in
particular ``set factory'' described in Section~\ref{examples}.
More examples are discussed in the 
\ifapp
appendix.
\fi

\subsection{Sets of locations}
We work in the instantiation ``sets of locations''.  Recall the
example, ``dummy allocation'' from Section~\ref{examples}.  Suppose
that $f\Vdash^{\Gamma\vdash T_\eff A} e$. Now, put
$\textit{dummy}(e)(\w)(\gamma \in \sem{\Gamma}\w) (\heap \in \Astores \w)
= e(\w)(\gamma)(\heap')$, where $\heap'$ is the heap
obtained by adding a dummy location to \heap. We have
$\textit{dummy}(f)\Vdash^{\Gamma\vdash
  T_\eff A} \textit{dummy}(e)$ since $\Vdash$ is oblivious to
extensions of the store. Therefore, reflexivity also furnishes a proof
of equality. It also means that, semantically, $\textit{dummy}(f)$
does not need to flag the allocation effect $\aEff{}$ since no
semantically visible world extension takes place. 

For the Interleaved Dummy Allocation example, on the other hand, there is
an extra step caused by the proper allocation, which yields a world
extension $\w \to \w_1$ and $\w
\to \w_1'$. In order to show the equivalence, we construct a
proof, i.e., a pull-back square $\w_1\sq{}{}{}{}\w_1'$, where the
allocated concrete locations are identified in its low point. Then the
reasoning is the same as above used for showing the semantic equivalence
of the Dummy example.

This is different in the following example. Define a semantic type $N$
of names by letting $N\w$ be the discrete setoid on the set $\w$ and
$Nu(\cloc)=u(\cloc)$ and $v\Vdash_\w^N\cloc\iff v=\reff\cloc$.
\iffull
\(
 \begin{array}{ll}
f ~=~\sem{\myref{0}}\\
g ~=~ \sem{\letin{x}{\myref{0}}{\letin{y}{\myref{0}}{
(x,y)}}}\\
h ~=~ \sem{\letin{x}{\myref{0}}{\letin{y}{\myref{0}}{
(y,x)}}}
 \end{array}
\)

\noindent
We now define semantic counterparts $\textsf{f}:S(1)\rightarrow
T_{\aEff{}}N$, 
$\textsf{g}, \textsf{h}:S(1)\rightarrow T_{\aEff{}}N$. We omit the dummy
arguments of type $1$.
\else
Moreover, 
$f =\sem{\myref{0}}, 
g = \sem{\letin{x}{\myref{0}}{\letin{y}{\myref{0}}{
(x,y)}}}$, and $ 
h = \sem{\letin{x}{\myref{0}}{\letin{y}{\myref{0}}{
(y,x)}}}$.
We now define semantic counterparts $\textsf{f}:S(1)\rightarrow
T_{\aEff{}}N$, 
$\textsf{g}, \textsf{h}:S(1)\rightarrow T_{\aEff{}}N$, where
\fi
\iffull
\[
 \begin{array}{lcl}
\textsf{f}_0\w(\sigma) & = & (\w_1,i_1,\sigma_1,\cloc_1), \\
\textsf{g}_0\w(\sigma) & = & (\w_2,i_2i_1,\sigma_2,(\cloc_1,\cloc_2))\\
\textsf{h}_0\w(\sigma) & = & (\w_2,i_2i_1,\sigma_2,(\cloc_2,\cloc_1))  
 \end{array}
\]
Here and in what follows it is assumed that
\else

\(
\hspace{-5mm}
 \begin{array}{lcllcllcl}
\textsf{f}_0\w(\sigma) & = & (\w_1,i_1,\sigma_1,\cloc_1),~ 
\textsf{g}_0\w(\sigma) & = & (\w_2,i_2i_1,\sigma_2,(\cloc_1,\cloc_2)),
\textrm{ and }
\textsf{h}_0\w(\sigma) & = & (\w_2,i_2i_1,\sigma_2,(\cloc_2,\cloc_1))  
 \end{array}
\)

\noindent
Here and in what follows it is assumed that
\fi
$\textit{new}(\sigma)=(\cloc_1,\sigma_1)$ and
$\textit{new}(\sigma_1)=(\cloc_2,\sigma_2)$ and
$\w_1=\w\cup\{\cloc_1\}$ and $\w_2=\w_1\cup\{\cloc_2\}$. Recall that
$\Astores\w\subseteq\Stores$. Finally, $i_1:\w\rightarrow\w_1$ and
$i_2:\w_1\rightarrow\w_2$ stand for the obvious inclusions.  We use
analogous definitions for the primed variants.

In order to define $\textsf{f}_{0.5}$ we start with
$u:\w\rightarrow\w'$ and $\sigma\in\Astores\w,\sigma'\in\Astores\w'$,
$R\in\Rscr(\aEff{})$ such that $(\sigma,u.\sigma')\in R\w$. Define
$u':\w_1\rightarrow\w_1'$ so that $u'i_1=i_1'u$, that is
$u'(\cloc\in\w)=u(\loc)$, $u'(\cloc_1)=\cloc_1'$. We now return the
pullback square $\w_1\sq{u'}{1}{1}{u'}\w_1'$ with apex $\w_1'$ and low
point $\w_1$ and the trivial proof that $u'.\cloc_1=\cloc_1'$. This 
settles the definition of $\textsf{f}_{0.5}$, since $R\w_1$ is total since
$R\in\Rscr(\aEff{})$. Notice though, that we cannot avoid the allocation
effect here.

The functions $g_{0.5}$ and $h_{0.5}$ are defined analogously.

We now construct a proof that $\textsf{g}\sim \textsf{h}$, recall that only
$\textsf{g}_0$ and
$\textsf{h}_0$ are needed for this. Given $\w,\sigma$ and the notation from
above this proof amounts to a pullback square
$\w_2\sq{x}{v}{x'}{v'}\w_2'$ such that $xi_2i_1=x'i_2'i_1'u$ and
$x.(\cloc_1,\cloc_2)=x'.(\cloc_2,\cloc_1)$ and $\sigma_2.v \sim
\sigma_2'.v'$. Note that, accidentally, the final abstract stores of
both computations are the same, namely, $\sigma_2$. Now let $f$ be the
bijection that swaps $\cloc_1,\cloc_2$ and fixes everything else. We
then put $\sq{x}{v}{x'}{v'} := \sq{1}{f}{f}{1}$.  Now, obviously
$(\cloc_1,\cloc_2)=f.(\cloc_2,\cloc_1)$ and $\sim$-equality of
abstract stores is trivial by definition.

\subsection{Heap PERs}
In this section we generalize our earlier collection of effect-dependent program
equivalences 
\cite{DBLP:conf/ppdp/BentonKBH07} to the abstract locations of the
Heap PERs instantiation. We first show how the
set factory indeed has the announced effect typings and thus can
participate in effect-dependent equivalences.

\paragraph{Set factory}
Let $\w$ be a world and $\sigma\in\Astores\w$. Suppose that $\sigma_1$
arises from $\sigma$ by allocating a fresh set data structure, e.g., a
linked list, with entry point(s) $E$. Let $\loc_1$ be the abstract
location describing this fresh data structure, i.e.,
$(\heap,\heap')\in\loc_1^R\iff$ the data structures starting from $E$
in $\heap,\heap'$ are well-formed, denote the same set, and do not
overlap with the footprints of all the abstract locations in $\w$. The
footprint $\loc_1^F$ comprises the locations that make up this 
data structure assuming that $(\heap,\heap)\in\loc^R$, otherwise any
value can be chosen. Finally, $\loc^G$ contains idempotent
functions, $\iota$, such that $\iota(\heap) = \heap_1$
and $\heap_1$ agree on all concrete locations from
$\dom{\heap}\supseteq\loc^F(\heap)$ and, moreover,
$\dom{\heap_1}\supseteq \dom\heap$. 

Now for any chosen region $\regid$ we add $\loc_1$ to $\regid$ to
yield a new world $\w_1$. The function $\textit{setfactory}_0\w\sigma$
then returns $\w_1$ and a tuple of semantic functions for reading,
membership, removal of which we only sketch reading here: If
$u:\w_1\rightarrow \w_2$ and $\sigma_1\in\Astores\w_1$ and
$i\in\mathbb{Z}$ then the reading function looks up $i$ in the data
structure starting at the entry points $E$ in $\sigma_1$. (Note that
$\sigma_1\in\Astores\w$ asserts that this data structure exists and is
well-formed.) The returned (abstract) store $\sigma_2$ might not be the
same as $\sigma$ because internal reorganizations, e.g., removal of
duplicates, might have occurred. However, no world extension is needed
and
$\sigma_1\sim\sigma_2$ holds. This together with the fact that
the outcome only depends on the $\loc^R$ equivalence class \iffull(in the
proof-relevant sense)\fi justifies a read-only typing for reading. 

\iffull
In order to model the object-oriented version with a basic type
$\textit{set}_\regid$ for each region $\regid$, we can use the p.p.f.\
$\textit{Set}_\regid\w=\{(\w_0,u,\loc) \mid \loc\in \w_0(r)$ and
$\loc$ is of the above format and $u:\w\rightarrow \w_0$
$\}$. Equality on $\textit{Set}_\regid\w$ is discrete and one puts
$v.(\w_0,u,\loc)=(\w_0,vu,\loc)$. The indirection through $\w_0$ is
necessary so that we recognize that a given abstract location is
indeed a ``set''.
\fi
\paragraph{Memoization} For the simple $\textit{memo}$ functional
from Section~\ref{examples} we produce just as in the previous example a
fresh
abstract location $\loc$ that contains the two newly allocated
concrete locations, say $\cloc_x,\cloc_y$, and on which we impose the
invariant $(\heap,\heap')\in \loc^R \iff $
$\heap(\loc_x),\heap'(\loc_x)$ contain the same integer value, say $i$
and that $\heap(\loc_y),\heap'(\loc_y)$ both contain the integer value
$f(i)$ where $f$ is the pure function to be memoised.\iffull We see in
Lemma~\ref{lemmasix} that if a function is semantically pure (empty
effect) then there is a world- and store-independent function
describing its action.\fi

\iffull
\subsubsection{Effect-dependent equivalences}
\else
\paragraph{Effect-dependent equivalences} Consider the following
notation
\fi

\iffull
Before we prove the soundness of a number of program equivalences, 
we introduce some notation. Given a set of effects $\eff$, we write 
$\rdsin{\eff} = \{\regid \mid \rEff\regid \in \eff\}$, $\wrsin{\eff} =
\{\regid \mid \wEff\regid \in \eff\}$,  $\alsin{\eff} = \{\regid \mid
\aEff\regid \in \eff\}$ and $\regs{\eff} = \rdsin{\eff} \cup \wrsin{\eff}
\cup \alsin{\eff}$. Moreover, the set $\nwrs(\eff) = \regs{\eff} \setminus
\wrsin{\eff}$. We also introduce the following piece of notation for
$\sigma, \sigma' \in \Astores \w$:
\[
\begin{array}{c}
\sigma \sim_{\rdsin{\eff, \w}} \sigma' \iff \forall \loc \in
\w(\rdsin{\eff}). (\sigma, \sigma') \in \loc^R \\
\sigma \sim_{\nwrs(\eff, \w)} \sigma' \iff \forall \loc \in
\w(\nwrs(\eff)). (\sigma, \sigma') \in \loc^R
\end{array}
\]
which specify that the abstract heaps $\sigma$ and $\sigma'$ are
\else 

\(
\begin{array}{c}
\sigma \sim_{\rdsin{\eff, \w}} \sigma' \iff \forall \loc \in
\w(\rdsin{\eff}). (\sigma, \sigma') \in \loc^R \\
\sigma \sim_{\nwrs(\eff, \w)} \sigma' \iff \forall \loc \in
\w(\nwrs(\eff)). (\sigma, \sigma') \in \loc^R
\end{array}
\)

\noindent
which specify that the abstract heaps $\sigma$ and $\sigma'$ are
\fi
equivalent on all the abstract locations $\loc$ in regions associated,
respectively, to read effects and no-writes in $\eff$.  

\begin{lemma}\label{lemmasix}
\label{lem:invariant}
Let $\Gamma \vdash e : \ety{\tau}{\eff}$. For any world $\w \in \world$,
and context $\gamma \in \sem{\Gamma}\w$, whenever $\sigma_0 , \sigma_0' \in
\Astores\w$ such that $\sigma_0
\sim_{\rdsin{\eff, \w}} \sigma_0'$, then $\cval(\sigma_0)$ and
$\cval(\sigma_0')$ where $\cval =\sem{\Gamma \vdash e :
\ety{\tau}{\eff}}\w(\gamma)$ are
equally defined  and if $\cval(\sigma_0)=(\w_1, u, \sigma_1,
\val)$ and 
$\cval(\sigma_0') = (\w_1', u', \sigma_1', \val')$
then there exist (continuously!) a pullback 
 $\w_1 \sq{x}{v}{x'}{v'} \w_1'$ with apex
$\overline{\w}$ and low point $\und{\w}$ and a proof of $x.\val\sim x'.\val'$ such that $xu=x'u'$ and the following is satisfied:
\begin{compactenum}
 \item for all $\loc \in \w$, we have either: $(\sigma_0,
\sigma_1.u) \in \loc^R$
and $(\sigma_0', \sigma_1'.u') \in \loc^R$ (remain
equivalent) or
$(\sigma_1.u, \sigma_1'.u') \in \loc^R$ (equally modified);
 \item if $\loc \in \w(\nwrs(\eff))$, then $(\sigma_0,
\sigma_1.u) \in \loc^R$ and $(\sigma_0',
\sigma_1'.u') \in \loc^R$.

\item There exists a morphism $\cval' \in \sem{\Gamma}\to T_\eff
\sem{\tau}$, such that  $\cval' \sim \cval$ and if
$\cval'(\w)(\gamma)\sigma_0 = (\w_\star, u_\star, \sigma_\star,
\val_\star)$, then for all regions $\regid \notin \alsin{\eff}$,
$\w_\star(\regid) = \w(\regid)$. 
\end{compactenum}
\end{lemma}
\iffull\begin{proof}
The proof that the values are equal in $\und{\w}$ follows directly from 
the definition of computations and effects.

For the first part,  we use the following relation $R$ defined for all
worlds $\w_1$, such that $u : \w \to \w_1$:

\(
\begin{array}{l}
 \{(\sigma, \sigma') \mid \sigma \sim_{\rdsin{\eff, \w}} \sigma' 
\land \forall \loc \in \w.\\
 \qquad (\sigma.u, \sigma_0) \in \loc^R \land
(\sigma'.u, \sigma_0') \in \loc^R \lor (\sigma.u, \sigma'.u) \in
\loc^R\}
\end{array}
\)

\noindent
Otherwise, for the worlds $\w_2$ not reachable from \w, the relation
$R\w_2$ is the trivial set. Notice that $R \in \Rscr(\eff)$ and it is
contravariant. The claim then 
follows directly.

The proof of the second part follows in a similar fashion, but we use the
following relation:

\(
\{(\sigma, \sigma') \mid \sigma \sim_{\rdsin{\eff, \w}} \sigma' 
\land \sigma \sim_{\nwrs(\eff, \w)} \sigma_0.u \}
\)

\noindent
And we use a similar relation for showing that $\sigma_0'$ and
$\sigma_1'.u'$ agree on the not written locations $\nwrs(\eff, \w)$.

For the third property, first, we show that there is an
isomorphism between $\w(\regid)$ and $\und{\w}(\regid)$ for all regions
$\regid \notin \alsin{\regid}$ by using the following relation:

\(
\{(\sigma, \sigma') \mid \sigma \sim \sigma' \land \forall
\regid \notin \alsin{\eff}. \#_\regid(\sigma), \#_\regid(\sigma')
\leq \#_\regid(\w) \} 
\)

\noindent
where $\#_\regid$ denotes the number of abstract locations coloured with 
$\regid$.
Clearly, $R \in
\Rscr(\eff)$ as $\eff$ does not contain any allocation effects. This gives
us one direction, while the other direction is obtained by the universal
property of pullbacks. Given this property, one can easily construct the 
function $\cval'$.
\end{proof}
\begin{proposition}\label{compr}
(commuting computations) Suppose that:
$
 \Gamma \vdash e_1 : \ety{\tau_1}{\eff_1}$ and
$\Gamma \vdash e_2 :
\ety{\tau_2}{\eff_2}
$,
where  $\rdsin{\eff_1} \cap \wrsin{\eff_2} = \rdsin{\eff_2} \cap
\wrsin{\eff_1}= 
\wrsin{\eff_1} \cap \wrsin{\eff_2} = \emptyset$. Let 
\[
 \begin{array}{l}
  e = \letin{x}{e_1} \letin{y}{e_2} (x, y) \quad \textrm{and} \quad
  e' = \letin{y}{e_2} \letin{x}{e_1} (x, y)
 \end{array}
\]
then $\sem{\Gamma \vdash e : \ety{ \tau_1 \times \tau_2 }{\eff_1\cup
\eff_2}} \sim \sem{\Gamma \vdash e' : \ety{ \tau_1 \times \tau_2
}{\eff_1\cup \eff_2}}$.
\end{proposition}

\begin{proof}
Assume a world \w and a context $\gamma \in \sem{\Gamma}\w$.
Let $\cval_i = \sem{\Gamma \vdash e_i : \ety{\tau_i}{\eff_i}}$ for
$i = 1,2$. 

It is enough to assume a pullback 
$\w \sq{1}{1}{1}{1} \w$, and an abstract heap $\sigma_0 \in
\Astores \w$. 
Assume that these functions are defined as follows:

\(
 \begin{array}{l}
 \cval_1(\w)(\gamma)\sigma_0 = (\w \uplus \w_1, u_1, \sigma_1, \val_1)\\
 \cval_2(\w\uplus\w_1)(u_1.\gamma) \sigma_1 = (\w \uplus \w_1 \uplus
\w_2, u_2, \sigma_2, \val_2)\\
 \cval_2'(\w)(\gamma)\sigma_0 = (\w \uplus \w_1', u_1', \sigma_1',
\val_1')\\
 \cval_1'(\w\uplus\w_2')(u_1'.\gamma) \sigma_1' = (\w \uplus \w_1'
\uplus \w_2', u_2', \sigma_2',\val_2')
 \end{array}
\)

\noindent
One can easily show that when one of these functions is undefined, then 
the corresponding function is also undefined. 

We need to show that there is a proof $\w \uplus \w_1 \uplus \w_2
\sq{x}{v}{x'}{v'} \w \uplus \w_1' \uplus \w_2'$ such that
$p : \sigma_2.v \sim \sigma_2'.v'$ and $p_1 : xu_2.\val_1 \sim x'.\val_2'$
and $p_2 : x.\val_2 \sim x'u_2'.\val_1'$.
Decompose $\w = \w_0 \uplus \q_1 \uplus \q_2$, where
$\w(\wrsin{\eff_i}) \subseteq \q_i$. 
The existence of such decomposition follows from the disjointness of write
effects in $\eff_1$ and $\eff_2$.

From Lemma~\ref{lem:invariant} and from 
the disjointness of reads and writes, it is the case that $\sigma_0$ and
$\sigma_1'$ agree on the locations in $\w_0 \uplus \q_1$.
That is, there is a proof $p : \sigma_0.1 \sim \sigma_1'.x_1$, defined 
using the proof $\w_0 \uplus \q_1 \sq{x_1}{1}{1}{x_1} \w_0
\uplus \q_1 \uplus \w_2'
$, where $x_1 : \w_0 \uplus \q_1 \to \w_0 \uplus \q_1 \uplus \w_2'$.
Applying $(e_1)_1$ to the objects above,
we get the pullback $\w_0 \uplus \q_1 \uplus \w_1 \sq{x_2}{v_2}{x_2'}{v_2'}
\w_0 \uplus \q_1
\uplus
\w_2' \uplus \w_1'$, and  proof $q : x_2.\val_1 \sim
x_2'.\val_2'$. Symmetrically, we obtain the proofs 
$\w_0 \uplus \q_2 \uplus \w_2 \sq{x_3}{v_3}{x_3'}{v_3'} \w_0 \uplus \q_2
\uplus \w_1' \uplus \w_2'$, and  $q' : x_3.\val_2 \sim
x_3'.\val_1'$. Hence, there is also a proof in the larger world $cod(x)$.

To see informally that the final heaps $\sigma_2$ and $\sigma_2'$ are
equal, we use
the following facts obtained using Lemma~\ref{lem:invariant}: $\sigma_2$
and $\sigma_1$ agree on the locations in $\w_0 \uplus \q_1$; moreover, 
$\sigma_2'$ and $\sigma_1$ agree on the locations in $\w_0 \uplus \q_1$; 
hence $\sigma_2$ and $\sigma_2'$ agree on the locations in $\w_0 \uplus
\q_1$. Symmetrically, we can also argue that $\sigma_2$ and $\sigma_2'$
agree on the locations in $\w_0 \uplus \q_2$. Composing these proofs 
(see comment after Lemma~\ref{lem:unique-upto-sim} why this is allowed), 
we get that $\sigma_2$ and $\sigma_2'$ agree on the locations in $\w$. 
Finally, since the locations allocated by one computation are not used by
the other computation, the final heaps are equal at the
apex world.
%
%
\end{proof}

The following propositions are also provable. All propositions are proved
in a similar way as the soundness proof of the commuting case, using 
Lemma~\ref{lem:invariant} when needed. For instance, the soundness proof
of the duplicated computation uses the third case in
Lemma~\ref{lem:invariant}.

\begin{proposition}[dead computation]\label{dead}
Suppose that
$\Gamma \vdash e : \ety{\unittype}{\eff}$,
that \iffull $\wrsin{\eff}=\emptyset$ \else 
$\eff$ contains no writes \fi and that $\sem{\Gamma\vdash
e: \ety{\unittype}{\eff}}\w(\gamma)(\sigma)$ is defined for all
$\w,\gamma\in\sem{\Gamma}\w, \sigma\in\Astores\w$. 
Then if for all worlds \w, all contexts $\gamma \in
\sem{\Gamma}\w$, and abstract heaps $\sigma \in \Astores\w$,
the function $\sem{\Gamma \vdash e}(\w)(\gamma)(\sigma)$ is defined, 
then $\sem{\Gamma \vdash e : \ety{\unittype}{\eff}} \sim \sem{\Gamma
\vdash \unitval : \ety{\unittype}{\eff}}$.
\end{proposition}

\begin{proof}
Assume a world \w and a context $\gamma \in \sem{\Gamma}\w$.
Let $\cval = \sem{\Gamma \vdash e : \ety{\tau}{\eff}}$. 
It is enough to assume a pullback 
$\w \sq{1}{1}{1}{1} \w$, and an abstract heap $\sigma_0 \in
\Astores \w$. Let $\cval(\w)(\gamma)\sigma_0  =  (\w, 1, \sigma_1,
\val_1)$. We need to construct a pullback such that $\val_1$ is equivalent
to $\unitval$ in its apex and $\sigma_1$ is equivalent to $\sigma_0$ in
its low point. Consider the pullback $\w_1 \sq{1}{u}{u}{1} \w$. Clearly
$\val_1 = \unitval$, and therefore the values are equivalent in $\w_1$.
Moreover, from the fact that $\wrsin{\eff} =
\emptyset$, $\sigma_1$ and $\sigma_0$ agree on all locations in $\w$.
Hence, $\sigma_1.u \sim \sigma_0$, which finishes the proof. 
\end{proof}

\begin{proposition}[duplicated computation]\label{dupl}
  Suppose that $\Gamma \vdash e:\ety{\tau}{\eff}$ and suppose
that
  $\reads(\eff)\cap\writes(\eff)=\allocs(\eff)=\emptyset$. Thus, $e$
  reads and writes on disjoint portions of the store and makes no
  allocations. The the terms  $e_1$ and $e_2$ below
\[\begin{array}{l}
\letin{x}{e}{(x,x)} ~\textrm{ and }~ \letin{x}{e} \letin{y}{e}{(x,y)}
\end{array}
\]
 are contextually equivalent. That is
formally $\sem{\Gamma \vdash e_1 : \ety{\tau \times \tau}{\eff}} \sim
\sem{\Gamma \vdash e_2 : \ety{\tau \times \tau}{\eff}}$.
\end{proposition}

\begin{proof}
Assume a world \w and a context $\gamma \in \sem{\Gamma}\w$.
Let $\cval = \sem{\Gamma \vdash e : \ety{\tau}{\eff}}$. 
It is enough to assume a pullback 
$\w \sq{1}{1}{1}{1} \w$, and an abstract heap $\sigma_0 \in
\Astores \w$. 
From Lemma~\ref{lemmasix} and since these functions do not allocate, we
can assume that they do not cause any world extension and are therefore
defined as follows:
\[
 \begin{array}{lcl}
 \cval(\w)(\gamma)\sigma_0  =  (\w, 1, \sigma_1, \val_1) &
\textrm{and}&
  \cval(\w)(\gamma)\sigma_1  =  (\w, 1, \sigma_2, \val_2).
 \end{array}
\]
We need to show that the values $\val_1$ and $\val_2$ are equivalent and 
the heaps $\sigma_1$, obtained by applying once $e$, and $\sigma_2$,
obtained by applying twice $e$, are also equal. 

Decompose $\w = \w_0 \uplus \w_r \uplus \w_w $, where $\w_r$ contains all
the regions read by $e$ and $\w_w$ all the regions written by $e$. This is
possible because of the disjointness of of $e$'s read and write effects. 
From Lemma~\ref{lemmasix} and the disjointness of  $e$'s read and write
effects, we have that $\sigma_0$ and $\sigma_1$ agree on the regions read
by $e$, that is, $\sigma_0 \sim_{\rdsin{\eff, \w}} \sigma_1$. Hence, again
from Lemma~\ref{lemmasix}, we have that the values $\val_1$ and $\val_2$
are equal. Moreover, the locations in $\w_w$ are equaly written, while the
locations in $\w_0 \uplus \w_r$ are left unchanged, that is,
$\sigma_1$ and $\sigma_2$ agree on the location in $\w$.
\end{proof}

\begin{proposition}[pure lambda hoist]\label{hoist}
Suppose that $
\Gamma\vdash e: \ety{Z}{\emptyset}$ and 
$\Gamma,x{:}X,y{:}Z\vdash e': \ety{Y}{\eff}$
Let $e_1$ and $e_2$ be respectively $
\lambda x.\letin{y}{e} {e'}$ and
$\letin{y}{e}{\lambda x.e'}$. 
Then $\sem{\Gamma\vdash  e_1 :
\ety{(\valty{X\effto\eff
Y})}{\emptyset}} \sim 
\sem{\Gamma\vdash  e_2 : \ety{(\valty{X\effto\eff
Y})}{\emptyset}}$.
\end{proposition}

\begin{proof}
Assume a world \w and a context $\gamma \in \sem{\Gamma}\w$.
Let $\cval = \sem{\Gamma \vdash e : \ety{\tau}{\eff}}$ and 
$\cval' = \sem{\Gamma, x: X, y: Z \vdash e' : \ety{\tau}{\eff}}$. 
It is enough to assume a pullback 
$\w \sq{1}{1}{1}{1} \w$, and an abstract heap $\sigma_0 \in
\Astores \w$. Since $e$ has no effects, we have no world extension:
\[
 \cval(\w)(\gamma)\sigma_0  =  (\w, 1, \sigma_1', \val_1') 
\]
Moreover, from Lemma~\ref{lemmasix}, $\sigma_1$ and $\sigma_0$ agree on
all locations. We now show that 
\[
\sem{\Gamma\vdash  \lambda x.\letin{y}{e} {e'(x,y)} : (\valty{X\effto\eff
Y})}
\sim
\sem{\Gamma\vdash \lambda x.e'(x,\val_1') : (\valty{X\effto\eff
Y})}
\]
In order to prove this, assume a morphism $v : \w \to \w_1$ and $a \in
\sem{X}\w_1$. We need then to prove that the computations resulting from 
applying $a$ to the functions above are equivalent in the pullback
$\w_1\sq{1}{1}{1}{1}\w_1$. For this, assume an abstract heap $\sigma \in
\Astores \w_1$. Since $e$ has no effect, we have no world extension:
\[
\begin{array}{l}
 \cval(\w_1)(\gamma)\sigma  =  (\w_1, 1, \sigma_1, \val_1) \\
 \cval'(\w_1)(\gamma, a, \val_1 )\sigma_1  =  (\w_2, 1, \sigma_2, \val_2)\\
\cval'(\w_1)(\gamma, a, \val_1' )\sigma  =  (\w_2', 1, \sigma_1', \val_2)
\end{array}
\]
Since $e$ is pure, we have $\val_1 = v.\val_1'$ and from
Lemma~\ref{lemmasix} we have that $\sigma_1$ and $\sigma$ agree on all
locations in $\w_1$ and in particular on locations read by $e'$. Hence, 
again by Lemma~\ref{lemmasix} the pullback proof exists where $\sigma_2$
and 
$\sigma_1'$ are equal in its low point and the resulting values are equal
in its apex.
\end{proof}
\else
We can validate all the effect-dependent program equivalences ``dead,
commuting, duplicated computation'' and ``pure lambda hoist'', as well
as the ``masking rule'' from previous work \cite{DBLP:conf/aplas/BentonKHB06} in our new, more powerful,
setting.
To give an impression of the
formulation of these validations we state the corresponding proposition for
``dead computation'' which is particularly interesting in that it contains
a termination precondition. The proof, and details of the other
equations are in the 
\ifapp
appendix,
\fi
which also 
contains a validation of loop unrolling optimisation 
described by Tristan and Leroy \cite{tristan10popl}. 
 \begin{proposition}[dead computation]\label{dead}
Suppose that
$\Gamma \vdash e : \ety{\unittype}{\eff}$,
that $\wrsin{\eff}=\emptyset$ and that $\sem{\Gamma\vdash
e: \ety{\unittype}{\eff}}\w(\gamma)(\sigma)$ is defined for all
$\w,\gamma\in\sem{\Gamma}\w, \sigma\in\Astores\w$. 
Then if for all worlds \w, all contexts $\gamma \in
\sem{\Gamma}\w$, and abstract heaps $\sigma \in \Astores\w$,
the function $\sem{\Gamma \vdash e}(\w)(\gamma)(\sigma)$ is defined, 
then $\sem{\Gamma \vdash e : \ety{\unittype}{\eff}} \sim \sem{\Gamma
\vdash \unitval : \ety{\unittype}{\eff}}$.
\end{proposition}
 \fi
\iffull
\begin{figure*}[t]
\begin{center}
\begin{small}
\begin{tabular}{l@{\quad}l@{\quad}l@{\quad}l}
Body of Loop & Prolog & Steady Program & Epilogue\\
\texttt{x := load(p);} & \texttt{p1 := p;} &
\texttt{store(p1, y);}
\textbf{[$\wEff{\regid_1}$]} & \texttt{store(p1, y);}
\textbf{[$\wEff{\regid_1}$]}\\ 

\texttt{y := x * c;} & \texttt{p2 := p;} & \texttt{p1 := p2 + 8;} &
\texttt{y := x2 * c;} \\

\texttt{store(p, y);} & \texttt{x1 := x;} & \texttt{y := x2 * c;}
& \texttt{store(p2, y);} \textbf{[$\wEff{\regid_2}$]}\\

\texttt{p := p + 8;} & \texttt{x2 := x;} & \texttt{x1 := load(p1);
\textbf{[$\rEff{\regid_1}$]}} & \texttt{x := x2;}\\

\texttt{i := i + 1;} & \texttt{x1 := load(p1);}
\textbf{[$\rEff{\regid_1}$]} &
\texttt{store(p2, y);} \textbf{[$\wEff{\regid_2}$]}
& \texttt{p := p2;}\\

& \texttt{p2 := p1 + 8;} & \texttt{p2 = p1 + 8;}\\
& \texttt{x2 := load(p2);} \textbf{[$\rEff{\regid_2}$]} &
\texttt{y
= x1 * c;}\\

& \texttt{y := x1 * c;} & \texttt{y = load(p2);
\textbf{[$\rEff{\regid_2}$]}}\\

& \texttt{i := i + 2;} & \texttt{i := i + 2;}\\
\end{tabular}
\end{small}
\vspace{-2mm}
\end{center}
\caption{Program obtained from the loop unrolling technique. Here
\texttt{p}, \texttt{p1} and \texttt{p2} are pointers and all \texttt{load}
and \texttt{store} operations are on 64 bit numbers (float).}
\label{fig:loop-unrolling}
\vspace{-4mm}
\end{figure*}

\subsection{Loop Unrolling}
Loop unrolling is a software pipelining technique used to enhance the use
of parallel processing. The idea is instead of iterating a loop 
in a sequential manner, one attempts to process a number of iterations
of the loop at the same time using multiple processors.

As described in \cite{tristan10popl} implementing and proving the
correctness of loop unrolling techniques is hard as
one needs to demonstrate that the program resulting from loop unrolling
that can be executed in parallel is equivalent to the original sequential
program. We briefly illustrate the power of our system with regions and
effects by one of the running examples in \cite{tristan10popl}. Consider
a loop program whose body is depicted in Figure~\ref{fig:loop-unrolling}.
Intuitively, this program is multiplying all the elements of an array of 
float values by the value \texttt{c}. Clearly, instead of executing this
program sequentially, we can execute different iterations in parallel. 
In particular, after applying the loop unrolling optimization to a program,
one obtains a program that is divided in three parts: the prolog, that
initializes all the variables, the steady state, that is iterated, and the
epilogue, that is executed when the loop condition is no longer true and
the loop is over. Figure~\ref{fig:loop-unrolling} contains the program
obtained by loop unrolling two iterations of the program above. The Prolog
and the Epilogue are executed at the beginning and the end, respectively,
while the Steady Program may be executed several times. 

The task is to show that the optimized program is equivalent to the
sequential program above. Using the unrolling equations from
Lemmas~\ref{fixlse} we can unroll the loop twice
($n=2$) and extract a prologue. We can then conclude with
effect-dependent equivalences, in particular Prop.~\ref{compr} as
follows. We use two regions $\regid_1$
and $\regid_2$. All even elements of the array, that is, \texttt{p},
\texttt{p + 16}, \texttt{p + 32}, \ldots, belong to the region $\regid_1$,
while all odd elements, that is, \texttt{p + 8},
\texttt{p + 24}, \texttt{p + 40}, \ldots, belong to the region $\regid_2$.
Given this setting, the read and write effects are as shown in
Figure~\ref{fig:loop-unrolling}. It is now a simple exercise to show that
any execution of the optimized program is equivalent to an execution of
the sequential program. For instance, any instruction with a read effect 
on $\regid_1$ can be permuted so that it appears immediately before the
following instruction with write effect $\regid_1$ on the same region
$\regid_1$. This is possible because the only effect between these two
instructions is a read on the other region $\regid_2$. The same is true
for permuting instructions that read on $\regid_2$. 

\fi
\subsection{State Dependent Abstract Data Types (ADT)}
We prove the equivalence of a number of programs involving state dependent
abstract data types. 

\paragraph{Awkward Example}
 The first example is Pitts and Stark's classic \emph{awkward}
example\cite{pitts98high}. Consider the
following two programs:
\iffull
\[
 e_1 = \letin{x}{\myref{0}} \lambda f. x:= 1; f(); !x
\qquad\textrm{and}\qquad
 e_2 = \lambda f. f(); 1.
\]
Intuitively, the expressions $e_1$ and $e_2$ are equivalent as they both
\else

\(
 e_1 = \letin{x}{\myref{0}} \lambda f. x:= 1; f(); !x
\qquad\textrm{and}\qquad
 e_2 = \lambda f. f(); 1.
\)

\noindent
Intuitively, the expressions $e_1$ and $e_2$ are equivalent as they both
\fi
return the value $1$, although $e_1$ uses a fresh location to do so.
We can formally prove the equivalence of these functions as follows:
Assign the region where $x$ is allocated as $\regid$. If $f$ has the type
$\unittype \effto{\eff} \unittype$ with effects $\eff$, then $e_1$ has 
type $\ety{(\unittype \effto{\eff} \unittype) \effto{\eff, \rEff{\regid},
\wEff{\regid}}
\inttype}{\eff, \aEff{\regid}}$, while 
$e_2$ has type $\ety{(\unittype \effto{\eff} \unittype) \effto{\eff}
\inttype}{\eff}$. Notice that  $\eff$ may contain
$\rEff{\regid}$
or $\wEff{\regid}$ or both. Moreover,
assume that the footprint of a location in
region \regid\ consists of a single concrete location \cloc, and that the
guarantee of a location $\cloc^G$ consist of a single function
$\mathit{write}_1$ such that $\mathit{write}_1(\heap) = \heap'$ where
$\heap'(\cloc) = 1$ and $\heap'(\cloc') = \heap(\cloc')$ for all other
locations. Clearly $e_1$ has such a write effect.

For proving the equivalence of $e_1$ and $e_2$, assume a world $\w$ and an
abstract heap $\sigma$. Let $\sem{e_1}\w\sigma = (\w \uplus \w_1 \uplus
\w_r, u_1, \val_1, \sigma_1)$ and $\sem{e_2}\w\sigma = (\w \uplus \w_1,
u_1, \val_2, \sigma_2)$. We need to construct a pullback square $\w \uplus
\w_1 \uplus \w_r\sq{}{}{}{}\w \uplus \w_1$ such that the values $\val_1$
and $\val_2$ are equal in its apex and $\sigma_1$ and $\sigma_2$ are equal
in its low point. Since $\wEff \regid$ is in the effects of $e_1$, we
have that $\val_1 = 1$. We also have $\val_2 = 1$ trivially. Hence
$\val_1$ and $\val_2$ are equal in the apex of the pullback square $\w
\uplus
\w_1 \uplus \w_r\sq{}{}{}{}\w \uplus \w_1$. Similarly, $\sigma_1$ when
taken to the low point of the square, that is, where the locations in
$\w_r$ are forgotten, the resulting heap is equivalent to $\sigma_2$.

\paragraph{Modified Awkward Example} 
Consider now the following variant of the Awkward example, due to
Dreyer et al.\cite{dreyer10icfp}:
\iffull
\[
 e_1 = \letin{x}{\myref{0}} \lambda f. x:= 0; f(); x:= 1; f(); !x
~\textrm{ and }~
 e_2 = \lambda f. f();f(); 1.
\]
The difference is that in the first program $x$ is written to $0$ and the
\else

\(
 e_1 = \letin{x}{\myref{0}} \lambda f. x:= 0; f(); x:= 1; f(); !x
~\textrm{ and }~
 e_2 = \lambda f. f();f(); 1.
\)

\noindent
The difference is that in the first program $x$ is written to $0$ and the
\fi
call-back function is used twice. Interestingly, however, the solution
given for the Awkward example works just fine. We can prove semantically
that the type of the program $e_1$ has the same type as before in the
Awkward example, where the only writes allowed on abstract location
assigned for $x$ is to write one.
Therefore, if $f$ has effect of writing on the region $\regid$, it will set
$x$ to one.

\paragraph{Callback with Lock Example}
We now show equivalence of the following programs, also due to Dreyer
et al.\cite{dreyer10icfp}:
\iffull
\[
\begin{array}{l@{\quad}l}
 e_1 = \letin{b}{\myref{\mtrue}} \letin{x}{\myref{0}}
& e_2 = \letin{b}{\myref{\mtrue}} \letin{x}{\myref{0}}\\
\quad \langle \lambda f. \keywd{if}~{!b}~\keywd{then}
& \quad \langle \lambda f. \keywd{if}~{!b}~\keywd{then}\\
\quad\quad({b:= \mfalse};
f(); x := !x + 1; b:= \mtrue)
&
\quad\quad({b:= \mfalse};
\letin{n}{!x} f(); 
\\
\quad ~\keywd{else}~(), \lambda\_.!x\rangle
& \quad\quad x := n + 1; b:= \mtrue)
\\& \quad ~\keywd{else}~(), \lambda\_.!x\rangle
\end{array}
\]
\else

\(
\hspace{-9mm}
\begin{array}{l@{\quad}l}
 e_1 = \letin{b}{\myref{\mtrue}} \letin{x}{\myref{0}}
& e_2 = \letin{b}{\myref{\mtrue}} \letin{x}{\myref{0}}\\
\quad \langle \lambda f. \keywd{if}~{!b}~\keywd{then}
& \quad \langle \lambda f. \keywd{if}~{!b}~\keywd{then}\\
\quad\quad({b:= \mfalse};
f(); x := !x + 1; b:= \mtrue)
&
\quad\quad({b:= \mfalse};
\letin{n}{!x} f(); 
\\
\quad ~\keywd{else}~(), \lambda\_.!x\rangle
& \quad\quad x := n + 1; b:= \mtrue)
\\& \quad ~\keywd{else}~(), \lambda\_.!x\rangle
\end{array}
\)

\noindent
\fi
Both programs produce a pair of functions, one incrementing the
value stored in $x$ and the second returning the value stored in $x$.
The boolean reference $b$ serves as lock in the
incrementing function. 
Once this function is called the value in $b$ is set to
$\mfalse$ and only after calling the call-back, the value in
$x$ is incremented is $b$ set again to $\mtrue$. However, the
implementation of 
the increment function is different. While the program to the 
left calls the call-back function $f()$ and then increments the value of
$x$ using the value stored in $x$, the program to the right remembers (in
$n$) the value of $x$ before the call-back is called and then uses it to
increment the value stored in $x$. 

Assume that $x$ and $b$ are in the footprint of the same abstract location
($\loc$) in the
region $\regid$. We show that these programs are equivalent under the type
\iffull
\[
\ety{(\unittype \effto{\eff} \unittype) \effto{\eff, \wEff{\regid},
\rEff{\regid}} \unittype) \times (\unittype \effto{\rEff{\regid}}
\unittype)}{\aEff{\regid}, \eff},
\]
 where $\eff$ may contain the effects
\else

\(
\ety{(\unittype \effto{\eff} \unittype) \effto{\eff, \wEff{\regid},
\rEff{\regid}} \unittype) \times (\unittype \effto{\rEff{\regid}}
\unittype)}{\aEff{\regid}, \eff},
\)

\noindent
 where $\eff$ may contain the effects
\fi
$\wEff{\regid}, \rEff{\regid}$. In particular, the location $\loc$ is
specified as follows: its footprint consists only of the concrete
locations storing $x$ and $b$, written $\cloc_b$ and $\cloc_x$, while its
rely-condition is equality. The more interesting is its guarantee condition
($\loc^G$), which contains 
the following idempotent functions $f_i$ for $i \in \mathbb{N}$:
$f_i(\heap) = \heap$ if $\heap(\cloc_b) = \mfalse$ and $f_i(\heap) =
\heap'$ if $\heap(\cloc_b) = \mtrue$, where $\heap'(\cloc_x) = i$ if
$\heap(\cloc_x) \leq i$ and $\heap'(\cloc_x) = \heap(\cloc_x)$; 
moreover, the value of $b$ is unchanged, that is, $\heap'(\cloc_b) =
\heap(\cloc_b)$.
It is easy to check that these functions are idempotent as well as their
composition. 

First, notice that indeed the two functions above have type $\wEff{\regid}$
as the increment of $x$ is captured by using some write function
$f_i$
and moreover $b$ is $\mtrue$. Now, to show that the two programs above
are equivalent, we need to show that the value stored in $x$ before and
after the call back is called is the same. This is the case, as even if
$\wEff{\regid} \in \eff$, the value stored in $b$ is $\mfalse$, which means
that any function $f_i$ used will leave the concrete locations storing $x$
and $b$ untouched.

Notice that if the read function also called the call-back, then the
reasoning above would break, as the call-back could modify the value
stored in $x$ because $b$ is $\mtrue$.

\section{Conclusions}
\label{sec:conclusions}
We have laid out the basic theory of proof-relevant logical relations
and shown how they can be used to justify nontrivial effect-dependent
program equivalences. We have also shown that proof-relevant logical
relations give direct-style justifications of the Pitts-Stark-Shinwell
equivalences for name generation. For the first time it was possible
to combine effect-dependent program equivalences with hidden
invariants allowing ``silent modifications'' that do not count towards
the ascription of an effect.
Earlier accounts of effect-dependent program equivalences
\cite{DBLP:conf/popl/KammarP12,DBLP:conf/ppdp/BentonKBH09,DBLP:conf/ppdp/BentonKBH07,DBLP:conf/aplas/BentonKHB06,DBLP:conf/icfp/ThamsborgB11}
do not provide such possibilities.

Proof-relevant logical relations or rather the sets $|A\w|$ where $A$
is a semantic type bear a vague relationship with the \emph{model
  variables} \cite{DBLP:journals/spe/CheonLSE05} from ``design by
contract'' \cite{DBLP:journals/computer/Meyer92} and more generally
data refinement \cite{DBLP:books/cu/RoeverE1998}. The commonality is
that we track the semantic behavior of a program part with abstract
functions on some abstracted set of data that may contain additional
information (the ``model''). The difference is that we do not focus on
particular proof methods or specification formalisms but that we
provide a general, sound semantic model for observational equivalence
and program transformation and not merely for functional
correctness. This is possible by the additional, also proof-relevant
part of the semantic equality proofs between the elements of the
models. We also note that our account rigorously supports higher-order
functions, recursion, and dynamic allocation. 

Our abstract locations draw upon several ideas from separation logic
\cite{DBLP:conf/lics/Reynolds02}, in particular footprints and the
conditions on rely/guarantee assumptions from
\cite{DBLP:conf/concur/VafeiadisP07}. Intriguingly, we did not need
something resembling the ``frame rule'' although perhaps the
$\Pi$-quantification over larger worlds in function spaces plays its
role.

Pullback-preserving functors and especially the instantiation
\emph{sets of locations} are inspired by FM-sets 
\cite{DBLP:journals/fac/GabbayP02} or rather the \emph{Schanuel topos}
to which they are equivalent (see Staton \cite{statonphd} for a
comprehensive account). The instantiations other than \emph{sets of
  locations}, as well as the use of setoids for the ``values'' of
these functors rather than plain sets is original to this work.

\iffull
For a while we worked with groupoids instead of setoids.  At some
point though we found that the required equations between proofs are
needed only to establish other such equations, never in order to prove
something that does not mention them and we thus took the somewhat
bold step to give up all equations between proofs leading us to
setoids. Should we ever want to model dependent types where the index
types are already nontrivial setoids then the passage to groupoids would
become necessary for then, by way of substitution, equality proofs
interfere with actual (semantic) terms. E.g.\ if $p:e\sim e'$ then we
expect a function $A[p]$ from $A[e]$ to $A[e']$.
\fi

We would like to have a semi-formal format that allows one to
integrate semantic arguments with typing and equality derivations more
smoothly. We would also like to allow proof-relevant partial
equivalences in the Heap PER instantiation, which essentially amounts
to the ability to store values with proof-relevant equality. In
particular, this would allow us to model higher-order store with some
layering policy  \cite{DBLP:journals/iandc/Boudol10}. For
unrestricted higher-order store as in
\cite{DBLP:conf/icfp/ThamsborgB11}, but with abstract locations, one
would need to overcome the well-known difficulties with circular
definition of worlds. Step-indexing \cite{DBLP:conf/esop/Ahmed06} is
an option, but we would prefer a domain-theoretic solution. The formal
similarity of our abstract locations with the rely-guarantee formalism
\cite{DBLP:journals/logcom/ColemanJ07,DBLP:conf/concur/VafeiadisP07}
suggests the intriguing possibility of an extension to concurrency.

We also believe that update operations governed by finite state
machines \cite{DBLP:conf/popl/AhmedDR09} can
be modelled as an instance of our framework and thus combined with
effect-dependency. The application of our general framework to effects other
than reading, writing, allocation deserves further investigation.

Indeed, we feel that with the transition to proof-relevance we have
opened a door to a whole new world that hopefully others will
investigate with us. 
\iffull
\bibliography{bib}
\else

\fi

\ifapp
\newpage
\appendix
\section{Online Appendix}
This appendix contains some additional technical material that was
omitted from the main body for space reasons. In particular,
Section~\ref{sec:app-syntax} contains standard details on semantics of
values and computations as well as of domain theory.
Section~\ref{sec:app-setoids} elaborates the Setoids theory, introducing
the definition of Isomorphic pullbacks and contains more properties of
p.p.f. In Section~\ref{sec:app-comp-model}, a third instantiation, more
complex than the sets of locations, but simpler than Heap PERs can be
found. Section~\ref{sec:app-logical-relations} contains most of the
machinery necessary to establish the Fundamental Theorem. Finally,
Section~\ref{sec:app-applications} contains further applications of our
setting. For instance, we prove the soundness of a number of re-writes,
such as the communting equation, duplication elimination, pure
lambda-hoist, etc. We also prove the soundness of the Masking rule and
discuss the loop-unrolling example in \cite{tristan10popl}.

\subsection{Syntax and Semantics}
\label{sec:app-syntax}
\mbox{}

\paragraph{Predomains}
A \emph{predomain} is an $\omega$-cpo, i.e. a partial order with suprema
of ascending chains. A \emph{domain} is a predomain with a least element,
$\bot$. \squelch{We use predomains and continuous functions, rather than sets
and functions, so as to be able to interpret recursive definitions.}
Recall that $f:A\rightarrow A'$ is \emph{continuous} if it is monotone
$x\leq
y \Rightarrow f(x)\leq f(y)$ and preserves suprema of ascending
chains, i.e., $f(\sup_i x_i)=sup_i f(x_i)$. Any set is a predomain with the
discrete
order. If $X$ is a set and $A$ a predomain then
any $f:X\rightarrow A$ is continuous. A subset $U$ of a
predomain $A$ is \emph{admissible} if whenever $(a_i)_i$ is an ascending
chain in $A$ such that $a_i\in U$ for all $i$, then $\sup_i a_i\in U$,
too. If $f:X\times A\rightarrow A$ is continuous and $A$ is a domain
then one defines $f^\dagger(x)=\sup_if_x^i(\bot)$ with
$f_x(a)=f(x,a)$. One has, $f(x,f^\dagger(x))=f^\dagger(x)$ and if
$U\subseteq A$ is admissible and $f:X\times U\rightarrow U$ then
$f^\dagger:X\rightarrow U$, too. We denote a partial (continuous)
function from set (predomain) $A$ to set (predomain) $B$ by $f:A
\partfun B$.  

\fulltrue
\paragraph{Semantics}
The untyped semantics of values and computations is given by
the recursive clauses in Figure~\ref{seme}; note the overloading of
semantic brackets for constants, values and computations. 
The notation
$\eta(x)$ stands for the $i$-th projection from $\eta\in\Values$
if $x$ is $x_i$ and $\eta[x{\mapsto}v]$ (functionally) updates the
$i$-th slot in $\eta$ when $x=x_i$.
\begin{figure*}[tph]
\iffull\[
\begin{array}{rcl}
\semV{x}\eta &=& \eta(x)\\ 
\semV{c}\eta &=& \semV{c}\\
\semV{(v_1,v_2)}\eta &=& (\semV{v_1}\eta,\semV{v_2}\eta)\\
\semV{v.i}\eta &=& d_i\ \mbox{if $i=1,2$, $\semV{v}\eta = (d_1,d_2)$}\\
\semV{\vfix{f}{x}t}\eta &=& \funn
{g^\dagger\,\eta}{\mbox{, where
  $g(\eta,u)=\lambda d.\semV{t}\eta[f{\mapsto}\funn{u},x{\mapsto}d]$}}
\end{array}
\]
\[
\begin{array}{rcll}
\semV{v}\eta\ \heap &=&(\heap,\semV{v}\eta)\\
\semV{\myif{v}{t_2}{t_3}}\eta\heap &=& 
\semV{t_2}\eta\heap & \mbox{if
  $\semV{v}\eta =\intt z$, $z\neq 0$}\\
\semV{\myif{x}{t_2}{t_3}}\eta &=& 
\semV{t_3}\eta\heap & \mbox{if $\semV{v}\eta=\intt 0$}
\\
\semV{\letin{x}{t_1}{t_2}}\eta\ \heap,\; &=&
\bot\mbox{, when $\semV{t_1}\eta\ \heap=\bot$}\\
\semV{\letin{x}{t_1}{t_2}}\eta\ \heap&=& 
\semV{t_2}\eta[x{\mapsto}u]\ \heap_1
\mbox{when $\semV{t_1}\eta\ \heap=(\heap_1,u)$}
\\
\semV{\myread{v}}\eta\ \heap &=& (\heap,\heap(\cloc))\mbox{, when
$\semV{v}\eta=\reff \cloc$}\\
\semV{\assign{v_1}{v_2}}\eta\ \heap &=&
(\heap[\cloc{\mapsto}\semV{v_2}\eta],\intt 0)\mbox{, if
$\semV{v_1}\eta=\reff
\cloc$}\\
\semV{\myref{v}}\eta\ \heap &=& \textit{new}(\heap,\semV{v}\eta)\\
\semV{v}\eta &=&\intt 0\mbox{, otherwise}\\
\semV{t}\eta\ \heap&=& (\heap,\intt 0)\mbox{, otherwise}
\end{array}
\]
\else
\vspace{-7mm}
\[
\begin{array}{l}
\semV{x}\eta = \eta(x), 
\semV{c}\eta = \semV{c}, 
\semV{(v_1,v_2)}\eta = (\semV{v_1}\eta,\semV{v_2}\eta),
\semV{v.i}\eta = d_i\ \mbox{if $i=1,2$, $\semV{v}\eta = (d_1,d_2)$}\\
\semV{\vfix{f}{x}t}\eta = \funn
{g^\dagger\,\eta}{\mbox{, where
  $g(\eta,u)=\lambda d.\semV{t}\eta[f{\mapsto}\funn{u},x{\mapsto}d]$}}
\end{array}
\]
\[
\begin{array}{l}
\semV{\myif{v}{t_2}{t_3}}\eta\heap =
\semV{t_2}\eta\heap\,   \mbox{if
  $\semV{v}\eta =\intt z$, $z\neq 0$},
\semV{\myif{x}{t_2}{t_3}}\eta =
\semV{t_3}\eta\heap\,  \mbox{if $\semV{v}\eta=\intt 0$}
\\
\semV{\letin{x}{t_1}{t_2}}\eta\ \heap,\; =
\bot\mbox{, when $\semV{t_1}\eta\ \heap=\bot$}, 
\semV{\letin{x}{t_1}{t_2}}\eta\ \heap= 
\semV{t_2}\eta[x{\mapsto}u]\ \heap_1
\mbox{when $\semV{t_1}\eta\ \heap=(\heap_1,u)$},\\
\semV{\myread{v}}\eta\ \heap = (\heap,\heap(\cloc))\mbox{, when
$\semV{v}\eta=\reff \cloc$}, 
\semV{\assign{v_1}{v_2}}\eta\ \heap =
(\heap[\cloc{\mapsto}\semV{v_2}\eta],\intt 0)\mbox{, if
$\semV{v_1}\eta=\reff
\cloc$},\\
\semV{\myref{v}}\eta\ \heap = \textit{new}(\heap,\semV{v}\eta)
\semV{v}\eta\ \heap =(\heap,\semV{v}\eta), 
\semV{v}\eta =\intt 0\mbox{, otherwise}, 
\semV{t}\eta\ \heap= (\heap,\intt 0)\mbox{, otherwise}
\end{array}
\]
\fi
\caption{Semantics of the untyped meta language \label{seme}}
\vspace{-8mm}
\end{figure*}

\subsection{Setoids}
\label{sec:app-setoids}
\mbox{}

\paragraph{More on dependency}
We should explain what continuity of a dependent function like
$t(-,-)$ is: if $(x_i)_i$ and $(y_i)_i$ and $(z_i)_i$ are ascending
chains in $A$ with suprema $x,y,z$ and $p_i\in A(x_i,y_i)$ and $q_i\in
A(y_i,z_i)$ are proofs such that $(x_i,y_i,p_i)_i$ and $(y_i,z_i,q_i)_i$
are ascending chains, too, with suprema $(x,y,p)$ and $(y,z,q)$ then
$(x_i,z_i,t(p_i,q_i))$ is an ascending chain of proofs (by
monotonicity of $t(-,-)$) and its supremum is $(x,z,t(p,q))$.

Formally, such dependent functions can be reduced to non-dependent ones
using pullbacks, that is $t$ would be a function defined on the pullback of
the second and first projections from $\{(x,y,p)\mid p\in A(x,y)\}$ to
$|A|$, but we find the dependent notation to be much more readable.

\paragraph{Isomorphic pullbacks}

\begin{definition}
Let $\world$ be a category of worlds.  Two pullbacks
$\w\sq{x}{u}{x'}{u'}\w'$ and $\w\sq{y}{v}{y'}{v'}\w'$ are isomorphic
if there is an isomorphism $f$ between the two low points of the
squares so that $vf=u$ and $v'f=u'$, thus also $uf^{-1}=v$ and
$u'f^{-1}=v'$.
\end{definition}
\fi\iffull
It is easy to see that pullback squares can be composed.
\begin{lemma}\label{preo}
  Given a category of worlds $\world$, such that $\w, \w', \w'' \in
\world$, 
  if $\w\sq{x}{u}{x'}{u'}\w'$ and $\w'\sq{y}{v}{y'}{v'}\w''$ are
  pullback squares as indicated then there exist $z,z',t,t'$ such that
  $\w\sq{zx}{ut}{z'y'}{v't'}\w''$ is also a pullback. 
\end{lemma}
\fi
\begin{proof}
  Choose $z,z',t,t'$ in such a way that $\sq{z}{x'}{z'}{y}$ and
  $\sq{u'}{t}{v}{t'}$ are pullbacks.  The verifications are then an
  easy diagram chase.
\end{proof}

Pullback squares can be decomposed as formally described below. This
property is used for instance in the definition of fibred setoids,
formalizing our notion of semantic computation. In particular, to
formalize that the executions of related computations do not depend on 
each other.

\begin{lemma}\label{decomp}
A pullback  square $\sq{x}{u}{x'}{u'}$ in a category of worlds is
isomorphic to $t(\sq{x}{1}{1}{x}, \sq{1}{x'}{x'}{1})$.
\end{lemma}

\paragraph{Pullback-preserving functors}
\iffull
\begin{lemma}
  If $A$ is a p.p.f., $u:\w\rightarrow \w'$ and $a,a'\in A\w$,
  there is a continuous function $A\w'(u.a,u.a')\rightarrow
  A\w(a,a')$. Moreover, the ``common ancestor'' $\underline{a}$ of $a$ and $a'$ is unique up to
$\sim$. 
\end{lemma}
Note that the ordering on worlds and world morphisms is discrete so
that continuity only refers to the $A\w'(u.a,u.a')$ argument.
\fi
\iffull
\begin{definition}[Morphism of functors]\label{morphfun}
  If $A,B$ are p.p.f., a morphism from $A$ to $B$ is a pair
  $e=(e_0,e_1)$ of continuous functions where
  $e_0:\Pi\w.A\w\rightarrow B\w$ and $e_1:\Pi\w.\Pi\w'.\Pi
  x:\w\rightarrow\w'.\Pi a\in A\w.\Pi a'\in
  A\w'.A\w'(x.a,a')\rightarrow B\w'(x.e_0(a),e_0(a'))$. A proof that
  morphisms $e,e'$ are equal is given by a continuous function
  $\mu:\Pi \w.\Pi a\in A\w.B\w(e(a),e'(a))$.
\end{definition}
These morphisms compose in the obvious way and so
the pullback-preserving functors and morphisms between them form a category. 
\fi

\paragraph{More on $S(A)$ and fibred setoids}
\iffull If $\sq{x}{u}{x'}{u'}$ and
  $\sq{y}{v}{y'}{v'}$ are two composable pullback squares with
  composite $\sq{zx}{ut}{z'y'}{v't'}$ and $p\in S(A) \sq{x}{x'}{u}{u'}
  (a,a')$ and $p'\in S(A) \sq{y}{y'}{v}{v'} (a',a'')$, then the composite
proof of $t_{S(A)}(p,p')\in S(A)\sq{zx}{ut}{z'y'}{v't'}(a,a'')$ is given
by $t_A(z.p,z'.p')$. Indeed, if $\hat\w=\cod z$ is the apex of the
composite square  then  $z.p\in A{\hat\w}(zx.a,zx'.a')$ and $z'.p'\in
A{\hat\w}(z'y.a',z'y'.a'')$ and $zx'.a'=z'y.a'$ since $zx'=z'y$ so the
two proofs compose in $A{\hat \w}$.  \fi

\iffull
\begin{lemma}\label{baf}
  Let $T$ be a fibred setoid.  The elements $\underline{t}$ given by pullback
  preservation are unique up to $\sim$.  If $u:\w\rightarrow \w'$ is
  an isomorphism then there is a continuous function
  $Tu:T\w\rightarrow T\w'$ and it is bijective up to $\sim$ with
  inverse $T(u^{-1})$. If $\sqsol$ and $\sqsol'$ are isomorphic pullback squares then there are continuous back and forth functions 
$\Pi t.\Pi t'.T\sqsol(t,t')\rightarrow T\sqsol'(t,t')$. 
\end{lemma}
\fi

\iffull
\begin{lemma}\label{meanterm}
  If $A$ is a p.p.f.\ and $T$ is a fibred setoid then in
  order to specify a morphism from $S(A)$ to $T$ with given first
  component $f_0: \Pi\w.A\w\rightarrow T\w$ it is enough to provide a
  continuous function $f_{0.5}:\Pi\w,\w'.\Pi x:\w\rightarrow \w'.\Pi
  a\in A\w.\Pi a'\in A\w'.A\w'(x.a,a')\rightarrow
  T\sq{x}{1}{1}{x}(f_0(a),f_0(a'))$.
\end{lemma} 
\begin{proof}
  If $(f_0,f_1)$ is a morphism we can define $f_{0.5}$ by
  $f_{0.5}(x,p)=f_1(x,a,a',p)$ noting that $p\in
  S(A)\sq{x}{1}{1}{x}(a,a')$.  Conversely, given $f_{0.5}$ to define
  $f_1$ we pick a pullback square $\w\sq{x}{u}{x'}{u'}\w'$ with apex
  $\overline{\w}$ and $a\in A\w, a'\in A\w'$ and $p\in A{\overline
    \w}(x.a,x'.a')$, i.e., a proof in $S(A){\sqsol}(a,a')$.  Applying
  $f_{0.5}$ to $r(-)$ yields the morphism $p_1\in
  T\sq{x}{1}{1}{x}(f_0(a),f_0(x.a))$; moreover, applying $f_{0.5}$ to
$s(p)$
  yields $p_2\in T\sq{x'}{1}{1}{x'}(f_0(a'),f_0(x.a))$.  Then,
  $t(p_1,s(p_2))\in
  Tt(\sq{x}{1}{1}{x},\sq{1}{x'}{x'}{1})(f_0(a),f_0(a'))$ so that
  Lemmas~\ref{decomp} and \ref{baf} yield the desired proof in the square
  $T\sq{x}{u}{x'}{u'}(f_0(a),f_0(a'))$.

  The second part of the lemma about equality is just a restatement of
  the definition of equality of morphisms of fibred setoids. 
\end{proof}\fi
\iffull
\begin{lemma}
Let $A,B$ be p.p.f.\ For every morphism $e:A\rightarrow B$ there  is a
morphism $S(e):S(A)\rightarrow S(B)$ such that $S(e)_0=e_0$. Thus, in
particular $S(-)$ is a full and faithful functor from the category  of
p.p.f.\ on $\world$ to the category of fibred setoids over $\world$. 
\end{lemma} 
\fi

\paragraph{On abstract heaps}
The definition of minimal pullback-preserving functor corresponds to
the p.p.f. used for values, but is used for abstract heaps. In
particular, an abstract heap at the low-point of a pullback square is
the result of forgetting locations from an abstract heap at its apex.

\iffull
Applying the definition of minimal ppf to the trivial minimal pullback
$\sq{u}{1}{1}{u}$, plus
nonemptiness, yields the following result.
\begin{lemma}
\label{lem:unique-upto-sim}
   For every $u:\w\rightarrow \w'$ and $\sigma\in\Astores\w$ there is
morphism of setoids $\Astores \w\rightarrow \Astores \w'$ which is right
inverse to $(-).u$. 
\end{lemma}
The ``unique up to $\sim$'' clause allows us in particular to assert the $\sim$-equality of two abstract stores $\sigma,\sigma'\in\Astores\overline{w}$ by proving $\sigma.x\sim\sigma'.x$ and $\sigma.x'\sim\sigma'.x'$ separately when $\sq{x}{u}{x'}{u'}$ is a minimal pullback  with apex $\overline{\w}$.
\fi

\subsection{Computational model}
\label{sec:app-comp-model}
\mbox{}

We now discuss a third instantiation of our framework, which captures
the setting developed in \cite{DBLP:conf/ppdp/BentonKBH09}.

\paragraph{Flat stores} 
The \emph{flat stores} instantiation assumes that heap
locations contain merely integer values and no pointers. 
Possible worlds are
finite sets of locations together with a function that associates each
location a \emph{region} taken from a fixed set $\Regids$ of
regions. World morphisms must preserve this tagging. We write
$\cloc\in\w$ and $\cloc\in\w(\regid)$ to mean that $\cloc$ occurs in $\w$
and with region $\regid$ in the second case.  Abstract stores
$\Astores \w$ comprise those heaps $\heap\in\Stores$ with $\dom
\heap\supseteq
\w$ and such that $\cloc\in\w$ and $\heap\in\Astores\w$ implies that
$\heap(\cloc)$ is an integer value, $\intt v$ for $v\in\mathbb{Z}$
 (thus all locations hold integer
values). We put $\heap\sim \heap'$ in $\Astores\w$ iff for all $\cloc\in
\w$
one has $\heap(\cloc)=\heap'(\cloc)$. In this case there is a unique proof,
say $\star$.  For morphism $u:\w\rightarrow \w'$ we define $\Astores
u:\Astores \w'\rightarrow \Astores \w$ by renaming concrete locations
according to $u$. The elementary effects are $\rEff\regid,
\wEff\regid,\aEff\regid$ representing reading from within, writing
into, allocating within a region $\regid$. The associated sets of
relations are given by
\iffull
\[
 \begin{array}{lcl}
R\in\mathcal{R}(\rEff{\regid})&\iff& (\sigma,\sigma')\in R\w \Rightarrow
\forall \cloc\in \w(\regid). \sigma(\cloc)=\sigma'(\cloc)
\\[3pt]

R\in\mathcal{R}(\wEff{\regid})&\iff& (\sigma,\sigma')\in R\w \Rightarrow
\forall \cloc\in \w(\regid).\forall v\in\mathbb{Z}.\Rightarrow 
(\sigma[\cloc{\mapsto}\intt{v}],\sigma'[\cloc{\mapsto}\intt{
v}])\in R\w
\\[3pt]

R\in\mathcal{R}(\aEff{\regid})&\iff& (\sigma,\sigma')\in R\w 
\Rightarrow \forall \w_1.\forall u\in\mathbf{I}(\w,\w_1). (\dom
{\w_1}\setminus \dom{\w} \subseteq \dom{\w_1(\regid)}) \\
&& \quad
\Rightarrow \forall \sigma_1\in \Astores \w_1,\sigma_1'\in \Astores \w_1'.
\sigma_1.u\sim \sigma\wedge \sigma_1'.u\sim\sigma'\wedge\\
&& \qquad \forall\cloc\in
\dom{\w_1}\setminus\dom \w.  \sigma_1(\cloc)=\sigma_1'(\cloc) \Rightarrow
(\sigma_1,\sigma_1')\in R \w_1
 \end{array}
\]
This essentially mirrors the setting of our earlier relation-based
\else

\(
\hspace{-5mm}
 \begin{array}{lcl}
R\in\mathcal{R}(\rEff{\regid})&\iff& (\sigma,\sigma')\in R\w \Rightarrow
\forall \cloc\in \w(\regid). \sigma(\cloc)=\sigma'(\cloc)
\\[3pt]

R\in\mathcal{R}(\wEff{\regid})&\iff& (\sigma,\sigma')\in R\w \Rightarrow
\forall \cloc\in \w(\regid).\forall v\in\mathbb{Z}.\Rightarrow 
(\sigma[\cloc{\mapsto}\intt{v}],\sigma'[\cloc{\mapsto}\intt{
v}])\in R\w
\\[3pt]

R\in\mathcal{R}(\aEff{\regid})&\iff& (\sigma,\sigma')\in R\w 
\Rightarrow \forall \w_1.\forall u\in\mathbf{I}(\w,\w_1). (\dom
{\w_1}\setminus \dom{\w} \subseteq \dom{\w_1(\regid)}) \\
&& \quad
\Rightarrow \forall \sigma_1\in \Astores \w_1,\sigma_1'\in \Astores \w_1'.
\sigma_1.u\sim \sigma\wedge \sigma_1'.u\sim\sigma'\wedge\\
&& \qquad \forall\cloc\in
\dom{\w_1}\setminus\dom \w.  \sigma_1(\cloc)=\sigma_1'(\cloc) \Rightarrow
(\sigma_1,\sigma_1')\in R \w_1
 \end{array}
\)

\noindent
This essentially mirrors the setting of our earlier relation-based
\fi
account of reading, writing, and allocation with integer values stores
\cite{DBLP:conf/ppdp/BentonKBH09} with the difference that allocation
is modelled with relations on the same level as reading and writing
and that the stores being related share the same layout. 

\subsection{Proof-relevant logical relations}
\label{sec:app-logical-relations}

In following establishes that the semantics of the monad corresponds
indeed to a semantic computation, that is, a fibred setoid.

\begin{proposition}
The semantic computation $T_\eff A$ as defined in
Definition~\ref{def:semantic-computation} is a fibred setoid.
\end{proposition}
\begin{proof}
  The tricky case is to show the existence of a transitive operation.
  It is here that we require the independence of abstract locations as
  stated in Definition~\ref{def:abstract-location}, which implies that
  $\Astores$ is also minimal-pullback-preserving.

Assume that there are proofs in $p_1 : T_\eff A \sq{x_1}{v_1}{x_1'}{v_1'}
(\cval, \cval')$ and $p_2 : T_\eff A \sq{x_2}{v_2}{x_2'}{v_2'} (\cval',
\cval'')$ where 
$\w \sq{x_1}{v_1}{x_1'}{v_1'} \w'$ and $\w' \sq{x_2}{v_2}{x_2'}{v_2'}
\w''$. We also have $\sigma \in \Astores \w$ and $\sigma'' \in \Astores
\w''$, such that they are equivalent in the pullback of
the low points of these two pullback squares. Let $\underline{\q}$ be such
pullback. 

In order to use the proofs $p_1$ and $p_2$, we need to construct
from $\sigma$ and $\sigma''$ an abstract heap $\sigma' \in \Astores\w'$.
Let $\overline{\q}$ be the minimal pullback over the apexes of the two pullback
squares $\w \sq{x_1}{v_1}{x_1'}{v_1'} \w'$ and $\w'
\sq{x_2}{v_2}{x_2'}{v_2'} \w''$. Then $\w$ and $\w''$ form a pullback
square with apex $\overline{\q}$ and low point $\underline{\q}$. Since
$\Astores$ is minimal-pullback-preserving, there is a $\sigma_{\q} \in \Astores
\overline{\q}$, such that it is equivalent to $\sigma$ and $\sigma''$ when
taken to the world $\underline{\q}$. We now define $\sigma' \in \Astores
\w'$ to be $\sigma_{\q}$ taken to the world $\w'$. We thus have $\sigma'
\in \Astores
\w'$, and $\sigma'' \in \Astores \w''$, such that $\sigma.v_1 \sim
\sigma'.v_1'$ and $\sigma'.v_2' \sim \sigma''.v_2'$.

We can now use the $p_1$ and $p_2$. In particular, 
let $\cval(\sigma) = (\w_1, u_1, \sigma_1, \val_1)$, $\cval'(\sigma') =
(\w_1', u_1', \sigma_1', \val_1')$, and $\cval''(\sigma'') = (\w_1'',
u_1'', \sigma_1'', \val_1'')$. From the proofs, we get two
pullback squares $\w_1\sq{}{}{}{}\w_1'$ and $\w_1'\sq{}{}{}{}\w_1''$. It
is easy to show that the values obtained are equal in the minimal pullback over 
the
apexes of these two pullback squares and that the abstract heaps are
equivalent in the pullback of their low points.
\end{proof} 

\begin{definition}[cartesian product]
If $(A,\Vdash^A)$ and $(B,\Vdash^B)$ are semantic types their cartesian product $(A\times B,\Vdash^{A\times B})$ is defined by $(A\times B)\w=A\w\times B\w$ (cartesian product of setoids) and $(v_1,v_2)\Vdash_\w^{A\times B}(a,b)\iff v_1\Vdash_\w^{A}a\wedge v_2\Vdash_\w^{B}b$. 
\end{definition}
\begin{definition}[function space]
  Let $(A,\Vdash^A)$ be a semantic type and $(T,\Vdash^T)$ be a
  semantic computation. We define a semantic type $(A{\Rightarrow}T,
  \Vdash^{A{\Rightarrow}T})$ as follows.  
  An object $f$ of $(A{\Rightarrow} T)\w$ is a pair $(f_0,f_1)$ of
  continuous functions where $f_0$ assigns to each $\w_1$ and
  $v:\w\rightarrow \w_1$ a continuous function $f_0(v):A\w_1\rightarrow
  T\w_1$. The second component $f_1$ assigns to each
  $v:\w\rightarrow \w_1$ and $v_1:\w_1\rightarrow \w_2$ a continuous
  function $\Pi a\in A\w_1.\Pi a'\in A\w_2.A\w_2(v_1.a,a')\rightarrow
  T\sq{v_1}{1}{1}{v_1}(f_0(v,a),f_0(v_1v,a'))$.

  If $f,f'\in |A{\Rightarrow}T|$ then a proof $\mu\in
  (A{\Rightarrow}T)(f,f')$ is a continuous function assigning to each
  $v:\w\rightarrow \w_1$ and $a\in A\w_1$ a proof $\mu(v,a)\in
  T\sq{1}{1}{1}{1}(f_0(v,a), f_0'(v,a))$.

  If $u:\w\rightarrow \w'$ and $f=(f_0,f_1)\in (A{\Rightarrow}T)\w$
  then $u.f\in (A{\Rightarrow}T)\w'$ is given by precomposition with
  $u$, i.e., $(u.f)_0(v,a)=f_0(vu,a)$, etc.


  As for the realisation relation $\Vdash^{A{\Rightarrow}T}$ we put
  $v\Vdash^{A{\Rightarrow}T}_\w f$ to mean that $v=\funn g$ for some
  $g$ and whenever $i:\w\rightarrow \w_1$ is an inclusion and
  $u\Vdash^A_{\w_1}a$ then $g(u)\Vdash^T_{\w_1} f(i,a)$.
\end{definition}
Notice that unlike morphisms the elements of the function space are
\emph{not} identified if they are ``provably equal.''
Notice also that if $v\Vdash^{A{\Rightarrow} T}_\w f$ implies
$v\Vdash^{A{\Rightarrow}T}_{\w_1}i.f$ whenever $i:\w\rightarrow \w_1$
is an inclusion.

In what follows we define semantic counterparts to the generic
syntactic constructions common to all instantiations, namely
application and abstraction, sequential composition, subeffecting, and
recursion
that allow us to define this interpretation of derivations in a
compositional fashion. Having given these semantic counterparts we
then omit the formal definition of the interpretation $\sem{-}$.
\begin{lemma}[Abstraction]
Let $\Gamma,A$ be  semantic types, $T$ a semantic computation. There is a function $\lambda$ so that if $e:S(\Gamma\times A)\rightarrow T$ is a morphism of fibred setoids then $\lambda(e):S(\Gamma)\rightarrow A{\Rightarrow}T$. Moreover, if $e\sim e'$ then $\lambda(e)\sim\lambda(e')$ and if $f\Vdash^{\Gamma\times A\rightarrow T} e$ then $\lambda\eta.\lambda a.f(\eta,a)\Vdash^{\Gamma\rightarrow A{\Rightarrow}T} \lambda(e)$. 
\end{lemma}
\begin{lemma}[Application]
Let $A$ be a semantic type and $T$ be a semantic computation. There  is a
morphism $\textit{app} : S((A{\Rightarrow}T)\times A)\rightarrow T$ and
$\lambda (f,a).f(a)\Vdash^{((A{\Rightarrow}T)\times A)\rightarrow T}
\textit{app}$.  
\end{lemma}

\begin{lemma}[subeffecting]
 Let $\Gamma, A$ be semantic types and $\eff, \eff'$ be effects. There is
a function \emph{subeff}, so that if $e: S(\Gamma) \to T_\eff A$, then 
$\emph{subeff}(e) : S(\Gamma) \to T_{\eff \cup \eff'} A$. Moreover, if $e
\sim e'$, then $\emph{subeff}(e) \sim \emph{subeff}(e')$. Finally, if $f
\Vdash^{\Gamma \to T_\eff A} e$ then $f \Vdash^{\Gamma \to T_{\eff \cup
\eff'} A} \emph{subeff}(e)$.
\end{lemma}

\begin{proof}
For the first component, $\emph{subeff}_0$, we use the same first component
$e_0$ of  $e$. What changes is the definition of the second
component, $\emph{subeff}_1$. It is defined only for relations $R \in
\Rscr(\eff \cup \eff')$, for which $e_1$ is also defined. For some related
given abstract heaps in $R$, $\emph{subeff}_1$ calls $e_1$ constructing the
corresponding pullback. For proofs the reasoning is similar.
\end{proof}

We elide assertions about $\sim$-versions of beta-eta-equality, and
the existence of ``value morphisms'' of type
$S(A)\rightarrow T_\eff A$ for any semantic type $A$. 

\begin{lemma}[let]
Let $\Gamma,A,B$ be semantic types and $\eff$ an effect. 
There is a function $\textit{let}$ such that if $e_1:S(\Gamma)\rightarrow
T_\eff A$ and $e_2:S(\Gamma\times A)\rightarrow T_\eff B$ are morphisms
then $\textit{let}(e_1,e_2):S(\Gamma)\rightarrow T_\eff B$. Moreover, if
$e_1\sim
e_1'$ and $e_2\sim e_2'$ then $\textit{let}(e_1,e_2)\sim
\textit{let}(e_1',e_2')$. 
Finally, if $f_1\Vdash^{\Gamma\rightarrow T_\eff A} e_1$ and
$f_2\Vdash^{\Gamma\times A\rightarrow T_\eff B} e_2$ then $\lambda \eta.
\lambda \heap.\textit{let } (\heap_1,\vval){=}f_1(\eta)(\heap)\textit{ in }
f_2(\eta,\vval)(\heap_1)\Vdash^{\Gamma\rightarrow T_\eff
A}\textit{let}(e_1,e_2)$. 
\end{lemma}

\begin{proof}
Consider the following definition for the first component of the morphism
$\textit{let}(e_1,e_2)$ which is only defined when ${e_1}$ and
${e_2}$ are defined. The type of this component is $\sem{\Gamma}\w \to
T_\eff \sem{B}\w$. Hence, assume a world $\w$, and a context $\gamma \in
\sem{\Gamma}\w$, then one returns an object $(\cval_0, \cval_1) \in
T_\eff \sem{B}\w$. The first component $\cval_0$ is:
\(
 \Pi \w. \Pi \gamma \in \sem{\Gamma}\w . \Pi \sigma \in \Astores \w.
{e_2}(\w_1)(\gamma,\val_1)\sigma_1
\)
where ${e_1}(\w)(\gamma)\sigma = (\w_1, u_1, \sigma_1, \val_1)$.

For the second component, $\cval_1$, assume a relation $R \in
\Rscr(\eff)$, and two abstract heaps $\sigma, \sigma' \in \Astores \w$
such that $(\sigma, \sigma') \in R\w$. From ${e_1}$ we get a 
proof $\w_1 \sq{x_1}{v_1}{x_1'}{v_1'} \w_1'$, where
${e_1}(\w)(\gamma)\sigma
=
(\w_1, u_1, \sigma_1, \val_1)$ and ${e_1}(\w)(\gamma)\sigma' =
(\w_1', u_1', \sigma_1', \val_1')$, such that $(\sigma_1.v_1,
\sigma_1'.v_1') \in R$ and $p : \sem{A}\overline{\w_1}(x_1.\val_1,
x_1'.\val_1')$. Applying ${e_2}$ on $\sigma_1.v_1$ and
$\sigma_1'.v_1'$ we get a proof $\q_2 \sq{y_2}{v_2}{y_2'}{v_2'} \q_2'$,
such that $(\tilde{\sigma_2}.v_2, \tilde{\sigma_2}'.v_2') \in R$. However,
we need to show that the heaps obtained from applying ${e_2}$ on $\sigma_1$
and $\sigma_1'$ (using the correct world and context), namely $\sigma_2$
and $\sigma_2'$, are related. For this we rely on
the morphism $(e_2)_1$. In particular,
we use $(e_2)_1$ on the pullback $\w_1 \sq{1}{x_1}{x_1}{1} \und{\w_1}$ and
obtain a pullback $\w_2 \sq{}{}{}{} \q_2$ such that $\sigma_2$ and
$\tilde{\sigma_2}$
are equal in its low point. Similarly, applying $(e_2)_1$ on the
pullback $\und{\w_1} \sq{x_1'}{1}{1}{x_1'} \w_1'$, we get a pullback 
$\q_2' \sq{}{}{}{} \w_2'$, where $\tilde{\sigma_2}'$ is equal to
$\sigma_2'$ in its pullback. Using Lemma~\ref{preo}, we compose the
pullbacks $\w_2 \sq{}{}{}{} \q_2$, $\q_2 \sq{}{}{}{} \q_2'$ and $\q_2'
\sq{}{}{}{} \w_2'$, obtaining a common pullback $\und{\q}$, where 
$\sigma_2$ and $\sigma_2'$ when taken to $\und{\q}$ are in $R$. 

The morphism ${\textit{let}(e_1,e_2)\sim
\textit{let}(e_1',e_2')}$ can be then defined when ${e_1 \sim e_1'}$
and ${e_2 \sim e_2'}$ are defined.
Assume a pullback $\w \sq{1}{1}{1}{1}\w$ and an abstract heap $\sigma \in
\Astores \w$ and a context $\gamma \in \sem{\Gamma}\w$. Using the
morphism between $e_1$ and $e_1'$ on these objects, we obtain a
pullback $\w_1 \sq{x_1}{v_1}{x_1'}{v_1'} \w_1'$, $p_1 \in
\sem{A}\overline{\w_1}(x_1.\val_1, x_1'.\val_1')$ and $q_1 : \sigma_1.v_1
\sim \sigma_1'.v_1'$, where ${e_1}(\w)(\gamma)\sigma = (\w_1, u_1,
\sigma_1,
\val_1)$ and ${e_1'}(\w)(\gamma)\sigma = (\w_1', u_1', \sigma_1',
\val_1')$. 
From the pullback preserving property of computations and $p_1$,
there is a common value $\und{\val} \in \sem{A}\und{\w_1}$ and context
$\und{\gamma} \in \sem{\Gamma}\und{\w_1}$ which are equal, respectively, to
$\val_1$ and $\val_1'$, and $\gamma$ and $\gamma'$ (when taken to the
correct world). We then construct a
proof $\sem{\Gamma \times A} \und{\w_1}$. We now apply twice the
morphism between $e_2$ and $e_2'$ once in the pullback $\w_1 \sq{}{}{}{}
\und{\w_1}$ and
another on the
pullback $\und{\w_1} \sq{}{}{}{} \w_1'$, obtaining two pullbacks
$\w_2 \sq{}{}{}{} \q_2$ and $\q_2 \sq{}{}{}{} \w_2'$. From
Lemma~\ref{preo}, we can compose them where the resulting values and
heaps are equal.
\end{proof}

\begin{lemma}[fix]\label{fixlse}
  Let $\Gamma,D$ be semantic types so that for each $\w$ the predomain
  $D\w$ is a domain with least element $\bot\w$ such that
  $(\bot\w,\bot\w,r(\bot\w))\leq (d,d',p)$ holds for every proof $p\in
  D(d,d')$ and such that $x.\bot_\w=\bot_{\w'}$ holds for every
  $x:\w\rightarrow\w'$.\footnote{For example $D=A {\Rightarrow}T_\eff B$
for semantic types $A,B$.}
\begin{compactenum}
\item[i] There then exists a function
  $\textit{fix}$ so that whenever $e:\Gamma\times D\rightarrow D$ then
  $\textit{fix}(e):\Gamma\rightarrow 
  D$
\item[ii] If $e\sim e'$ then $\textit{fix}(e)\sim
  \textit{fix}(e')$. Furthermore, the fixpoint and unrolling equations
  from Lemma~\ref{fixlse} hold. 
\item[iii] Finally, if $f\Vdash^{\Gamma\times D\rightarrow D}e$ then
$f^\dagger\Vdash \textit{fix}(e)$.    
\end{compactenum}
\end{lemma}
\begin{proof}
For every $\w$ we have $e_0\w:\Gamma\w\times D\w
\rightarrow D\w$. We can thus form $\textit{fix}(e)_0\w:=(e_0\w)^\dagger :
\Gamma \w\rightarrow D\w$.  It remains to define $\textit{fix}(e)_1$. To do that, we
recall that we 
have an ascending chain of elements 
$\textit{fix}^n(e)_0\w(\gamma)\in D\w$ given by 
$\textit{fix}^0(e)_0\w(\gamma)=\bot_\w$ and 
$\textit{fix}^{n+1}(e)_0\w(\gamma)=e_0\w(\gamma,\textit{fix}
^n(e)_0\w(\gamma))$ and have
$\textit{fix}(e)_0\w(\gamma)=\sup_n\textit{fix}^n(e)_0\w\gamma$. 
Now suppose that $\gamma\in\Gamma\w$ and $x:\w\rightarrow \w'$ and
$\gamma'\in\Gamma\w'$ and $p\in\Gamma\w'(x.\gamma,\gamma')$. Write
$d_n=\textit{fix}^n_0\w(\gamma)$ and $d_n'=\textit{fix}^n_0\w'(\gamma')$.
Inductively, we get proofs $p_n\in
D\w'(x.d_n,d_n')$ where $p_0=r(\bot_{\w'})$ (note that
$x.\bot_\w=\bot_{\w'}$) and $p_{n+1}=e_1(p,p_n)$. Since
$(x.\bot_\w,\bot_{\w'},r(\bot_{\w'}))\leq (x.d_1,d_1',p_1)$ we obtain by
monotonicity of $e_1$ and induction that $(x.d_n,d_n',p_n)$ is an ascending
chain with supremum $(x.\sup_n d_n,\sup_n d_n',q)$ for some proof $q$
which we take as $\textit{fix}(e)_1(p)$. Note that the passage from $p$ to
$q$ is continuous. 
\end{proof}

\subsection{Applications}
\label{sec:app-applications}\mbox{}

The following lemma formalizes our intuition that 

\paragraph{Lemma~\ref{lemmasix}}
\begin{proof}
The proof that the values are equal in $\und{\w}$ follows directly from 
the definition of computations and effects.

For the first part,  we use the following relation $R$ defined for all
worlds $\w_1$, such that $u : \w \to \w_1$:

\(
\begin{array}{l}
 \{(\sigma, \sigma') \mid \sigma \sim_{\rdsin{\eff, \w}} \sigma' 
\land \forall \loc \in \w.\\
 \qquad (\sigma.u, \sigma_0) \in \loc^R \land
(\sigma'.u, \sigma_0') \in \loc^R \lor (\sigma.u, \sigma'.u) \in
\loc^R\}
\end{array}
\)

\noindent
Otherwise, for the worlds $\w_2$ not reachable from \w, the relation
$R\w_2$ is the trivial set. Notice that $R \in \Rscr(\eff)$ and it is
contravariant. The claim then 
follows directly.

The proof of the second part follows in a similar fashion, but we use the
following relation:

\(
\{(\sigma, \sigma') \mid \sigma \sim_{\rdsin{\eff, \w}} \sigma' 
\land \sigma \sim_{\nwrs(\eff, \w)} \sigma_0.u \}
\)

\noindent
And we use a similar relation for showing that $\sigma_0'$ and
$\sigma_1'.u'$ agree on the not written locations $\nwrs(\eff, \w)$.

For the third property, first, we show that there is an
isomorphism between $\w(\regid)$ and $\und{\w}(\regid)$ for all regions
$\regid \notin \alsin{\regid}$ by using the following relation:

\(
\{(\sigma, \sigma') \mid \sigma \sim \sigma' \land \forall
\regid \notin \alsin{\eff}. \#_\regid(\sigma), \#_\regid(\sigma')
\leq \#_\regid(\w) \} 
\)

\noindent
where $\#_\regid$ denotes the number of abstract locations coloured with 
$\regid$.
Clearly, $R \in
\Rscr(\eff)$ as $\eff$ does not contain any allocation effects. This gives
us one direction, while the other direction is obtained by using the
inclusion morphisms. Given this property, one can easily construct the 
function $\cval'$.
\end{proof}

\begin{proposition}\label{compr}
(commuting computations) Suppose that:
$
 \Gamma \vdash e_1 : \ety{\tau_1}{\eff_1}$ and
$\Gamma \vdash e_2 :
\ety{\tau_2}{\eff_2}
$,
where  $\rdsin{\eff_1} \cap \wrsin{\eff_2} = \rdsin{\eff_2} \cap
\wrsin{\eff_1}= 
\wrsin{\eff_1} \cap \wrsin{\eff_2} = \emptyset$. Let 
\[
 \begin{array}{l}
  e = \letin{x}{e_1} \letin{y}{e_2} (x, y) \quad \textrm{and} \quad
  e' = \letin{y}{e_2} \letin{x}{e_1} (x, y)
 \end{array}
\]
then $\sem{\Gamma \vdash e : \ety{ \tau_1 \times \tau_2 }{\eff_1\cup
\eff_2}} \sim \sem{\Gamma \vdash e' : \ety{ \tau_1 \times \tau_2
}{\eff_1\cup \eff_2}}$.
\end{proposition}

\begin{proof}
Assume a world \w and a context $\gamma \in \sem{\Gamma}\w$.
Let $\cval_i = \sem{\Gamma \vdash e_i : \ety{\tau_i}{\eff_i}}$ for
$i = 1,2$. 

It is enough to assume a pullback 
$\w \sq{1}{1}{1}{1} \w$, and an abstract heap $\sigma_0 \in
\Astores \w$. 
Assume that these functions are defined as follows:

\(
 \begin{array}{l}
 \cval_1(\w)(\gamma)\sigma_0 = (\w \uplus \w_1, u_1, \sigma_1, \val_1)\\
 \cval_2(\w\uplus\w_1)(u_1.\gamma) \sigma_1 = (\w \uplus \w_1 \uplus
\w_2, u_2, \sigma_2, \val_2)\\
 \cval_2'(\w)(\gamma)\sigma_0 = (\w \uplus \w_1', u_1', \sigma_1',
\val_1')\\
 \cval_1'(\w\uplus\w_2')(u_1'.\gamma) \sigma_1' = (\w \uplus \w_1'
\uplus \w_2', u_2', \sigma_2',\val_2')
 \end{array}
\)

\noindent
One can easily show that when one of these functions is undefined, then 
the corresponding function is also undefined. 

We need to show that there is a proof $\w \uplus \w_1 \uplus \w_2
\sq{x}{v}{x'}{v'} \w \uplus \w_1' \uplus \w_2'$ such that
$p : \sigma_2.v \sim \sigma_2'.v'$ and $p_1 : xu_2.\val_1 \sim x'.\val_2'$
and $p_2 : x.\val_2 \sim x'u_2'.\val_1'$.
Decompose $\w = \w_0 \uplus \q_1 \uplus \q_2$, where
$\w(\wrsin{\eff_i}) \subseteq \q_i$. 
The existence of such decomposition follows from the disjointness of write
effects in $\eff_1$ and $\eff_2$.

From Lemma~\ref{lem:invariant} and from 
the disjointness of reads and writes, it is the case that $\sigma_0$ and
$\sigma_1'$ agree on the locations in $\w_0 \uplus \q_1$.
That is, there is a proof $p : \sigma_0.1 \sim \sigma_1'.x_1$, defined 
using the proof $\w_0 \uplus \q_1 \sq{x_1}{1}{1}{x_1} \w_0
\uplus \q_1 \uplus \w_2'
$, where $x_1 : \w_0 \uplus \q_1 \to \w_0 \uplus \q_1 \uplus \w_2'$.
Applying $(e_1)_1$ to the objects above,
we get the pullback $\w_0 \uplus \q_1 \uplus \w_1 \sq{x_2}{v_2}{x_2'}{v_2'}
\w_0 \uplus \q_1
\uplus
\w_2' \uplus \w_1'$, and  proof $q : x_2.\val_1 \sim
x_2'.\val_2'$. Symmetrically, we obtain the proofs 
$\w_0 \uplus \q_2 \uplus \w_2 \sq{x_3}{v_3}{x_3'}{v_3'} \w_0 \uplus \q_2
\uplus \w_1' \uplus \w_2'$, and  $q' : x_3.\val_2 \sim
x_3'.\val_1'$. Hence, there is also a proof in the larger world $cod(x)$.

To see informally that the final heaps $\sigma_2$ and $\sigma_2'$ are
equal, we use
the following facts obtained using Lemma~\ref{lem:invariant}: $\sigma_2$
and $\sigma_1$ agree on the locations in $\w_0 \uplus \q_1$; moreover, 
$\sigma_2'$ and $\sigma_1$ agree on the locations in $\w_0 \uplus \q_1$; 
hence $\sigma_2$ and $\sigma_2'$ agree on the locations in $\w_0 \uplus
\q_1$. Symmetrically, we can also argue that $\sigma_2$ and $\sigma_2'$
agree on the locations in $\w_0 \uplus \q_2$. Composing these proofs 
(see comment after Lemma~\ref{lem:unique-upto-sim} why this is allowed), 
we get that $\sigma_2$ and $\sigma_2'$ agree on the locations in $\w$. 
Finally, since the locations allocated by one computation are not used by
the other computation, the final heaps are equal at the
apex world.
%
%
\end{proof}

The following propositions are also provable. All propositions are proved
in a similar way as the soundness proof of the commuting case, using 
Lemma~\ref{lem:invariant} when needed. For instance, the soundness proof
of the duplicated computation uses the third case in
Lemma~\ref{lem:invariant}.

\begin{proposition}[dead computation]
Suppose that
$\Gamma \vdash e : \ety{\unittype}{\eff}$,
that \iffull $\wrsin{\eff}=\emptyset$ \else 
$\eff$ contains no writes \fi and that $\sem{\Gamma\vdash
e: \ety{\unittype}{\eff}}\w(\gamma)(\sigma)$ is defined for all
$\w,\gamma\in\sem{\Gamma}\w, \sigma\in\Astores\w$. 
Then if for all worlds \w, all contexts $\gamma \in
\sem{\Gamma}\w$, and abstract heaps $\sigma \in \Astores\w$,
the function $\sem{\Gamma \vdash e}(\w)(\gamma)(\sigma)$ is defined, 
then $\sem{\Gamma \vdash e : \ety{\unittype}{\eff}} \sim \sem{\Gamma
\vdash \unitval : \ety{\unittype}{\eff}}$.
\end{proposition}

\begin{proof}
Assume a world \w and a context $\gamma \in \sem{\Gamma}\w$.
Let $\cval = \sem{\Gamma \vdash e : \ety{\tau}{\eff}}$. 
It is enough to assume a pullback 
$\w \sq{1}{1}{1}{1} \w$, and an abstract heap $\sigma_0 \in
\Astores \w$. Let $\cval(\w)(\gamma)\sigma_0  =  (\w, 1, \sigma_1,
\val_1)$. We need to construct a pullback such that $\val_1$ is equivalent
to $\unitval$ in its apex and $\sigma_1$ is equivalent to $\sigma_0$ in
its low point. Consider the pullback $\w_1 \sq{1}{u}{u}{1} \w$. Clearly
$\val_1 = \unitval$, and therefore the values are equivalent in $\w_1$.
Moreover, from the fact that $\wrsin{\eff} =
\emptyset$, $\sigma_1$ and $\sigma_0$ agree on all locations in $\w$.
Hence, $\sigma_1.u \sim \sigma_0$, which finishes the proof. 
\end{proof}

\begin{proposition}[duplicated computation]\label{dupl}
  Suppose that $\Gamma \vdash e:\ety{\tau}{\eff}$ and suppose
that
  $\reads(\eff)\cap\writes(\eff)=\allocs(\eff)=\emptyset$. Thus, $e$
  reads and writes on disjoint portions of the store and makes no
  allocations. The the terms  $e_1$ and $e_2$ below
\[\begin{array}{l}
\letin{x}{e}{(x,x)} ~\textrm{ and }~ \letin{x}{e} \letin{y}{e}{(x,y)}
\end{array}
\]
 are contextually equivalent. That is
formally $\sem{\Gamma \vdash e_1 : \ety{\tau \times \tau}{\eff}} \sim
\sem{\Gamma \vdash e_2 : \ety{\tau \times \tau}{\eff}}$.
\end{proposition}

\begin{proof}
Assume a world \w and a context $\gamma \in \sem{\Gamma}\w$.
Let $\cval = \sem{\Gamma \vdash e : \ety{\tau}{\eff}}$. 
It is enough to assume a pullback 
$\w \sq{1}{1}{1}{1} \w$, and an abstract heap $\sigma_0 \in
\Astores \w$. 
From Lemma~\ref{lemmasix} and since these functions do not allocate, we
can assume that they do not cause any world extension and are therefore
defined as follows:
\[
 \begin{array}{lcl}
 \cval(\w)(\gamma)\sigma_0  =  (\w, 1, \sigma_1, \val_1) &
\textrm{and}&
  \cval(\w)(\gamma)\sigma_1  =  (\w, 1, \sigma_2, \val_2).
 \end{array}
\]
We need to show that the values $\val_1$ and $\val_2$ are equivalent and 
the heaps $\sigma_1$, obtained by applying once $e$, and $\sigma_2$,
obtained by applying twice $e$, are also equal. 

Decompose $\w = \w_0 \uplus \w_r \uplus \w_w $, where $\w_r$ contains all
the regions read by $e$ and $\w_w$ all the regions written by $e$. This is
possible because of the disjointness of of $e$'s read and write effects. 
From Lemma~\ref{lemmasix} and the disjointness of  $e$'s read and write
effects, we have that $\sigma_0$ and $\sigma_1$ agree on the regions read
by $e$, that is, $\sigma_0 \sim_{\rdsin{\eff, \w}} \sigma_1$. Hence, again
from Lemma~\ref{lemmasix}, we have that the values $\val_1$ and $\val_2$
are equal. Moreover, the locations in $\w_w$ are equaly written, while the
locations in $\w_0 \uplus \w_r$ are left unchanged, that is,
$\sigma_1$ and $\sigma_2$ agree on the location in $\w$.
\end{proof}

\begin{proposition}[pure lambda hoist]\label{hoist}
Suppose that $
\Gamma\vdash e: \ety{Z}{\emptyset}$ and 
$\Gamma,x{:}X,y{:}Z\vdash e': \ety{Y}{\eff}$
Let $e_1$ and $e_2$ be respectively $
\lambda x.\letin{y}{e} {e'}$ and
$\letin{y}{e}{\lambda x.e'}$. 
Then $\sem{\Gamma\vdash  e_1 :
\ety{(\valty{X\effto\eff
Y})}{\emptyset}} \sim 
\sem{\Gamma\vdash  e_2 : \ety{(\valty{X\effto\eff
Y})}{\emptyset}}$.
\end{proposition}

\begin{proof}
Assume a world \w and a context $\gamma \in \sem{\Gamma}\w$.
Let $\cval = \sem{\Gamma \vdash e : \ety{\tau}{\eff}}$ and 
$\cval' = \sem{\Gamma, x: X, y: Z \vdash e' : \ety{\tau}{\eff}}$. 
It is enough to assume a pullback 
$\w \sq{1}{1}{1}{1} \w$, and an abstract heap $\sigma_0 \in
\Astores \w$. Since $e$ has no effects, we have no world extension:
\[
 \cval(\w)(\gamma)\sigma_0  =  (\w, 1, \sigma_1', \val_1') 
\]
Moreover, from Lemma~\ref{lemmasix}, $\sigma_1$ and $\sigma_0$ agree on
all locations. We now show that 
\[
\sem{\Gamma\vdash  \lambda x.\letin{y}{e} {e'(x,y)} : (\valty{X\effto\eff
Y})}
\sim
\sem{\Gamma\vdash \lambda x.e'(x,\val_1') : (\valty{X\effto\eff
Y})}
\]
In order to prove this, assume a morphism $v : \w \to \w_1$ and $a \in
\sem{X}\w_1$. We need then to prove that the computations resulting from 
applying $a$ to the functions above are equivalent in the pullback
$\w_1\sq{1}{1}{1}{1}\w_1$. For this, assume an abstract heap $\sigma \in
\Astores \w_1$. Since $e$ has no effect, we have no world extension:
\[
\begin{array}{l}
 \cval(\w_1)(\gamma)\sigma  =  (\w_1, 1, \sigma_1, \val_1) \\
 \cval'(\w_1)(\gamma, a, \val_1 )\sigma_1  =  (\w_2, 1, \sigma_2, \val_2)\\
\cval'(\w_1)(\gamma, a, \val_1' )\sigma  =  (\w_2', 1, \sigma_1', \val_2)
\end{array}
\]
Since $e$ is pure, we have $\val_1 = v.\val_1'$ and from
Lemma~\ref{lemmasix} we have that $\sigma_1$ and $\sigma$ agree on all
locations in $\w_1$ and in particular on locations read by $e'$. Hence, 
again by Lemma~\ref{lemmasix} the pullback proof exists where $\sigma_2$
and 
$\sigma_1'$ are equal in its low point and the resulting values are equal
in its apex.
\end{proof}

\paragraph{Masking}
We now justify soundness of the masking rule shown below:
\[
 \infer[\textrm{Masking}]{\Gamma \vdash t :
\ety{\tau}{\eff\setminus\{\rEff{\regid},\wEff{\regid}, \aEff{\regid}\}}}
{\Gamma \vdash t : \ety{\tau}{\eff} 
\quad \regid \notin \regs{\Gamma} \cup \regs{\tau}}
\]
which allows one to mask effects, that is, allowing it to behave closer to
pure functions. As discussed in \cite{DBLP:conf/ppdp/BentonKBH07}, as the
effect-dependent equations can be applied only if some conditions on the
set of effects is satisfied, the masking of effects may enable the use of
such equations. (See the commutation computation equation.)

Assume that for 
for every set of regions $R$, we take a different instantiation $\world_R$
where all abstract locations get colors from $R$. Within $\world_R$ we can
interpret app, lambda, fix, etc.
If $R\subseteq R'$ and $X$ is a semantic type over $\world_{R'}$ denote $X
| R$ its restriction to $\world_R$. 
In our setting, we prove of the soundness of the
masking rule by providing morphisms between the objects in
$\world_R$ and objects in $\world_{R'}$ when restricted to $R$, where $R
\subseteq R'$. This corresponds in our
setting to the Masking Lemma in \cite{DBLP:conf/ppdp/BentonKBH07} and is
formalized by introducing the notion of matching pairs:
 Let $X$ be a semantic
type over $\world_R$ and $X'$ be a semantic type over
$\world_{R'}$. The two form a \emph{matching pair} if there are morphisms
$i: X \to X'|R$ and $j: X'|R \to X$ both tracked by the identity on the
level of values and isomorphisms w.r.t. $\sim$. The idea is that if $\tau$
only mentions regions in $R$ then $\sem{\tau}$ with respect to $R$ and
$\sem{\tau}$ with respect to $R'$ will be a matching pair.


Suppose that $\w\in \world_R$. 
If $\sigma \in \Astores \w$ then, since $\w$ can be viewed also over $R'$,
we can understand $\sigma$ as living in $\world_{R'}$. Conversely, if $\w
\in \world_{R'}$ and $\sigma\in \Astores \w$, then we also have 
$\sigma \in \Astores \w|R$ by coarsening. This is because if $\sigma$
satisfies all the contracts in the larger worlds involving the regions
$R'$, then it also satisifies the contracts for the regions in the smaller
set $R$. In fact, every world $\w \in \world_{R'}$ induces a world $\w|R
\in \world_{R}$.

We now prove that if only regions from $R$ are mentioned in $\tau$ then
$\sem{\tau}R$ and $\sem{\tau}R'$ form a matching pair where $\sem{\cdot}R$
denotes the interpretation with respect to $\world_R$: Suppose that $\eff$
mentions all of $R'$ and that $(\Gamma,\Gamma'), (A,A')$
are matching pairs and that $e:\Gamma' \to T_\eff A'$ is a morphism tracked
by $f:\Values \to \Comps$. There then exists a morphism $mask(e) : \Gamma
\to T_{\eff|R} A$  also tracked by $f$ and if $e\sim e'$ then $mask(e)
\sim mask(e')$.

Let the morphisms $i_\Gamma$ and $j_\Gamma$ due to the fact that
$(\Gamma, \Gamma')$ form a matching pair and $i_A$ and $j_A$ due to the
fact that $(A, A')$ form a matching pair. It is then easy to prove the
soundness of masking by using the morphism
$mask(e)\w(\gamma)(\sigma) =
\letin{(\sigma_1,v)}{e(i_\Gamma(\gamma))(\sigma)} (\sigma_1, j_A(v))$.

\iffull
\begin{figure*}[t]
\begin{center}
\begin{small}
\begin{tabular}{l@{\quad}l@{\quad}l@{\quad}l}
Body of Loop & Prolog & Steady Program & Epilogue\\
\texttt{x := load(p);} & \texttt{p1 := p;} &
\texttt{store(p1, y);}
\textbf{[$\wEff{\regid_1}$]} & \texttt{store(p1, y);}
\textbf{[$\wEff{\regid_1}$]}\\ 

\texttt{y := x * c;} & \texttt{p2 := p;} & \texttt{p1 := p2 + 8;} &
\texttt{y := x2 * c;} \\

\texttt{store(p, y);} & \texttt{x1 := x;} & \texttt{y := x2 * c;}
& \texttt{store(p2, y);} \textbf{[$\wEff{\regid_2}$]}\\

\texttt{p := p + 8;} & \texttt{x2 := x;} & \texttt{x1 := load(p1);
\textbf{[$\rEff{\regid_1}$]}} & \texttt{x := x2;}\\

\texttt{i := i + 1;} & \texttt{x1 := load(p1);}
\textbf{[$\rEff{\regid_1}$]} &
\texttt{store(p2, y);} \textbf{[$\wEff{\regid_2}$]}
& \texttt{p := p2;}\\

& \texttt{p2 := p1 + 8;} & \texttt{p2 = p1 + 8;}\\
& \texttt{x2 := load(p2);} \textbf{[$\rEff{\regid_2}$]} &
\texttt{y
= x1 * c;}\\

& \texttt{y := x1 * c;} & \texttt{y = load(p2);
\textbf{[$\rEff{\regid_2}$]}}\\

& \texttt{i := i + 2;} & \texttt{i := i + 2;}\\
\end{tabular}
\end{small}
\vspace{-2mm}
\end{center}
\caption{Program obtained from the loop unrolling technique. Here
\texttt{p}, \texttt{p1} and \texttt{p2} are pointers and all \texttt{load}
and \texttt{store} operations are on 64 bit numbers (float).}
\label{fig:loop-unrolling}
\vspace{-4mm}
\end{figure*}

\paragraph{Example: Loop Unrolling}
Loop unrolling is a software pipelining technique used to enhance the use
of parallel processing. The idea is instead of iterating a loop 
in a sequential manner, one attempts to process a number of iterations
of the loop at the same time using multiple processors.

As described in \cite{tristan10popl} implementing and proving the
correctness of loop unrolling techniques is hard as
one needs to demonstrate that the program resulting from loop unrolling
that can be executed in parallel is equivalent to the original sequential
program. We briefly illustrate the power of our system with regions and
effects by one of the running examples in \cite{tristan10popl}. Consider
a loop program whose body is depicted in Figure~\ref{fig:loop-unrolling}.
Intuitively, this program is multiplying all the elements of an array of 
float values by the value \texttt{c}. Clearly, instead of executing this
program sequentially, we can execute different iterations in parallel. 
In particular, after applying the loop unrolling optimization to a program,
one obtains a program that is divided in three parts: the prolog, that
initializes all the variables, the steady state, that is iterated, and the
epilogue, that is executed when the loop condition is no longer true and
the loop is over. Figure~\ref{fig:loop-unrolling} contains the program
obtained by loop unrolling two iterations of the program above. The Prolog
and the Epilogue are executed at the beginning and the end, respectively,
while the Steady Program may be executed several times. 

The task is to show that the optimized program is equivalent to the
sequential program above. Using the unrolling equations from
Lemmas~\ref{fixlse} we can unroll the loop twice
($n=2$) and extract a prologue. We can then conclude with
effect-dependent equivalences, in particular Prop.~\ref{compr} as
follows. We use two regions $\regid_1$
and $\regid_2$. All even elements of the array, that is, \texttt{p},
\texttt{p + 16}, \texttt{p + 32}, \ldots, belong to the region $\regid_1$,
while all odd elements, that is, \texttt{p + 8},
\texttt{p + 24}, \texttt{p + 40}, \ldots, belong to the region $\regid_2$.
Given this setting, the read and write effects are as shown in
Figure~\ref{fig:loop-unrolling}. It is now a simple exercise to show that
any execution of the optimized program is equivalent to an execution of
the sequential program. For instance, any instruction with a read effect 
on $\regid_1$ can be permuted so that it appears immediately before the
following instruction with write effect $\regid_1$ on the same region
$\regid_1$. This is possible because the only effect between these two
instructions is a read on the other region $\regid_2$. The same is true
for permuting instructions that read on $\regid_2$. 
\fi
\end{document}